%% file: main.tex
\documentclass[12pt]{article}
\usepackage{amsfonts, amsmath, amsthm, booktabs, enumitem, float, graphicx, natbib, setspace}
\usepackage[margin=1.1in]{geometry}

\newtheorem{theorem}{Theorem}[section]
\newtheorem{proposition}[theorem]{Proposition}
\newtheorem{lemma}[theorem]{Lemma}
\newtheorem{corollary}{Corollary}[theorem]
\theoremstyle{remark}
\newtheorem{remark}[theorem]{Remark}

\DeclareMathOperator{\Beta}{Beta}
\DeclareMathOperator{\N}{N}
\DeclareMathOperator{\E}{E}
\DeclareMathOperator{\B}{B}
\DeclareMathOperator{\Var}{Var}
\DeclareMathOperator{\Cov}{Cov}

\newcommand{\IF}{\operatorname{IF}}

\begin{document}

\title{Gaussian Boundary Inference in a Hypergeometric Heavy-Tailed Family}

\author{
    Steve Lawford\thanks{ENAC (University of Toulouse), 7 avenue Edouard Belin, BP 54005, 31055, Toulouse, Cedex 4, France (email: steve.lawford@enac.fr).
    }
}

\date{}

\maketitle

\begin{abstract}
\noindent This paper develops inference for a Gaussian-nested hypergeometric family of distribution functions. The family
\[
    G_c(z) = \frac12 + z\,\frac{\Gamma(c-1/2)}{2\sqrt2\,\Gamma(c)}\,{}_1F_1\!\left(\frac12;c;-\frac{z^2}{2}\right),\quad c\ge\frac32,
\]
contains the standard normal distribution at the boundary
$c=3/2$. Away from the boundary, the density has algebraic
tail behaviour $g_c(z)\sim(c-3/2)|z|^{-3}$, so the parameter $c$ indexes a directed heavy-tailed deformation of the Gaussian law. The distribution also admits an equivalent beta-precision normal scale-mixture representation, in which the Gaussian boundary corresponds to degenerate unit precision. I establish the admissibility of the hypergeometric family, derive minimum-distance estimators, and obtain both regular interior asymptotics and nonstandard boundary asymptotics under the normal null. The standardised fit-improvement statistic converges to the mixture distribution $\tfrac12\delta_0+\tfrac12\chi_1^2$. I extend the theory to plug-in location-scale procedures, including a robust
median/IQR version. Simulations document accurate null size
and directed power against heavy-tailed alternatives. An
application to daily S\&P~500 returns, both unconditionally
and after GARCH(1,1) filtering, illustrates the empirical
implications for tail fitting and risk quantiles, and documents that GARCH filtering substantially reduces but does not eliminate the symmetric heavy-tailed departure detected by the test.
\end{abstract}


\medskip


\newpage

\section{Introduction}\label{sec:introduction}

Normality is one of the most useful reference assumptions in econometrics, statistics, and financial modelling. It provides tractable likelihoods, familiar asymptotic approximations, and a benchmark against which departures in skewness, kurtosis, and tail behaviour are measured. Yet in many empirical settings, and especially in economics and finance, the Gaussian model is best understood as a local approximation rather than a literal description of the data. Asset returns, forecast errors, regression residuals, and macro-financial shocks frequently display more mass in the shoulders and tails than the Gaussian law permits \citep{mandelbrot63, fama65, cont01, gabaix09}. Detecting such departures is important, but detection alone is not enough:  rejection of normality by an omnibus test \citep{anderson_darling52, jarque_bera87, bai_ng05, doornik_hansen08} says that something is wrong with the Gaussian approximation; it does not say what distributional direction has been found, nor what model should replace the Gaussian benchmark.

This paper develops a new directed test of normality based on a Gaussian-nested hypergeometric family of distribution functions. The family is
\[
    G_c(z) = \frac12 + z\,\frac{\Gamma(c-1/2)}{2\sqrt2\,\Gamma(c)}\,
    {}_1F_1\!\left(\frac12;c;-\frac{z^2}{2}\right),
    \quad c\ge\frac32,
\]
where ${}_1F_1$ is Kummer's confluent hypergeometric function \citep{abramowitz_stegun72, olver_etal10}. The standard normal distribution is obtained exactly at the finite boundary value
\[
    G_{3/2} = \Phi.
\]
For every $c > 3/2$, the density $g_c = G_c'$ is symmetric and has algebraic tails,
\[
    g_c(z) \sim \left(c-\frac32\right)|z|^{-3}, \quad |z|\to\infty.
\]
Thus the Gaussian law appears as a critical boundary case: at $c = 3/2$, the algebraic tail component vanishes and the density reduces to the standard normal; for $c > 3/2$, the distribution has finite mean but infinite variance. The parameter $c$ therefore indexes a specific heavy-tailed deformation of the Gaussian law.

The construction is not an arbitrary use of a special function. It proceeds from the classical identity
\[
    \Phi(z) = \frac12 + \frac{z}{\sqrt{2\pi}}\,
    {}_1F_1\!\left(\frac12;\frac32;-\frac{z^2}{2}\right).
\]
A natural hypergeometric ansatz would allow both the numerator and denominator parameters of ${}_1F_1$ to vary. However, the endpoint restrictions required of a distribution function force the numerator parameter to be $1/2$. The remaining denominator parameter $c$ is then free, and monotonicity holds on the admissible region $c \geq 3/2$. The result is a genuine one-parameter family of distribution functions, not merely a formal perturbation of the Gaussian density.

This feature distinguishes the present approach from many standard flexible approximations to normality. Gram-Charlier, Edgeworth, and seminonparametric expansions \citep{gallant_nychka87, jondeau_rockinger01} can be useful for representing deviations from a benchmark density, but they often require truncation choices, sign restrictions, or additional corrections to ensure positivity and integrability. By contrast, $G_c$ is an admissible distribution function by construction. It also differs from standard heavy-tailed parametric families. In the Student $t$ family, normality is approached as the degrees of freedom tend to infinity, and the tail index varies with the shape parameter. In the hypergeometric family, normality occurs at the finite boundary $c=3/2$, and all non-Gaussian members share the same  tail index. The parameter $c$ controls the scale of the algebraic tail rather than its index. The family is therefore best understood as a Gaussian-boundary tail-deformation model: a one-parameter perturbation of the normal law in a single, well-defined distributional direction. A companion paper in progress frees the tail index, yielding a two-parameter family that nests the present case.

The hypergeometric family admits a useful probabilistic representation. For $c>3/2$,
\[
    Z \sim G_c \quad \Longleftrightarrow \quad
    Z = \frac{\varepsilon}{\sqrt{\Lambda}},
    \quad \varepsilon\sim \N(0,1), \quad
    \Lambda\sim \operatorname{Beta}\left(1,c-\frac32\right),
\]
with $\varepsilon$ and $\Lambda$ independent. Thus $G_c$ is a beta-precision normal scale mixture. The Gaussian boundary $c=3/2$ corresponds to the degenerate case $\Lambda\equiv1$, while $c>3/2$ introduces beta-distributed precision heterogeneity. Small realisations of $\Lambda$ generate high-variance conditional Gaussian observations and hence the algebraic tails of $G_c$. This representation connects the hypergeometric family to the normal variance-mixture literature \citep{andrews_mallows74, barndorff-nielsen_etal82}, but the inference problem here is not general estimation of a normal mixture: it is the boundary test of Gaussian degeneracy within a specific hypergeometric path.

The paper is also motivated by a broader line of work on hypergeometric functions as econometric modelling tools. Hypergeometric functions have long appeared in econometrics and statistics, especially in exact finite-sample distribution theory, unit-root asymptotics, ratios of quadratic forms, and related problems where nonstandard distributions can be written in terms of special functions \citep{phillips83, abadir93, abadir95, hillier01, forchini02}. In that literature, hypergeometric functions usually arise as analytic representations of distributions or moments arising from a given model or inferential problem. A separate and more constructive line of work, initiated by \citet{abadir99}, proposed estimating hypergeometric parameters directly as a flexible but structured way to model unknown nonlinear functions, with applications to omitted nonlinearity, constructive nonlinear modelling, and density estimation in finance. That programme was developed further in \citet{lawford01} and \citet{abadir_rockinger03}.

The appeal of that programme is clear. Generalised hypergeometric functions contain many familiar functional forms as special or limiting cases; they possess analytic identities, transformations, integral representations, and asymptotic expansions; and their parameters can alter functional shape parsimoniously. They therefore offer a middle ground between rigid parametric specifications and fully nonparametric methods: flexible enough for empirical realism, yet structured enough for interpretation and computation.

However, the generality of that programme also made it difficult. Estimating unrestricted ${}_pF_q$ functions raises questions about admissibility, identification, numerical evaluation, analytic continuation, and asymptotic theory, and the resulting procedures can be computationally delicate and theoretically hard to characterise. The present paper takes a narrower and more disciplined route. Rather than using a general ${}_pF_q$ as a flexible approximator, it isolates one specific confluent hypergeometric object that is analytically tractable, is an admissible distribution function, nests the Gaussian law exactly, and yields a well-defined boundary testing problem. In this sense, the paper contributes a full theoretical treatment of one canonical hypergeometric distributional path, as a foundation for broader hypergeometric modelling.

The econometric problem is to estimate $c$ and test the Gaussian boundary. Let $Z_1,\ldots,Z_n$ be i.i.d. with empirical distribution function $\widehat F_n$. For a fixed finite measure $\nu$, define the minimum-distance criterion
\[
    Q_n(c) = \int_\mathbb{R} \{\widehat F_n(z)-G_c(z)\}^2\,d\nu(z),
    \quad c\in\mathcal C=[3/2,c_U],
\]
and let
\[
    \widehat c_n\in\arg\min_{c\in\mathcal C}Q_n(c).
\]
The Gaussian null corresponds to the boundary value $c_0=3/2$. The test statistic is the improvement in fit from estimating $c$ rather than imposing $c=c_0$:
\[
    T_n = n\,\{Q_n(c_0)-Q_n(\widehat c_n)\}.
\]
After standardisation by null constants determined by $\nu$, $\Phi$, and the hypergeometric score direction $\dot G_{c_0}$, the statistic has the nonstandard boundary limit
\[
    \widetilde T_n \Rightarrow \frac12\delta_0 + \frac12\chi_1^2,
\]
where $\delta_0$ denotes a point mass at zero. The limiting distribution reflects the fact that the Gaussian null lies on the boundary of the admissible parameter space: negative fluctuations of the unconstrained estimator would imply $c<3/2$, which is inadmissible, and are therefore projected to the boundary. The limit is consistent with the general boundary theory of \citet{chernoff54}, \citet{self_liang87} and \citet{andrews01}; the contribution is to derive it rigorously for a hypergeometric minimum-distance criterion in which the Gaussian null arises as an exact finite boundary rather than an asymptotic limit.

The paper develops the corresponding minimum-distance theory. First, I prove consistency of $\widehat{c}_n$, under correct specification and under misspecification, where the limit is the pseudo-true hypergeometric projection of the underlying distribution. Second, I derive asymptotic normality when the pseudo-true value is interior. Third, I derive the boundary distribution under the Gaussian null and the associated fit-improvement statistic. Fourth, I establish directed consistency: the test has power tending to one against fixed alternatives for which the best-fitting hypergeometric distribution improves strictly on the Gaussian boundary. Finally, I derive local asymptotic power under $\sqrt{n}$-local alternatives. The local power index is
\[
    \delta_\nu(\eta) =
    \frac{2\int_\mathbb{R} \eta(z)\,\dot G_{c_0}(z)\,d\nu(z)}{\sqrt{\Omega_0(\nu)}},
\]
where $\eta$ is the local departure from normality and $\Omega_0(\nu)$ is the null asymptotic variance of the score functional. This expression makes explicit that the hypergeometric test is directed rather than omnibus: it is locally sensitive to departures that project positively onto the hypergeometric score direction $\dot G_{c_0}$.

For empirical use, I develop plug-in location-scale versions of the test. The first uses the sample mean and standard deviation: the null theory follows by replacing the Brownian bridge covariance kernel with the projected empirical-process kernel induced by estimating location and scale. The second, and operationally more useful, version uses the sample median and normal-calibrated IQR scale: this robust plug-in statistic is motivated by the fact that the hypergeometric alternatives have infinite variance for $c>3/2$, so standard-deviation standardisation can be unstable under the alternatives of greatest interest. Median/IQR standardisation preserves tail-shape information more effectively and has the same boundary-mixture null limit, with the plug-in variance constant replaced by the corresponding robust projected-process variance.

The choice of $\nu$ is part of the definition of the minimum-distance criterion. The paper focuses on a pre-specified tail-sensitive measure
\[
    d\nu_{1/2}(z) \propto [\Phi(z)\{1-\Phi(z)\}]^{-1/2}\,d\Phi(z).
\]
This choice increases the relative importance of the tails while keeping $\nu$ finite and symmetric. The local power calculation shows that one could, in principle, choose $\nu$ to maximise power against a specified direction, but the optimal weighting would depend on the unknown alternative and could place excessive weight in regions where the empirical distribution is noisy. Instead, the paper fixes a transparent tail-weighted criterion and studies its finite-sample behaviour.

The contribution of this paper is not a more powerful normality test: omnibus procedures such as Jarque--Bera, Anderson--Darling, and Shapiro--Wilk often have higher power against classical non-Gaussian alternatives. The contribution is an inferential architecture: admissibility, estimation, boundary asymptotics, plug-in implementation, and finite-sample calibration, for a specific hypergeometric distributional model in which the Gaussian law arises as an exact finite boundary. The normality testing application is the natural vehicle for demonstrating that architecture; the fitted parameter $\widehat c$ and its tail-risk implications are what the framework delivers. The simulations support the theory. Under the Gaussian null, the fixed-scale statistic has accurate finite-sample size from $n=25$ onwards, validating the boundary mixture approximation in small samples. The plug-in statistics are well calibrated at moderate $n$; the robust median/IQR plug-in is mildly liberal in small samples, with the distortion decreasing with $n$. The mean/SD plug-in has correct null size but essentially zero power against hypergeometric alternatives with $c>3/2$, because sample-standard-deviation standardisation is not regular under infinite-variance alternatives and destroys the tail signal; this validates the recommendation of the robust version. Joint estimation of the shape parameter together with unknown location and scale is severely size-distorted: the exact-profiled statistic has size that increases markedly with $n$ over the simulated range, reaching $31\%$ at nominal $10\%$ when $n=1600$, confirming the necessity of the plug-in approach.

Power simulations confirm that the fixed-scale statistic is highly powerful against hypergeometric alternatives in the designed direction, and that the robust plug-in has good power against hypergeometric and symmetric heavy-tailed alternatives, including Student $t$ and Laplace distributions, though with a cost in power relative to the fixed-scale test against hypergeometric alternatives, due to the projection loss from median/IQR standardisation. The fixed-scale test has essentially zero power against non-hypergeometric alternatives, including the Student $t_3$, reflecting the directed nature of the test: despite being heavy-tailed, the $t_3$ departure has negligible positive projection onto the hypergeometric score direction; its survival-tail index is 3, rather than the index 2 shared by the hypergeometric alternatives. As expected from the local power theory, both hypergeometric procedures have negligible power against skewed alternatives; this is a consequence of the symmetry of the hypergeometric family and not a defect. Omnibus tests such as Jarque--Bera, Anderson--Darling, and Shapiro--Wilk often have higher power against classical non-Gaussian alternatives, but they provide neither a fitted heavy-tail direction nor an interpretable hypergeometric parameter estimate with direct tail-risk implications.

The empirical illustration applies the robust plug-in hypergeometric statistic to daily S\&P~500 returns, both unconditionally and after GARCH(1,1) filtering. In both exercises the Gaussian boundary is rejected with $\widehat c>3/2$; GARCH filtering substantially reduces the estimated
tail departure but does not eliminate it. For raw returns the Student $t$ and Laplace fit better overall, while for GARCH-filtered residuals the hypergeometric ranks second by log score. In both cases the fitted $\widehat c$ implies materially more extreme far-tail quantiles than the
Student $t$, despite the latter's better overall fit. The hypergeometric model is not proposed as a universal replacement for conventional heavy-tailed distributions but as a constructive diagnostic: it asks whether the departure from normality is consistent with a particular Gaussian-boundary algebraic-tail deformation, and what tail-risk implications that deformation carries.

The paper contributes to several literatures. It contributes to normality and goodness-of-fit testing \citep{shapiro_wilk65, jarque_bera87, bai_ng05, doornik_hansen08}: unlike omnibus tests, rejection of the hypergeometric statistic has a direction and an associated fitted distribution. It connects to empirical-process and minimum-distance inference \citep{shorack_wellner86, bai03}, where the weak convergence arguments used here are standard tools applied in a nonstandard boundary setting. It relates to boundary inference and one-sided testing, where limiting distributions involve projections and mixtures of chi-squared laws \citep{chernoff54, self_liang87, andrews01}. It builds on flexible density modelling and seminonparametric methods \citep{gallant_nychka87, gallant_tauchen89, jondeau_rockinger01}, from which it differs by working with an admissible distribution function rather than a truncated expansion. It contributes to the use of hypergeometric functions in econometrics, including exact distribution theory, unit-root asymptotics, and direct estimation of hypergeometric parameters \citep{phillips83, abadir93, abadir95, abadir99, lawford01, abadir_rockinger03}, and to hypergeometric distribution modelling \citep{gordy98}. Finally, it connects to heavy-tailed financial return modelling \citep{mandelbrot63, fama65, bollerslev87, nelson91, mcneil_etal15}, contributing a directed test and fitted distribution rather than a model selected by information criteria.

The paper provides a full theoretical treatment of a hypergeometric distribution family with an exact Gaussian boundary, minimum-distance estimation, nonstandard boundary inference, and a robust plug-in implementation with documented finite-sample behaviour. The resulting test complements existing normality tests rather than replacing them. It is not designed to detect all departures from normality, but to detect, estimate, and interpret a specific heavy-tailed deformation of the Gaussian law; and to do so with a fitted parameter that carries direct tail-risk meaning. More broadly, the paper suggests that hypergeometric econometrics need not begin with a fully general ${}_{p}F_q$ specification. One productive route is to identify special-function structures that satisfy the shape restrictions required by the econometric object of interest, and to develop inference for the resulting low-dimensional family. The present family is one such object. Future work can extend the foundation to varying tail-index and asymmetric hypergeometric families, generated residuals from volatility models, multivariate settings, likelihood-based estimation, and broader classes of hypergeometric distributional deformations.

The rest of the paper is organised as follows. Section~\ref{sec:hgm_family} introduces the hypergeometric family and its main distributional properties. Section~\ref{sec:inference} defines the minimum-distance criterion, states the regularity conditions, and derives consistency, interior asymptotics, boundary inference, and local asymptotic power. Section~\ref{sec:plugin} develops plug-in location-scale versions and their asymptotic theory, including the robust median/IQR statistic. Section~\ref{sec:simulations} reports size and power simulations. Section~\ref{sec:application} presents the S\&P~500 empirical illustration. Section~\ref{sec:conclusion} concludes and discusses extensions. Proofs are in Appendix~\ref{app:proofs}. Simulation implementation details, including the full quadrature specification, null constants, and power design, are in Appendix~\ref{app:simulations}. Additional empirical tables and figures are in Appendix~\ref{app:empirical}. Appendix~\ref{app:crossover} derives the tail dominance and Value-at-Risk (VaR) implications of a departure from the Gaussian boundary. Appendix~\ref{app:joint} provides a short self-contained theoretical treatment of joint location-scale estimation as a reference against which the plug-in procedures of Section~\ref{sec:plugin} are compared. The logical dependencies among the main results are given in Appendix~\ref{app:dependencies}.

\section{The hypergeometric family}\label{sec:hgm_family}

This section introduces the hypergeometric family $\{G_c:c \geq 3/2\}$, establishes that it consists of admissible distribution functions, and derives its main distributional properties. The starting point is a classical identity that expresses the standard normal distribution function in terms of Kummer's confluent hypergeometric function ${}_1F_1$. Generalising the denominator parameter of that representation from $3/2$ to a free parameter $c \geq 3/2$ generates a one-parameter family of symmetric distribution functions that contains the Gaussian as an exact boundary case. The properties of the family (admissibility, tail behaviour, moment existence, the phase transition at $c=3/2$, and the beta-precision scale-mixture representation) are derived from the endpoint constraints of Lemma~\ref{lem:endpoint} and standard properties of the ${}_1F_1$ function, and are stated in the lemmas, proposition, and corollaries below; unimodality and the characteristic function are established as consequences of the mixture representation; the location-scale extension follows from the basic distributional properties of the family. Readers unfamiliar with hypergeometric functions will find the necessary background in \citet{abramowitz_stegun72}, \citet{abadir99}, and \citet{olver_etal10}; each property of the ${}_1F_1$ function used in the proofs is stated explicitly when first needed.

\begin{lemma}[Confluent hypergeometric representation of $\Phi$]\label{lem:normalcdf}
    For all \(z\in\mathbb R\),
    \begin{equation}\label{eq:normalcdf}
        \Phi(z) = \frac12 + \frac{z}{\sqrt{2\pi}}\,
        {}_1F_1\!\left(\frac12;\frac32;-\frac{z^2}{2}\right).
    \end{equation}
\end{lemma}

\noindent Equation~\eqref{eq:normalcdf}
is the constructive starting point of the paper. Replacing the denominator
parameter $3/2$ by a free parameter $c\geq 3/2$ generates a candidate
one-parameter family. The next result shows that the endpoint restrictions
required of a distribution function force the numerator parameter to equal
$1/2$ and determine the normalising constant uniquely.

\begin{lemma}[Endpoint constraints on the hypergeometric ansatz]
\label{lem:endpoint}
Let $a,c\in\mathbb{R}$ with $c\notin\mathbb{Z}_{0,-}$ and
$c-a\notin\mathbb{Z}_{0,-}$, and let $\kappa(a,c)$ be a real constant. Define
\[
    G_{a,c}(z) = \frac12 + z\,\kappa(a,c)\,{}_1F_1\!\left(a;c;-\frac{z^2}{2}\right), \quad z\in\mathbb{R}.
\]
If $G_{a,c}$ satisfies the endpoint conditions
\[
    \lim_{z\to-\infty}G_{a,c}(z)=0,
    \quad
    \lim_{z\to+\infty}G_{a,c}(z)=1,
\]
then $a=1/2$, and the normalising constant is determined uniquely by
\begin{equation}\label{eq:Khalf}
    \kappa\!\left(\tfrac12,c\right) = \frac{\Gamma(c-1/2)}{2\sqrt{2}\,\Gamma(c)}.
\end{equation}
\end{lemma}

\begin{remark}
At $c=3/2$, using $\Gamma(3/2)=\sqrt{\pi}/2$,
\[
    \kappa\!\left(\tfrac12,\tfrac32\right) = \frac{\Gamma(1)}{2\sqrt{2}\,\Gamma(3/2)} = \frac{1}{\sqrt{2\pi}},
\]
which matches the constant in Lemma~\ref{lem:normalcdf}. Hence
$G_{1/2,3/2}=\Phi$, confirming that the normal distribution is
recovered exactly at the boundary.
\end{remark}

\noindent Having fixed $a=1/2$, we henceforth write
\begin{equation}\label{eq:Gc_def}
    G_c := G_{1/2,c}, \quad \kappa(c) := \kappa\!\left(\tfrac12,c\right) = \frac{\Gamma(c-1/2)}{2\sqrt{2}\,\Gamma(c)},
    \quad c\geq\frac32.
\end{equation}

Lemma~\ref{lem:endpoint} fixes $a=1/2$ and determines the normalising
constant $\kappa(c)$. It remains to verify that $G_c$ is non-decreasing,
so that $g_c = G_c' \geq 0$ and $G_c$ is a genuine distribution
function. This is the content of the next lemma.

\begin{lemma}[Monotonicity of $G_c$]\label{lem:monotonicity}
For every $c\geq 3/2$, the function $G_c$ defined in \eqref{eq:Gc_def} satisfies $G_c'(z)\geq 0$ for all $z\in\mathbb{R}$.
\end{lemma}

\noindent The proof in fact establishes strict positivity:
$G_c'(z)>0$ for all $z\in\mathbb{R}$ and all $c\geq\tfrac32$. Combining Lemmas~\ref{lem:endpoint} and~\ref{lem:monotonicity} yields two immediate corollaries.

\begin{corollary}[$G_c$ is a distribution function]\label{cor:Gc_distn}
For every $c\geq 3/2$, $G_c$ defined in \eqref{eq:Gc_def} is a
distribution function on $\mathbb{R}$.
\end{corollary}

\begin{corollary}[Density of $G_c$]\label{cor:density}
For every $c\geq 3/2$, $G_c$ is absolutely continuous with density
\begin{equation}\label{eq:gc}
    g_c(z) = \kappa(c)\,{}_1F_1\!\left(\frac32;c;-\frac{z^2}{2}\right).
\end{equation}
The density $g_c$ is even, and $g_{3/2}=\phi$.
\end{corollary}

\noindent The family $\{G_c: c\geq 3/2\}$ is therefore a well-defined one-parameter family of distribution functions indexed by $c$. At the boundary $c=3/2$, the density reduces to the standard normal: $g_{3/2}=\phi$. The family varies continuously and smoothly in its parameter, as the next two lemmas show.

\begin{lemma}[Continuity of $G_c$ in $c$]\label{lem:continuity}
For each fixed $z\in\mathbb{R}$, the map $c\mapsto G_c(z)$ is
continuous on $[3/2,\infty)$. Moreover, $G_c\to\Phi$ uniformly
in $z$ as $c\downarrow 3/2$.
\end{lemma}

\begin{lemma}[Differentiability of $G_c$ in $c$]\label{lem:differentiability}
For each fixed $z\in\mathbb{R}$, the map $c\mapsto G_c(z)$ is infinitely differentiable on $(3/2,\infty)$. In particular, the first derivative is
\begin{equation}\label{eq:Gdot}
    \dot G_c(z) := \frac{\partial G_c(z)}{\partial c}
    = z\left[\dot \kappa(c)\,{}_1F_1\!\left(\frac12;c;-\frac{z^2}{2}\right)
    + \kappa(c)\,\frac{\partial}{\partial c}\,{}_1F_1\!\left(\frac12;c;-\frac{z^2}{2}\right)
    \right],
\end{equation}
where
\begin{equation}\label{eq:Kdot}
    \dot \kappa(c) = \frac{\partial \kappa(c)}{\partial c}
    = \kappa(c)\left[\psi(c-\tfrac12)-\psi(c)\right],
\end{equation}
$\psi=\Gamma'/\Gamma$ is the digamma function, and
\begin{equation}\label{eq:1F1dot}
    \frac{\partial}{\partial c}\,{}_1F_1\!\left(\frac12;c;-\frac{z^2}{2}\right)
    = -\sum_{k=1}^{\infty}
    \frac{(1/2)_k}{(c)_k\,k!}\left(-\frac{z^2}{2}\right)^k\,
    \left[\psi(c+k)-\psi(c)\right].
\end{equation}
\end{lemma}

\noindent The admissibility of $G_c$ shown above does not yet reveal how the family departs from the Gaussian. The next lemma makes this precise by characterising the tail behaviour of $g_c$.

\begin{lemma}[Tail behaviour of $g_c$]
\label{lem:tail}
For $c\geq 3/2$, as $|z|\to\infty$:
\begin{enumerate}[label=\textup{(\roman*)}]
    \item if $c>3/2$, the density has an algebraic tail,
    \begin{equation}\label{eq:tail_algebraic}
        g_c(z) = \left(c-\frac32\right)|z|^{-3} + o(|z|^{-3});
    \end{equation}
    \item if $c=3/2$, then $g_{3/2}=\phi$
exactly, with super-exponential tail decay
$g_{3/2}(z)=\frac{1}{\sqrt{2\pi}}\,e^{-z^2/2}
\to0$ faster than any power of $|z|^{-1}$
as $|z|\to\infty$.
\end{enumerate}
\end{lemma}

\noindent The boundary $c=3/2$ therefore separates Gaussian from
algebraic tail behaviour: for $c>3/2$ the density decays as
$|z|^{-3}$, with the scale of the algebraic component governed
by $c-3/2$. This is the sense in which $c$ indexes a directed
heavy-tailed deformation of the Gaussian law. Integrating the tail of Lemma~\ref{lem:tail} gives the corresponding survival function asymptotic.

\begin{corollary}[Survival function tail]\label{cor:survival}
For $c>3/2$, as $z\to\infty$,
\begin{equation}\label{eq:survival}
    1-G_c(z) \sim \frac{c-3/2}{2}\,z^{-2}.
\end{equation}
In particular, the survival function has tail index $2$ for every $c>3/2$: the index is fixed across the family and $c$ governs only the scale of the tail.
\end{corollary}

\noindent Lemma~\ref{lem:tail} and Corollary~\ref{cor:survival} together characterise the tail structure of the family completely. Two remarks place this structure in context.

\begin{remark}[Phase transition at $c=3/2$]\label{rem:phasetransition}
For $c>3/2$ the algebraic tail coefficient $c-3/2>0$ decreases continuously to zero as $c\downarrow3/2$. At the boundary it vanishes
exactly and $g_{3/2}=\phi$ by Corollary~\ref{cor:density}: the tail switches from algebraic to super-exponential.
\end{remark}

\begin{remark}[Comparison with standard heavy-tailed families]
\label{rem:comparison}
The non-Gaussian members $\{G_c : c>3/2\}$ of the family are distinct from standard symmetric heavy-tailed families. By Corollary~\ref{cor:survival}, the survival function satisfies $1-G_c(z)\sim\frac{c-3/2}{2}z^{-2}$, so the tail index is fixed at $2$ for every $c>3/2$: the parameter $c$ governs the scale of the algebraic tail, not its index. Symmetric $\alpha$-stable laws have tail index $\alpha\in(0,2)$ in the heavy-tailed case, so no non-Gaussian stable law has tail index $2$. The Student $t_k$ family (equivalently, the Pearson Type~VII family with shape parameter $m=(k+1)/2$) has tail index $k$, so index $2$ forces $k=2$; but $G_c$ is not a rescaling of $t_2$, since the peak $g_c(0)=\kappa(c)$ and the tail coefficient $c-3/2$ vary with $c$ at different rates, whereas a scale family ties them through a single parameter: rescaling $z\mapsto z/\sigma$ multiplies the peak by $\sigma^{-1}$ and the tail coefficient by $\sigma^2$, so their product
$\kappa(c)\,(c-3/2)^{1/2}$ would be constant across a scale family but varies with $c$ in the hypergeometric family. The hypergeometric family shares with $t_2$ only the moment structure (finite mean, infinite variance), not its functional form. More broadly, the Student $t_k$ family for general $k$ is symmetric and, like the hypergeometric family, unable to capture skewness; it places the Gaussian at the asymptotic limit $k\to\infty$ rather than at a finite boundary point. The structural advantage of $\{G_c\}$ over the Student $t$ family is precisely this finite boundary: the exact location of the Gaussian at $c=3/2$ enables the $\frac12\delta_0+\frac12\chi_1^2$ boundary limit
of Theorem~\ref{thm:fit_improvement}.
\end{remark}

\noindent The fixed tail index has immediate implications for the existence of moments.

\begin{corollary}[Moment existence]\label{cor:moments}
Let $c>3/2$ and $Z\sim G_c$. For $r>0$,
\[
    \E[|Z|^r]<\infty
    \quad\Longleftrightarrow\quad
    r<2.
\]
In particular, $\E[Z]=0$ by symmetry and $\E[Z^2]=\infty$: every
non-Gaussian member of the family has zero mean and infinite variance.
\end{corollary}

\begin{remark}[Tail departure from normality]\label{rem:heavytail}
The moment structure is markedly non-Gaussian for every $c>3/2$,
regardless of how close $c$ is to the boundary: $G_c$ is
symmetric with finite mean but infinite variance, since the
algebraic tail $|z|^{-3}$ is too heavy for a second moment to
exist. The parameter $c$ indexes departure from normality in the tails.
\end{remark}

\begin{figure}[t]
    \centering
    \includegraphics[width=\textwidth]{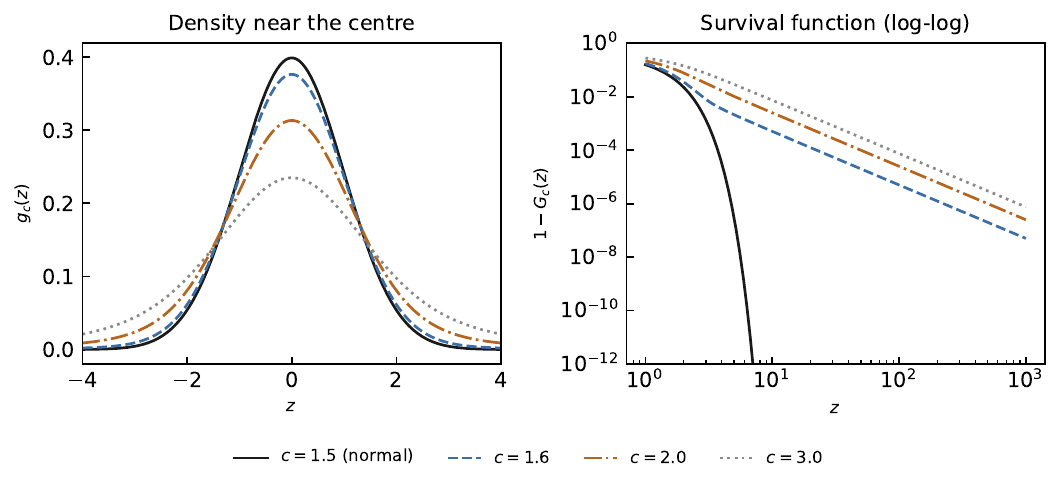}
    \caption{The hypergeometric family $G_c$ and its density $g_c$ for $c\in\{1.5,1.6,2.0,3.0\}$, with $c=3/2$ the standard normal. \emph{Left:} densities $g_c$ on linear axes, showing that the family is a smooth perturbation of the normal near the centre, with peak height decreasing in $c$. \emph{Right:} survival functions $1-G_c$ on log-log axes; for $c>3/2$ the curves are asymptotically linear with common slope $-2$, reflecting the fixed tail index in
    $1-G_c(z)\sim\tfrac{c-3/2}{2}\,z^{-2}$ for every $c>3/2$; the normal ($c=3/2$) falls away with super-exponential decay.}
    \label{fig:hgm_family}
\end{figure}

Figure~\ref{fig:hgm_family} illustrates the two features of the family established above: proximity to the normal at the centre and qualitatively distinct tails. The left panel plots the densities $g_c$: near the origin the curves are close to the standard normal, to which they reduce exactly at $c=3/2$; as $c$ increases the peak falls and mass shifts toward the tails, interpolating continuously away from normality. The right panel plots the survival functions $1-G_c$ on log-log axes, on which an algebraic tail appears as a straight line with slope equal to
the negative tail index. For every $c>3/2$ the slope is $-2$, consistent with Corollary~\ref{cor:survival}; $c$ shifts the intercept but not the slope, confirming that the tail index is fixed and $c$ governs only the tail scale. The normal case $c=3/2$ is qualitatively different: its survival function curves steeply downward, reflecting super-exponential decay, and falls below every algebraic tail beyond a moderate threshold. The figure thus displays the phase transition of
Remark~\ref{rem:phasetransition} directly.

The distributional properties established above (algebraic tails, fixed tail index, and finite mean with infinite variance) all follow from the hypergeometric density \eqref{eq:gc}, but they also admit a unified probabilistic explanation that connects the construction to a classical mixture structure and provides intuition for the role of $c$.

\begin{proposition}[Beta-precision normal scale mixture]\label{prop:mixture}
For every $c>3/2$, let $\varepsilon\sim\N(0,1)$ and
$\Lambda\sim\Beta(1,c-3/2)$ be independent. Then
$Z=\varepsilon/\sqrt{\Lambda}$ has distribution $G_c$. Equivalently,
\begin{equation}\label{eq:mixture_density}
    g_c(z) = \int_0^1 \frac{\sqrt{\lambda}}{\sqrt{2\pi}}\,
    e^{-\lambda z^2/2}\,\left(c-\tfrac32\right)(1-\lambda)^{c-5/2}\,d\lambda,
\end{equation}
the marginal density obtained by averaging $\N(0,\lambda^{-1})$
densities over $\Lambda\sim\Beta(1,c-3/2)$.
\end{proposition}

\begin{remark}[Interpretation and boundary]\label{rem:mixture}
The representation clarifies the role of $c$ at two levels.
Probabilistically, $c-3/2$ is the second shape parameter of the Beta distribution governing the latent precision $\Lambda\in(0,1)$: small values of $\Lambda$ generate high-variance conditional Gaussian observations and hence the algebraic tails of $G_c$. 

As $c \downarrow 3/2$, the $\mathrm{Beta}(1,c-3/2)$ distribution converges weakly to a point mass at one, $\Lambda \equiv 1$. We define the boundary case by this degenerate limiting distribution, so $Z\mid(\Lambda=1)\sim\mathrm{N}(0,1)$ and $G_{3/2}=\Phi$ exactly. Nonnegativity of $g_c$ is immediate from \eqref{eq:mixture_density}, since the integrand is a product of a $\N(0,\lambda^{-1})$ density and a $\Beta(1,c-3/2)$ density, both non-negative. The representation connects the hypergeometric family to the normal variance-mixture literature
\citep{andrews_mallows74, barndorff-nielsen_etal82}, but the inference problem studied here is not the general estimation of a normal mixture: it is the boundary test of Gaussian degeneracy within this particular beta-precision hypergeometric path.
\end{remark}

\begin{corollary}[Unimodality]\label{cor:unimodal}
For every $c\geq\tfrac32$, the density $g_c$ is strictly unimodal with mode at $z=0$: $g_c(0)>g_c(z)$ for all $z\neq0$, and $g_c$ is strictly decreasing on $(0,\infty)$.
\end{corollary}

\begin{corollary}[Characteristic function]\label{cor:charfun}
Let $Z\sim G_c$ with $c>3/2$. Then
\begin{equation}\label{eq:charfun}
    \varphi_c(t) = \E[e^{itZ}] = \Gamma(c-1/2)\,\exp\!\left(-\frac{t^2}{2}\right)\,
    \Psi\!\left(c-\frac32;\,0;\,\frac{t^2}{2}\right), \quad t\in\mathbb{R},
\end{equation}
where $\Psi$ is Tricomi's confluent hypergeometric function. Equivalently,
\begin{equation}\label{eq:charfun_integral}
    \varphi_c(t) = (c-3/2)\int_0^1 \exp\!\left(-\tfrac{t^2}{2\lambda}\right)\,
    (1-\lambda)^{c-5/2}\,d\lambda, \quad t\in\mathbb{R}.
\end{equation}
As $c\downarrow 3/2$, $\varphi_c(t)\to\exp(-t^2/2)$,
the characteristic function of the standard normal.
\end{corollary}

\begin{remark}[Tricomi versus Kummer]\label{rem:charfun}
The characteristic function involves Tricomi's $\Psi$ rather than the Kummer ${}_1F_1$ that appears in $g_c$. This is not accidental: the density averages the conditional Gaussian density over the beta-precision distribution via the Euler integral, giving ${}_1F_1$; the characteristic function averages $\exp(-t^2/(2\lambda))$ over the same distribution, but the reciprocal $\lambda^{-1}$ in the exponent produces the Tricomi representation instead. Because the second parameter is zero, there is no simple ${}_1F_1$ expression for $\Psi(c-\tfrac32;0;\tfrac{t^2}{2})$; the
nearest identity \citep[eq.~13.1.7]{abramowitz_stegun72} is
$\Psi(c-\tfrac32;0;\tfrac{t^2}{2})=\tfrac{t^2}{2}\,\Psi(c-\tfrac12;2;\tfrac{t^2}{2})$. Since $\E[Z^2]=\infty$ for all $c>3/2$
by Corollary~\ref{cor:moments}, and since twice differentiability of $\varphi_c$ at the origin is equivalent to finite second moment, $\varphi_c$ is not twice differentiable at $t=0$. For every $c>3/2$, methods that rely on second moments, moment-generating functions, or cumulant expansions are unavailable. Although these methods are well defined at the Gaussian boundary $c=3/2$, they are not available on any right neighbourhood of that boundary and therefore do not provide a regular basis for inference on c. This provides a further motivation for the minimum-distance distribution-function approach of Section~\ref{sec:inference}.
\end{remark}

\begin{corollary}[Location-scale extension]\label{cor:location_scale_family}
For $\mu\in\mathbb{R}$, $\sigma>0$, and $c\geq\tfrac32$, define
$G_{\mu,\sigma,c}(x)=G_c\bigl((x-\mu)/\sigma \bigr)$. Then $G_{\mu,\sigma,c}$ is a distribution function on $\mathbb{R}$ with
density $g_{\mu,\sigma,c}(x)=\sigma^{-1} g_c\bigl((x-\mu)/\sigma\bigr)$, and
$G_{\mu,\sigma,3/2}(x)=\Phi\bigl((x-\mu)/ \sigma\bigr)$. For $c>\tfrac32$, the density is symmetric about $\mu$ with
\[
    g_{\mu,\sigma,c}(x) \sim \bigl(c-\tfrac32\bigr)\sigma^2|x-\mu|^{-3},
    \quad|x|\to\infty.
\]
If $X\sim G_{\mu,\sigma,c}$ with $c>\tfrac32$, then $\E[|X-\mu|^r]<\infty
\Leftrightarrow r<2$; in particular $\E[X]=\mu$, while the variance does
not exist.
\end{corollary}

\section{Minimum-distance inference}\label{sec:inference}

This section develops the inferential theory for the hypergeometric family. The minimum-distance criterion measures the squared $\nu$-distance between the empirical distribution function and the parametric family $\{G_c:c\in\mathcal{C}\}$, and the estimator $\widehat c_n$ minimises it over the admissible parameter space. The Gaussian null $H_0:c=3/2$ lies on the boundary of $\mathcal{C}$, which is the source of the nonstandard limiting distribution of the test statistic. Consistency, interior asymptotics, boundary inference, directed consistency, and local asymptotic power are derived in turn.

\subsection{Criterion, estimator, and test statistic}\label{subsec:criterion}

Let $Z_1,\ldots,Z_n$ be i.i.d.\ with distribution function $F_0$ and empirical distribution function $\widehat F_n$. For a fixed finite measure $\nu$ on $\mathbb{R}$, define the minimum-distance criterion
\[
    Q_n(c) = \int_{\mathbb{R}} \{\widehat F_n(z)-G_c(z)\}^2\,d\nu(z),\quad c\in\mathcal{C}=\left[\frac32,c_U\right],
\]
and its population counterpart
\[
    Q(c) = \int_{\mathbb{R}}\{F_0(z)-G_c(z)\}^2\,d\nu(z).
\]
The parameter space $\mathcal{C}=[3/2,c_U]$ is compact, with the lower bound $3/2$ dictated by the admissibility region of the family and $c_U\in(3/2,\infty)$ a fixed finite upper bound chosen large enough not to bind in applications. The choice of $\nu$ is part of the definition of the criterion. Since the parameter $c$ controls a Gaussian-to-heavy-tail deformation, a natural class of tail-sensitive measures is
\begin{equation}\label{eq:nu_lambda}
    d\nu_\lambda(z) \propto
    [\Phi(z)\{1-\Phi(z)\}]^{-\lambda}\,d\Phi(z),
    \quad 0\leq\lambda<1,
\end{equation}
which increases the relative weight on the tails as $\lambda$ increases while keeping $\nu_\lambda$ finite. The substitution $u=\Phi(z)$ transforms $\nu_\lambda$ to a $\Beta(1-\lambda,1-\lambda)$ measure on $(0,1)$ with total mass given by the beta function $\B(1-\lambda,1-\lambda)=(\Gamma(1-\lambda))^2/\Gamma(2-2\lambda)$, which is finite if and only if $\lambda<1$. The simulations and empirical illustration use  $\lambda=1/2$, giving
\[
    d\nu_{1/2}(z) \propto [\Phi(z)\{1-\Phi(z)\}]^{-1/2}\,d\Phi(z).
\]
The rationale for the $d\Phi(z)$ base measure, the tail-weighting structure, and the choice $\lambda=1/2$ are discussed in Remark~\ref{rem:nu_choice}. The minimum-distance estimator is any
\[
    \widehat c_n\in\arg\min_{c\in\mathcal{C}}Q_n(c).
\]
The Gaussian null $H_0:F_0=\Phi$ corresponds to the boundary value $c_0=3/2$. The test statistic is the improvement in fit from estimating $c$ rather than imposing $c=c_0$:
\[
    T_n = n\{Q_n(c_0)-Q_n(\widehat c_n)\}.
\]
Since $\widehat c_n\in\mathcal{C}$ and $c_0=3/2$ is the left
endpoint of $\mathcal{C}$, $T_n\geq 0$ by construction. The
asymptotic properties of $\widehat c_n$ and $T_n$ are derived in
the subsections below.

\begin{remark}[Structure of $\nu_\lambda$ and the choice $\lambda=1/2$]
\label{rem:nu_choice} Two structural features of the family $\nu_\lambda$ deserve comment: (i) \textit{Role of the $d\Phi(z)$ base.}
Using $d\Phi(z)=\phi(z)\,dz$ as the base measure rather than Lebesgue measure $dz$ is not just a normalisation convenience. On the quantile scale the weight $[u(1-u)]^{-\lambda}$ diverges at $u=0$ and $u=1$, and against Lebesgue measure this divergence is not integrable for any $\lambda>0$. The factor $\phi(z)$ in $d\Phi(z)$ provides super-exponential damping in both tails, converting the divergent weight into the integrable unnormalised $\Beta(1-\lambda,1-\lambda)$ density on $(0,1)$. It follows that $\lambda=0$ gives $d\nu_0\propto d\Phi(z)$, a measure concentrated near the centre, and Lebesgue measure $dz$ (Cram\'er--von Mises) is not a member of the $\nu_\lambda$ family for any $\lambda\in[0,1)$. (ii) \textit{Choice of $\lambda$.}
Fixing $\lambda=1/2$ delivers a finite, symmetric, tail-sensitive criterion with analytically computable null constants. In principle $\lambda$ could be selected to maximise power against a target alternative class, but the optimal choice depends on the direction of departure from the Gaussian null, which is unknown in practice; a data-adaptive procedure would require controlling the joint distribution of the selection step and the test statistic, a non-standard problem requiring bootstrap calibration. The cost of fixing $\lambda=1/2$ relative to an oracle-optimal choice is expected to be small against symmetric heavy-tailed alternatives, precisely the class for which the hypergeometric score direction $\dot G_{c_0}$ and the $\nu_{1/2}$ weighting are well-aligned; the local power analysis of
Section~\ref{subsec:localpower} makes this precise.
\end{remark}

\begin{remark}[Comparison with likelihood-based estimation]\label{rem:mle}
Maximum likelihood estimation (MLE) is feasible in principle using the explicit density in Corollary~\ref{cor:density}, or through the beta-precision mixture representation of Proposition~\ref{prop:mixture}. Its boundary theory is, however, non-regular: the model is not differentiable in quadratic mean at the Gaussian boundary, and the usual $\sqrt{n}$ likelihood expansion and regular Fisher-information theory do not apply. The minimum-distance approach avoids this problem and extends naturally to hypergeometric models where explicit densities may not be available.
\end{remark}

\subsection{Assumptions}\label{subsec:assumptions}

\begin{enumerate}
    \item[(A0)] \textit{(Family regularity.)}
    For each $c\in\mathcal{C}$, $G_c$ is a distribution
    function; for each $z\in\mathbb{R}$, the map
    $c\mapsto G_c(z)$ is continuous.

    \item[(A1)] \textit{(Compact parameter space.)}
    $\mathcal{C}=[3/2,c_U]$ with $c_U\in(3/2,\infty)$ fixed.

    \item[(A2)] \textit{(Finite measure.)}
    $\nu(\mathbb{R})<\infty$.

    \item[(A3)] \textit{(Identification.)}
    $Q$ has a unique minimiser $c^\star\in\mathcal{C}$.

    \item[(A4)] \textit{(Tail sensitivity.)}
    $\nu((M,\infty))>0$ and $\nu((-\infty,-M))>0$ for all
    $M>0$.
\end{enumerate}

\noindent Assumption (A0) is not a maintained assumption but a verified property of the hypergeometric family: $G_c$ is a distribution function for every $c\geq 3/2$ by Corollary~\ref{cor:Gc_distn}, and $c\mapsto G_c(z)$ is continuous for each fixed $z$ by Lemma~\ref{lem:continuity}. It is listed explicitly because it is the condition on the parametric family that drives both uniform convergence of $Q_n$ and continuity of $Q$, and would need to be re-verified for any extension of the family. Assumption (A1) ensures minimisers exist and the argmin theorem applies. The upper bound $c_U$ is a technical device that does not affect the asymptotic distribution under the null and is chosen large enough not to
bind in practice. Assumption (A2) is part of the definition of $\nu$ and makes the criterion well defined. It also converts the pointwise Glivenko--Cantelli bound into a uniform bound on $Q_n-Q$. Assumption (A3) is the substantive condition. Under correct specification it holds by identifiability: distinct values of $c$ yield distinct tail coefficients in Lemma~\ref{lem:tail}, hence distinct distributions. Under misspecification $c^\star$ is the pseudo-true value and (A3) is a genuine condition on the geometry of $F_0$ relative to $\{G_c\}$. In practice this is a weak regularity condition: multiple minimisers would require a non-generic alignment between $F_0$ and the hypergeometric family, analogous to the standard identification conditions in minimum-distance and GMM estimation. Assumption (A4) requires $\nu$ to assign positive mass to every tail region. It ensures the criterion is sensitive to the tail differences that distinguish members of the family; it is needed only for identification under correct specification, and is
satisfied by $\nu_{1/2}$ by construction.

\begin{remark}[Identification, misspecification, and the boundary]
\label{rem:identification}
Three features of the setup merit comment. (i)~The identifying signal for $c$ weakens as $c\downarrow 3/2$: by Lemma~\ref{lem:continuity}, $G_c\to\Phi$ uniformly, so all features that distinguish $G_c$ from the Gaussian, including the algebraic tail and the peak height $g_c(0)=\kappa(c)$, vanish continuously at the boundary. This is intrinsic to any test of normality against alternatives that approach the
null in distribution, rather than a limitation of the present
procedure; it is precisely what the boundary distribution theory
below addresses. (ii)~Under misspecification, $\widehat c_n$ converges to the pseudo-true value $c^\star$, the minimum-distance projection of $F_0$ onto $\{G_c:c\in\mathcal{C}\}$: an interpretable index of tail heaviness whether or not $F_0\in\{G_c\}$. (iii) Under the Gaussian null, $c^\star=c_0=3/2$ lies on the boundary of $\mathcal{C}$. Consistency is unaffected, but the boundary position means that standard $\sqrt{n}$-asymptotic normality fails: the limit is of one-sided type, as derived below.

Taken together, rejection of $H_0:c=3/2$ carries a direction
(departure from normality toward heavier algebraic tails, quantified by $\widehat c_n-3/2$) that omnibus tests do not provide. The cost is reduced power near the null in small samples; the procedure is most informative at larger sample sizes and on data where the departure is of the symmetric heavy-tailed kind it is designed to detect. The simulations
below examine this trade-off directly.
\end{remark}

\subsection{Consistency}\label{subsec:consistency}

The first main result establishes almost sure convergence of
$\widehat c_n$ to the pseudo-true value $c^\star$, both under
correct specification and under misspecification. Under correct specification $c^\star$ is the true parameter and $Q(c^\star)=0$; under
misspecification $c^\star$ is the closest member of the family
to $F_0$ in the minimum-distance sense and $Q(c^\star)>0$.

\begin{theorem}[Consistency]\label{thm:consistency}
Under \textup{(A0)--(A3)}, any $\widehat c_n\in\arg\min_{c\in\mathcal{C}}Q_n(c)$ satisfies
$\widehat c_n\to c^\star$ almost surely.
\end{theorem}

\begin{remark}[Pseudo-true value and misspecification]\label{rem:pseudotrue}
Under misspecification, $\widehat c_n$ consistently estimates the pseudo-true value $c^\star$: the parameter indexing the hypergeometric distribution closest to $F_0$ in the minimum-distance sense. A rejection of $H_0:c=3/2$ therefore retains meaning under misspecification: it indicates that the best hypergeometric approximation to $F_0$ is not the Gaussian boundary; and $\widehat c_n-3/2$ quantifies the degree of tail departure from normality in the minimum-distance sense.
\end{remark}

\noindent Under correct specification, \textup{(A3)} is verified by the following lemma, which uses \textup{(A4)} to establish that distinct members of the family are distinguished by the criterion $Q$.

\begin{lemma}[Identification under correct specification]
\label{lem:identifiability}
Under \textup{(A4)}, the map $c\mapsto G_c$ is identified on $\mathcal{C}$: $G_c=G_{c_0}$ $\nu$-almost everywhere if and only if $c=c_0$. Consequently, if $F_0=G_{c_0}$ for some $c_0\in\mathcal{C}$, then $c^\star=c_0$ is the unique minimiser of $Q$, so \textup{(A3)} holds.
\end{lemma}

\begin{corollary}[Consistency under correct specification]
\label{cor:consistency_correct}
Under \textup{(A0)--(A2)} and \textup{(A4)}, if $F_0=G_{c_0}$ for some $c_0\in\mathcal{C}$, then $\widehat c_n\to c_0$ almost surely. In particular, under the Gaussian null $F_0=\Phi=G_{3/2}$, $\widehat c_n\to 3/2$ almost surely.
\end{corollary}

\subsection{Interior asymptotics}\label{subsec:interior}

\noindent Theorem~\ref{thm:consistency} establishes that $\widehat c_n$ converges to $c^\star$ almost surely. The next step is to determine the rate of convergence and the limiting distribution. When $c^\star$ lies in the interior of $\mathcal{C}$, differentiability of $c\mapsto G_c(z)$
established in Lemma~\ref{lem:differentiability} justifies differentiation of $Q_n$ under the integral sign under the regularity conditions of Theorem~\ref{thm:interior} below, so $Q_n$ is locally smooth in $c$ and the standard $\sqrt{n}$-theory applies. The asymptotic distribution is normal with a sandwich variance determined by three objects: the derivative $\dot G_{c^\star}(z):=\partial G_c(z)/\partial c|_{c=c^\star}$, a function of $z\in\mathbb{R}$ given explicitly by Lemma~\ref{lem:differentiability} (I write $\dot G_{c^\star}$ suppressing $z$ hereafter); the scalar curvature $H_{c^\star}:=\partial^2 Q(c)/\partial c^2|_{c=c^\star}$; and the scalar limiting variance $\Omega_{c^\star}$ of $\sqrt{n}\,\partial Q_n(c)/\partial c|_{c=c^\star}$. These objects also drive the boundary theory and plug-in extensions below, giving a modular structure that carries through the paper. The following additional assumptions govern the local behaviour
of $Q_n$ near $c^\star$.

\begin{enumerate}
    \item[(B1)] \textit{(Smoothness.)} For each $z\in\mathbb{R}$, the map $c\mapsto G_c(z)$ is twice continuously differentiable on a neighbourhood $\mathcal{N}$ of $c^\star$.

    \item[(B2)] \textit{(Uniform domination.)} There exist
    $\nu$-integrable functions $M_1,M_2$ with
    $|\dot G_c(z)|\leq M_1(z)$ and
    $|\ddot G_c(z)|\leq M_2(z)$ for all $c\in\mathcal{N}$,
    with $\int M_1^2\,d\nu<\infty$ and
    $\int M_2\,d\nu<\infty$.

    \item[(B3)] \textit{(Local identification.)}
    $H_{c^\star}:=Q''(c^\star)>0$.
\end{enumerate}

\noindent Assumption~(B1) is a verified property of the hypergeometric family. It is listed explicitly because it is the condition required for differentiation under the integral sign in $Q_n$, and would need re-verification for extensions beyond $G_c$. By Lemma~\ref{lem:differentiability}, $c\mapsto G_c(z)$ is already known to be smooth on $(3/2,\infty)$ for each fixed $z$; this is strengthened to real-analyticity by Lemma~\ref{lem:B1}, using the analyticity of $K(c)$ and the meromorphic structure of $c\mapsto{}_1F_1(\frac12;c;-z^2/2)$. Real-analyticity is stronger than required for \textup{(B1)}, which needs only twice continuous differentiability; it is established here for completeness and for use in extensions of the framework.

\begin{lemma}[Verification of \textup{(B1)}]\label{lem:B1}
For each fixed $z\in\mathbb{R}$, the map $c\mapsto G_c(z)$ is real-analytic on $(1/2,\infty)$, hence smooth on any neighbourhood of $c^\star$ in $\mathcal{C}=[3/2,c_U]$. In particular, \textup{(B1)} holds.
\end{lemma}

\noindent Assumption (B2) is likewise verified for the present family: $\dot G_c$ and $\ddot G_c$ are uniformly bounded in $z$ and $c\in\mathcal{N}$, so constant dominators suffice. Uniform domination plays two distinct roles in the proof: square-integrability of $M_1$ controls the limiting variance $\Omega$, while integrability of $M_2$ controls the Taylor remainder and the uniform convergence of $Q_n''$ to $Q''$.

\begin{lemma}[Verification of \textup{(B2)}]\label{lem:B2}
For any compact neighbourhood $\mathcal{N}$ of $c^\star$ in $(1/2,\infty)$, there exist finite constants $C_1,C_2$ such that $|\dot G_c(z)|\leq C_1$ and $|\ddot G_c(z)|\leq C_2$ for all $z\in\mathbb{R}$ and all $c\in\mathcal{N}$. In particular, for $c^\star\geq 3/2$, setting $M_1(z)=C_1$ and $M_2(z)=C_2$, both integrability conditions of \textup{(B2)} hold under \textup{(A2)}.
\end{lemma}

\noindent Assumption~(B3) requires the population criterion to have strictly positive curvature at $c^\star$, ruling out a locally flat criterion. Under correct specification, (B3) is verified by the following lemma. Under misspecification, $H_{c^\star}$ includes the additional term
$-2\int\{F_0-G_{c^\star}\}\ddot G_{c^\star}\,d\nu$, and \textup{(B3)} is a genuine maintained assumption.

\begin{lemma}[Positivity of the score-direction integral]\label{lem:B3}
Suppose $\nu$ assigns positive mass to some open interval. Then, for any $c^\star\geq3/2$,
\[
    2\int_{\mathbb{R}}\dot G_{c^\star}(z)^2\,d\nu(z) > 0.
\]
In particular, if $F_0=G_{c^\star}$, this quantity equals $H_{c^\star}=Q''(c^\star)$, verifying \textup{(B3)} under correct specification.
\end{lemma}

\noindent The score variance $\Omega_{c^\star}$ appearing in the asymptotic normality result below is, by the same argument, non-degenerate under the same hypothesis on $\nu$.

\begin{lemma}[Positivity of $\Omega$]\label{lem:Omega_pos}
Suppose $\nu$ assigns positive mass to some open interval. Then, for any $c^\star\geq3/2$, $\Omega_{c^\star}>0$.
\end{lemma}

\noindent With assumptions \textup{(B1)--(B3)} in place, together with the non-degeneracy of $\Omega_{c^\star}$ just established, $\widehat c_n$ is asymptotically normal at the standard $\sqrt{n}$ rate, with a non-degenerate limiting variance.

\begin{theorem}[Interior asymptotic normality]\label{thm:interior}
Suppose the conditions of Theorem~\ref{thm:consistency} and \textup{(B1)--(B3)} hold, and that $c^\star\in\operatorname{int}(\mathcal{C})$. Then
\[
    \sqrt{n}\,(\widehat c_n-c^\star)\;\Rightarrow\;
    \N\!\left(0,\,\frac{\Omega_{c^\star}}{H_{c^\star}^2}
    \right),
\]
where
\begin{equation}\label{eq:Omega}
    \Omega_{c^\star} = 4\iint_{\mathbb{R}^2}
    \bigl[F_0(z\wedge y)-F_0(z)F_0(y)\bigr]
    \dot G_{c^\star}(z)\,\dot G_{c^\star}(y)
    \,d\nu(z)\,d\nu(y),
\end{equation}
and $z\wedge y:=\min(z,y)$.
\end{theorem}

\begin{remark}[Structure of the sandwich variance]\label{rem:sandwich}
The asymptotic variance $\Omega_{c^\star}/H_{c^\star}^2$ in Theorem~\ref{thm:interior} has the sandwich form standard in minimum-distance asymptotics. The numerator $\Omega_{c^\star}$, given in \eqref{eq:Omega}, is the variance of the limiting score $\sqrt{n}\,Q_n'(c^\star)$: a weighted integral of the Brownian bridge covariance $F_0(z\wedge y)-F_0(z)F_0(y)$ against the score direction $\dot G_{c^\star}$. The denominator $H_{c^\star}^2$ is the squared curvature of the population criterion at $c^\star$. Under correct specification, $F_0=G_{c^\star}$, so $\sqrt{n}(\widehat F_n-G_{c^\star})\Rightarrow\mathbb{B}_{G_{c^\star}}$ and the covariance kernel in \eqref{eq:Omega} reduces to $G_{c^\star}(z\wedge y)-G_{c^\star}(z)G_{c^\star}(y)$, so that $\Omega_{c^\star}$ becomes the variance of the $G_{c^\star}$-Brownian bridge projected onto $\dot G_{c^\star}$. For the plug-in tests of Section~\ref{sec:plugin}, the same sandwich structure holds with the Brownian bridge covariance replaced by the projected covariance kernel that accounts for estimation of location and scale; $H_{c^\star}$ and $\dot G_{c^\star}$ are unchanged.
\end{remark}

\begin{remark}[Assumption \textup{(B3)} under misspecification]\label{rem:B_misspec}
Under misspecification, (B1) and (B2) remain verified properties of $G_c$. Only (B3) changes character: the curvature $H$ involves the additional term
$-2\int_\mathbb{R}\{F_0(z)-G_{c^\star}(z)\}\ddot G_{c^\star}(z)\,d\nu(z)$, and positive curvature is no longer automatic. This is the standard second-order local identification condition in minimum-distance asymptotics under misspecification \citep{newey_mcfadden94}, maintained as an assumption.
\end{remark}

\subsection{Boundary inference}\label{subsec:boundary}

This subsection establishes the limiting distribution of $\widehat c_n$ at the Gaussian null $c_0=3/2$, which lies on the boundary of $\mathcal{C}$, where $Q_n'(\widehat c_n)=0$ need not hold, and the interior argument of Theorem~\ref{thm:interior} does not apply. Beyond the failure of the first-order condition, the one-sided constraint $c\geq c_0$ reshapes the limit itself: sample paths along which the unconstrained Gaussian fluctuation would be negative are instead redirected to $\widehat c_n=c_0$, so the limit places mass $1/2$ at zero and is half-normal on $(0,\infty)$, rather than normal. The proof uses a localisation argument. Writing $c=c_0+h/\sqrt{n}$ with $h\geq 0$, the scaled criterion $L_n(h)=n\{Q_n(c_0+h/\sqrt{n})-Q_n(c_0)\}$ converges to a quadratic limit $L(h)$ whose minimiser over $h\geq 0$ is the censored Gaussian $\max\{0,Z\}$. The key ingredients are the score limit from Donsker's theorem, as in the proof of Theorem~\ref{thm:interior}; the two properties of $Q_n''$ established in Lemma~\ref{lem:curvature}: strict positivity of $H_{c_0}$ and uniform convergence of $Q_n''$ to $Q''$ on a neighbourhood of $c_0$; and the tightness of $\sqrt{n}(\widehat c_n-c_0)$ established in Lemma~\ref{lem:tightness}.

\begin{lemma}[Curvature at the Gaussian null]
\label{lem:curvature}
Let $\mathcal N\subset(1/2,\infty)$ be any compact neighbourhood of $c_0$, as in \textup{(B1)--(B2)}. Under \textup{(B1)--(B2)}, $F_0=\Phi=G_{3/2}$, and $\nu$ assigning positive mass to some open interval,
\begin{equation}\label{eq:Hc0}
    H_{c_0} = Q''(c_0) = 2\int_{\mathbb{R}}
    \dot G_{c_0}(z)^2\,d\nu(z) > 0.
\end{equation}
Moreover,
\begin{equation}\label{eq:Qn_uniform}
    \sup_{c\in\mathcal{N}}|Q_n''(c)-Q''(c)|\;\to_p\; 0.
\end{equation}
\end{lemma}

\noindent Although $c_0=3/2$ is the left endpoint of the parameter space $\mathcal C=[3/2,c_U]$, it is interior to $(1/2,\infty)$, so $\mathcal N$ is necessarily two-sided and extends strictly below $c_0$ as well as above. This costs nothing: (B1)--(B2) are verified on all of $(1/2,\infty)$ by Lemmas~\ref{lem:B1} and~\ref{lem:B2}, and in the proofs of Theorems~\ref{thm:interior} and~\ref{thm:boundary} the supremum in \eqref{eq:Qn_uniform} is applied only at points $c\geq c_0$ within $\mathcal C$. The following lemma establishes that $\sqrt{n}(\widehat c_n-c_0)$ remains bounded in probability: the tightness condition required by the argmax continuous mapping theorem.

\begin{lemma}[Tightness of the boundary estimator]\label{lem:tightness}
Suppose the conditions of
Corollary~\ref{cor:consistency_correct} and
\textup{(B1)--(B2)} hold, that $F_0=\Phi=G_{3/2}$, and that
$\nu$ assigns positive mass to some open interval. Then
\[
    \sqrt{n}\,(\widehat c_n-c_0) = O_p(1).
\]
\end{lemma}

\noindent With both lemmas in place, the boundary limit follows from a localisation argument and the argmax continuous mapping theorem of \citet{vandervaart_wellner23}.

\begin{theorem}[Boundary limit under the Gaussian null]\label{thm:boundary}
Suppose the conditions of Corollary~\ref{cor:consistency_correct} and
\textup{(B1)--(B2)} hold, that $F_0=\Phi=G_{3/2}$, and that $\nu$ assigns positive mass to some open interval. Then
\[
    \sqrt{n}\,(\widehat c_n-c_0)\;\Rightarrow\;\max\{0,\,Z\},
    \quad Z\sim\N\!\left(0,\,\frac{\Omega_{c_0}}{H_{c_0}^2}\right),
\]
where $H_{c_0}=2\int_{\mathbb{R}}\dot G_{c_0}(z)^2\,d\nu(z)$ and
\begin{equation}\label{eq:Omega_c0}
    \Omega_{c_0} = 4\iint_{\mathbb{R}^2}
    \bigl[\Phi(z\wedge y)-\Phi(z)\Phi(y)\bigr]
    \dot G_{c_0}(z)\,\dot G_{c_0}(y)
    \,d\nu(z)\,d\nu(y).
\end{equation}
The limit places mass $1/2$ at zero and is half-normal on $(0,\infty)$.
\end{theorem}

\begin{remark}[Interior versus boundary]
Theorems~\ref{thm:interior} and~\ref{thm:boundary} describe qualitatively different regimes. When $c^\star>3/2$, the estimator is regular: $\sqrt{n}$-consistent and asymptotically normal with sandwich variance $\Omega_{c^\star}/H_{c^\star}^2$. When $c^\star=c_0=3/2$, the one-sided constraint instead produces the point-mass-plus-half-normal limit of
Theorem~\ref{thm:boundary}.
\end{remark}

\noindent This nonstandard limit has an immediate consequence for testing: normal critical values from Theorem~\ref{thm:interior} do not apply at $c_0$. I construct a test from the fit-improvement statistic
$T_n=n\{Q_n(c_0)-Q_n(\widehat c_n)\}\geq0$, the minimum-distance analogue of a likelihood-ratio statistic; Remark~\ref{rem:test_principle} discusses this choice relative to the Wald and score alternatives.

\begin{theorem}[Boundary limit of the fit-improvement
statistic] \label{thm:fit_improvement}
Under the conditions of Theorem~\ref{thm:boundary}, the fit-improvement statistic $T_n=n\{Q_n(c_0)-Q_n(\widehat c_n)\}$ satisfies
\begin{equation}\label{eq:Tn_limit}
    T_n \;\Rightarrow\;\frac12\delta_0
    + \frac12\!\left(\frac{\Omega_{c_0}}{2H_{c_0}}\chi_1^2\right),
\end{equation}
where $\delta_0$ is a point mass at zero, and the standardised statistic $\widetilde T_n=(2H_{c_0}/\Omega_{c_0})\,T_n$ satisfies
\begin{equation}\label{eq:Tntilde_limit}
    \widetilde T_n \;\Rightarrow\;
    \frac12\delta_0+\frac12\chi_1^2.
\end{equation}
\end{theorem}

\begin{remark}[Choice of test principle]\label{rem:test_principle}
Three statistics suggest themselves for testing $c=c_0$: a Wald statistic based on $\widehat c_n-c_0$, studentised by $\Omega_{c_0}^{1/2}/H_{c_0}$; a score (Lagrange multiplier) statistic based on $Q_n'(c_0)$, studentised by $\Omega_{c_0}^{1/2}$; and the standardised criterion-difference statistic $\widetilde T_n$ above. Since $H_{c_0}$ and $\Omega_{c_0}$ are known exactly under the null, all three statistics are immediately operational, with no nuisance parameters to estimate. At an interior $c^\star$, these three coincide to first order, exactly as in the classical Wald--LR--LM trinity. This equivalence fails at the boundary. The score is evaluated at the fixed point $c_0$, before any constrained optimisation occurs, so $\sqrt{n}\,Q_n'(c_0)\Rightarrow\N(0,\Omega_{c_0})$ regardless of the constraint $c\geq c_0$: an LM test retains a standard null distribution. The Wald and criterion-difference statistics, by contrast, are both built from $\widehat c_n$ and inherit its censoring, giving the nonstandard limits $\max\{0,Z\}$ and $\tfrac12\delta_0+\tfrac12\chi_1^2$ respectively: the familiar finding of the literature on testing at the boundary of the parameter space \citep{chernoff54, self_liang87, andrews01}. I use standardised $\widetilde T_n$ rather than the regular LM alternative for two reasons. First, regularity of the null distribution says nothing about power at the boundary, where the local-power equivalence of Wald, LR, and LM at interior points need not carry over. Second, $\widetilde T_n$ is built from the same minimum-distance criterion $Q_n$ used throughout the paper, whereas an LM test requires a separate, score-based apparatus. I prefer $\widetilde T_n$ to the Wald alternative because $\widetilde T_n$ is the exact difference in criterion values at $c_0$ and $\widehat c_n$, whereas the Wald statistic relies on a single local quadratic approximation to $Q_n$ anchored at $c_0$; the criterion-difference statistic therefore tracks the finite-sample shape of $Q_n$ more closely, the standard rationale for preferring likelihood-ratio-type statistics over Wald statistics more generally.
\end{remark}

\begin{remark}[Null constants and critical values]\label{rem:critical_values}
The standardised statistic $\widetilde T_n$ is operational once $\nu$ has been chosen. The null constants $H_{c_0}$ and $\Omega_{c_0}$, given in \eqref{eq:Hc0} and \eqref{eq:Omega_c0}, are determined entirely by $\Phi$, $\dot G_{c_0}$, and $\nu$; they are not unknown population parameters requiring estimation. Both, moreover, reduce to closed-form expressions in the same special functions used to define the model itself: $\dot G_{c_0}$ is given explicitly by \eqref{eq:Gdot}--\eqref{eq:1F1dot} of Lemma~\ref{lem:differentiability} in terms of the Gamma and digamma functions and the $c$-derivative of ${}_1F_1$, while $\Phi$ itself is available exactly via the confluent hypergeometric representation \eqref{eq:normalcdf} of Lemma~\ref{lem:normalcdf}, with no numerical integration of the Gaussian density required. Evaluating $H_{c_0}$ and $\Omega_{c_0}$ at any point therefore requires only ${}_1F_1$, $\Gamma$, and $\psi$: the same three special functions on which the entire hypergeometric family is built; though the integrals defining $H_{c_0}$ and $\Omega_{c_0}$ over $\nu$ are one- and two-dimensional respectively, and generally still require numerical quadrature; what is avoided is any further numerical approximation of the integrands.

For the Gaussian-weighted choice $d\nu(z)=d\Phi(z)$: the $\lambda=0$ member of the tail-weighted family $\nu_\lambda$ in \eqref{eq:nu_lambda}, rather than the $\lambda=1/2$ choice used in the simulations and empirical illustration below (see Remark~\ref{rem:nu_choice}); writing $a(u)=\dot G_{c_0} \{\Phi^{-1}(u)\}$ for $u\in(0,1)$, the constants take the form
\[
    H_{c_0} = 2\int_0^1 a(u)^2\,du, \quad \Omega_{c_0}
    = 4\int_0^1\!\int_0^1 \{\min(u,v)-uv\}\,a(u)\,a(v)\,du\,dv,
\]
both of which are evaluable by numerical quadrature. For general $\nu_\lambda$, including the tail-weighted $\nu_{1/2}$ used in this paper, the analogous integrals can be computed over $\mathbb{R}$, or, after the same substitution, over the corresponding $\Beta(1-\lambda,1-\lambda)$ measure on $(0,1)$. Since $\nu_\lambda$ in \eqref{eq:nu_lambda} is defined only up to a positive constant, $H_{c_0}$ and $\Omega_{c_0}$ individually depend on that constant for $\lambda>0$; the estimator $\widehat c_n$, the asymptotic variance $\Omega_{c_0}/H_{c_0}^2$, and the standardised statistic $\widetilde T_n$ do not. The limiting null law $\frac12\delta_0+\frac12\chi_1^2$ yields a simple rejection rule. For an asymptotic level-$\alpha$ upper-tail test,
\[
    \Pr\!\left(\tfrac12\delta_0+\tfrac12\chi_1^2>k_\alpha\right)
    = \tfrac12\Pr(\chi_1^2>k_\alpha) = \alpha,
\]
so $k_\alpha=\chi_{1,1-2\alpha}^2$, the $(1-2\alpha)$-quantile of the $\chi_1^2$ distribution. The rejection rule is
\[
    \text{reject }H_0:F_0=\Phi
    \quad\Longleftrightarrow\quad
    \widetilde T_n > \chi_{1,1-2\alpha}^2.
\]
For $\alpha=0.05$, the critical value is $\chi_{1,0.90}^2\approx 2.706$; for $\alpha=0.01$, it is $\chi_{1,0.98}^2\approx 5.412$. These are strictly smaller than the standard $\chi_1^2$ critical values ($3.841$ and $6.635$).
\end{remark}

\noindent The moments of the limiting distribution are also available in closed form.
\begin{corollary}[Moments of $\widetilde T_\infty$]
\label{cor:moments_Ttilde}
Under the conditions of Theorem~\ref{thm:fit_improvement},
$\widetilde T_n\Rightarrow\widetilde T_\infty\sim\frac12
\delta_0+\frac12\chi_1^2$, with
\[
    \E[\widetilde T_\infty] = \frac12, \quad
    \Var(\widetilde T_\infty) = \frac54.
\]
\end{corollary}

\noindent The simulated mean and variance of $\widetilde T_n$ are reported alongside the limiting values $1/2$ and $5/4$ from Corollary~\ref{cor:moments_Ttilde} as an informal diagnostic of how closely the finite-sample distribution tracks the boundary asymptotics. Theorems~\ref{thm:boundary} and~\ref{thm:fit_improvement}, together with
Corollary~\ref{cor:moments_Ttilde}, complete the boundary inference theory: $\widetilde T_n$ has a known null limit, an explicit rejection rule, and null constants computable from $\nu$ and $\dot G_{c_0}$ alone. The next two results establish directed consistency against fixed alternatives and non-trivial local asymptotic power against $\sqrt{n}$-local alternatives.

\subsection{Directed consistency}\label{subsec:directed}

Theorems~\ref{thm:boundary} and~\ref{thm:fit_improvement} establish the null distribution of $\widetilde T_n$. The following theorem shows that the test is also consistent: under any fixed alternative for which the best hypergeometric approximation strictly improves on the Gaussian boundary, the statistic diverges to infinity and the test rejects with
probability tending to one.

\begin{theorem}[Directed consistency]\label{thm:directed}
Suppose the conditions of Theorem~\ref{thm:consistency} hold,
and that $\nu$ assigns positive mass to some open interval. Let
\begin{equation}\label{eq:Delta}
    \Delta = Q(c_0)-Q(c^\star)
\end{equation}
be the population fit improvement. If $\Delta>0$, then
$Q_n(c_0)-Q_n(\widehat c_n)\to_p\Delta$, and consequently
\[
    \widetilde T_n \;=\; \left(\frac{2H_{c_0}}{\Omega_{c_0}}\right)
    n\{Q_n(c_0)-Q_n(\widehat c_n)\} \;\to_p\; \infty.
\]
Hence $\Pr_{F_0}(\widetilde T_n>k_\alpha)\to 1$ for any fixed
critical value $k_\alpha<\infty$.
\end{theorem}

\begin{remark}[Directed, not omnibus]\label{rem:directed}
Consistency is directed at alternatives for which the best hypergeometric approximation strictly improves on the Gaussian boundary, that is, $\Delta>0$: the heavy-tailed deformation encoded by $\{G_c:c>3/2\}$. It need not detect every departure from normality: under a skewed alternative the symmetric family cannot capture, the minimum-distance projection $c^\star$ can remain at $c_0=3/2$, giving $\Delta=0$ and no guaranteed power. This is precisely the sense in which the test complements rather than replaces omnibus procedures such as Jarque--Bera.
\end{remark}

\noindent Theorem~\ref{thm:directed} establishes power against fixed alternatives; the next result characterises power against local alternatives approaching the null at rate $n^{-1/2}$.

\subsection{Local asymptotic power}\label{subsec:localpower}

Theorem~\ref{thm:directed} establishes power tending to one against fixed alternatives with $\Delta>0$. The complementary local result characterises power against $\sqrt{n}$-local alternatives. The key is the projection of the local departure from normality onto the score direction $\dot G_{c_0}$. Consider a triangular array with distribution functions
\begin{equation}\label{eq:local_alt}
    F_n(z) = \Phi(z) + \frac{\eta(z)}{\sqrt{n}} + r_n(z),
\end{equation}
where $\eta$ is the local departure from normality. The score functional and its variance under the local alternative are determined by the inner product
\[
    A_\nu(\eta) = \int_{\mathbb{R}}\eta(z)\,\dot G_{c_0}(z)\,d\nu(z),
\]
and the Brownian bridge covariance kernel $K(z,y)=\Phi(z\wedge y)-\Phi(z)\Phi(y)$.

\begin{theorem}[Local asymptotic power]\label{thm:local_power}
Suppose \textup{(B1)--(B2)} hold and $\nu$ assigns positive mass to some open interval, so that Theorem~\ref{thm:fit_improvement} applies at $F_0=\Phi$. Let $F_n$ satisfy \eqref{eq:local_alt} with
\[
    \int_{\mathbb{R}}|\eta(z)\,\dot G_{c_0}(z)|\,d\nu(z)<\infty,
\]
and $\sqrt{n}(\widehat F_n(z)-\Phi(z))\Rightarrow
\mathbb{B}_\Phi(z)+\eta(z)$ in $\ell^\infty(\mathbb{R})$,
where $\mathbb{B}_\Phi$ is a $\Phi$-Brownian bridge. Then
\begin{equation}\label{eq:Tn_local}
    \widetilde T_n\;\Rightarrow\;
    (Z-\delta_\nu(\eta))^2\,\mathbf{1}\{Z-\delta_\nu(\eta)<0\},
    \quad Z\sim\N(0,1),
\end{equation}
where the local power index is
\begin{equation}\label{eq:delta}
    \delta_\nu(\eta) = \frac{2A_\nu(\eta)}{\sqrt{\Omega_{c_0(\nu)}}}.
\end{equation}
The asymptotic rejection probability of the level-$\alpha$ test is
\begin{equation}\label{eq:power}
    \pi_\nu(\eta) =
    \Phi\bigl\{\delta_\nu(\eta)-\Phi^{-1}(1-\alpha)\bigr\}.
\end{equation}
\end{theorem}

\begin{remark}[Directed local sensitivity]
\label{rem:local_power}
The functional limit $\sqrt{n}(\widehat F_n-\Phi)\Rightarrow \mathbb{B}_\Phi+\eta$ is maintained as the standard local empirical process limit for a triangular array drifting toward $\Phi$ at rate $n^{-1/2}$ \citep[cf.][Thm.~3.10.12] {vandervaart_wellner23}, rather than re-derived here; setting $\eta\equiv0$ recovers Donsker's theorem under the fixed null (Theorem~\ref{thm:fit_improvement}). The local power index
$\delta_\nu(\eta)=2A_\nu(\eta)/\sqrt{\Omega_{c_0}(\nu)}$ measures the projection of $\eta$ onto the score direction $\dot G_{c_0}$, standardised by the null standard deviation. Alternatives with positive projection have power above size; alternatives orthogonal to $\dot G_{c_0}$, including skewed departures since $\dot G_{c_0}$ is odd, have $\delta_\nu(\eta)=0$ and power equal to size asymptotically. This is the local counterpart of Remark~\ref{rem:directed}: the test is locally sensitive only to departures aligned with the hypergeometric score direction.
\end{remark}

\begin{corollary}[Local hypergeometric alternatives]\label{cor:local_hgm}
Under the local hypergeometric alternatives
$F_n(z)=G_{c_0+\tau/\sqrt{n}}(z)$ with $\tau\geq 0$,
\[
    \delta_\nu(\tau) =
    \frac{\tau H_{c_0}(\nu)}{\sqrt{\Omega_{c_0}(\nu)}},
\]
and the asymptotic rejection probability is
\[
    \pi_\nu(\tau) = \Phi\!\left\{
    \frac{\tau H_{c_0}(\nu)}{\sqrt{\Omega_{c_0}(\nu)}}
    -\Phi^{-1}(1-\alpha)\right\}.
\]
\end{corollary}

\noindent Theorems~\ref{thm:boundary}, \ref{thm:fit_improvement}, \ref{thm:directed} and \ref{thm:local_power} complete the asymptotic theory for the fixed-scale test. The next section develops plug-in location-scale versions.

\section{Plug-in location-scale procedures}\label{sec:plugin}

The fixed-scale test assumes known location and scale. In practice one tests the composite Gaussian null $H_0:F_0\in \{\N(\mu,\sigma^2):\mu\in\mathbb{R},\sigma>0\}$ by standardising with estimated location and scale before applying the fixed-scale procedure. Two standardisations are considered. The first uses the sample mean and standard deviation, which are natural under the Gaussian null. The second uses the sample median and normal-calibrated IQR, which is motivated by the infinite variance of $G_c$ for $c>3/2$, under which the sample standard deviation is an unstable scale estimator and median/IQR standardisation preserves the tail signal more effectively. In both cases, the fixed-scale theory extends immediately: consistency, interior asymptotics, boundary inference, and directed consistency all carry over with $\widehat F_n$ replaced by the standardised empirical distribution. The only structural change is the projected covariance kernel, which replaces the Brownian bridge kernel $K(z,y)=\Phi(z\wedge y)-\Phi(z)\Phi(y)$ to account for estimation of location and scale. The curvature $H_{c_0}$, score direction $\dot G_{c_0}$, null limit $\frac12\delta_0+\frac12\chi_1^2$, and rejection rule $\widetilde T_n^\diamond>\chi_{1,1-2\alpha}^2$ are all unchanged.

\subsection{Location-scale standardisation and plug-in
statistics}\label{subsec:plugin_stats}

Let $X_1,\ldots,X_n$ be the observed sample. I consider two location-scale standardisations, denoted by superscript $\diamond\in\{P,R\}$: mean/standard-deviation ($P$) and median/IQR ($R$).\\
\\
\textit{Mean/standard-deviation ($P$).} Define
\[
    \widehat\mu_n^P = \bar X_n, \quad (\widehat\sigma_n^P)^2 =
    \frac{1}{n}\sum_{i=1}^n(X_i-\bar X_n)^2,
\]
and let $Y_i^P=(X_i-\widehat\mu_n^P)/\widehat\sigma_n^P$, with empirical distribution $\widehat F_n^P$.\\
\\
\textit{Median/IQR ($R$).} Let $\widehat q_n(p)$ denote the $p$-th sample quantile and $a=\Phi^{-1}(3/4)$. Let
\begin{equation}\label{eq:median_iqr}
    \widehat\mu_n^R = \widehat q_n(1/2), \quad
    \widehat\sigma_n^R =
    \frac{\widehat q_n(3/4)-\widehat q_n(1/4)}{2a},
\end{equation}
where $2a$ is the interquartile range of $\Phi$, used to calibrate $\widehat\sigma_n^R$. Let $Y_i^R=(X_i-\widehat\mu_n^R)/\widehat\sigma_n^R$, with empirical distribution $\widehat F_n^R$. For each standardisation $\diamond\in\{P,R\}$, define the plug-in criterion, estimator, and fit-improvement statistic
\[
    Q_n^\diamond(c) = \int_{\mathbb{R}}
    \{\widehat F_n^\diamond(z)-G_c(z)\}^2\,d\nu(z),
    \quad c\in\mathcal{C},
\]
\[
    \widehat c_n^\diamond \in
    \arg\min_{c\in\mathcal{C}}Q_n^\diamond(c), \quad
    T_n^\diamond = n\{Q_n^\diamond(c_0)
    -Q_n^\diamond(\widehat c_n^\diamond)\}.
\]
\noindent Under the Gaussian null, $\widehat\mu_n^P\to\mu_0$ and $\widehat\sigma_n^P\to\sigma_0$ almost surely by the strong law of large numbers, since $\mathrm{E}[X_i^2] <\infty$; similarly $\widehat\mu_n^R\to\mu_0$ and $\widehat\sigma_n^R\to\sigma_0$ almost surely by the consistency of sample quantiles at the population quantiles
corresponding to probabilities $1/4$, $1/2$, and $3/4$, where the Gaussian
density is continuous and positive. The two conditions are not equivalent: the mean/SD case requires a finite second moment, while the median/IQR case requires only local density regularity, with no moment condition.

\subsection{Asymptotic theory}\label{subsec:plugin_asymptotics}

Under the Gaussian null, both standardisations yield a mean-zero
Gaussian empirical-process limit. The only change relative to the
fixed-scale theory is that the Brownian-bridge covariance kernel is
replaced by the projected kernel induced by the chosen location-scale
standardisation. This projection accounts for the first-order effect of estimating location and scale and leads to the same boundary limit after using the corresponding plug-in score variance. The next two lemmas give the projected empirical processes for the mean/standard-deviation and median/IQR standardisations, respectively.

\begin{lemma}[Projected empirical process: mean/SD]\label{lem:plugin_ep}
Suppose $X_i\stackrel{\text{i.i.d.}}{\sim}\N(\mu_0,\sigma_0^2)$.
Then, in $\ell^\infty(\mathbb{R})$,
\begin{equation}\label{eq:ZP}
    \sqrt{n}\{\widehat F_n^P(z)-\Phi(z)\}
    \;\Rightarrow\;
    \mathbb{Z}_P(z),
\end{equation}
where $\mathbb{Z}_P$ is the mean-zero Gaussian process admitting the
representation
\[
    \mathbb{Z}_P(z) = \mathbb{B}_\Phi(z) + \phi(z)Z_1 + z\phi(z)Z_2.
\]
Here $\mathbb{B}_\Phi$ is a $\Phi$-Brownian bridge and
$(\mathbb{B}_\Phi,Z_1,Z_2)$ is jointly Gaussian with
\[
    Z_1\sim\N(0,1),\quad Z_2\sim\N(0,1/2),\quad \Cov(Z_1,Z_2)=0,
\]
and cross-covariances
\[
    \Cov(\mathbb{B}_\Phi(z),Z_1)=-\phi(z),\quad
    \Cov(\mathbb{B}_\Phi(z),Z_2)=-\frac12 z\phi(z).
\]
Consequently, the covariance kernel of $\mathbb{Z}_P$ is
\begin{equation}\label{eq:KP}
    K_P(z,y) = \Phi(z\wedge y)-\Phi(z)\Phi(y) -\phi(z)\phi(y)
    -\frac12 zy\,\phi(z)\phi(y).
\end{equation}
\end{lemma}

\noindent The two correction terms in $K_P$ reflect the components of the empirical process absorbed by estimating $\mu$ and $\sigma$ respectively.

\begin{lemma}[Projected empirical process: median/IQR]\label{lem:robust_ep}
Suppose $X_i\stackrel{\text{i.i.d.}}{\sim}\N(\mu_0,\sigma_0^2)$.
Let $a=\Phi^{-1}(3/4)$ and define the influence functions
\[
    \IF_m(w) = \frac{1/2-\mathbf{1}\{w\leq 0\}}{\phi(0)}, \quad \IF_s(w)
    =
    \frac{1}{2a}\left\{\frac{3/4-\mathbf{1}\{w\leq a\}}{\phi(a)} -
    \frac{1/4-\mathbf{1}\{w\leq -a\}}{\phi(a)}\right\}.
\]
Then, in $\ell^\infty(\mathbb{R})$,
\[
    \sqrt{n}\{\widehat F_n^R(z)-\Phi(z)\}\;\Rightarrow\;\mathbb{Z}_R(z),
\]
where $\mathbb{Z}_R$ is the mean-zero Gaussian process admitting the
representation
\[
    \mathbb{Z}_R(z) = \mathbb{B}_\Phi(z) + \phi(z)Z_m + z\phi(z)Z_s.
\]
Here $\mathbb{B}_\Phi$ is a $\Phi$-Brownian bridge and $(\mathbb{B}_\Phi,Z_m,Z_s)$ is jointly Gaussian, where
\[
    Z_m\sim \N(0,\E[\IF_m(W)^2]), \quad
    Z_s\sim \N(0,\E[\IF_s(W)^2]), \quad
    \Cov(Z_m,Z_s)=\E[\IF_m(W)\IF_s(W)],
\]
with $W\sim\N(0,1)$, and
\[
    \Cov(\mathbb{B}_\Phi(z),Z_m) =
    \E[(\mathbf{1}\{W\leq z\}-\Phi(z))\IF_m(W)],
\]
\[
    \Cov(\mathbb{B}_\Phi(z),Z_s) =
    \E[(\mathbf{1}\{W\leq z\}-\Phi(z))\IF_s(W)].
\]
Consequently, the covariance kernel of $\mathbb{Z}_R$ is
\begin{equation}\label{eq:KR}
    K_R(z,y) = \E[\xi_z(W)\xi_y(W)],
\end{equation}
where
\[
    \xi_z(w) = \mathbf{1}\{w\leq z\}-\Phi(z) + \phi(z)\IF_m(w) +
    z\phi(z)\IF_s(w).
\]
\end{lemma}

\begin{remark}[Modularity of the plug-in correction]
\label{rem:plugin_modularity}
The plug-in score variance for standardisation
$\diamond\in\{P,R\}$ is
\begin{equation}\label{eq:plugin_Omega}
    \Omega_{c_0}^\diamond = 4\iint_{\mathbb{R}^2}
    K_\diamond(z,y)\,\dot G_{c_0}(z)\,\dot G_{c_0}(y)\,d\nu(z)\,d\nu(y).
\end{equation}
Relative to the fixed-scale case, the minimum-distance criterion, the boundary argument, and the limiting one-sided quadratic form are unchanged. What changes is the Gaussian empirical-process kernel: the Brownian-bridge kernel is replaced by the projected kernel $K_\diamond$, and consequently the score variance $\Omega_{c_0}$ is replaced by $\Omega_{c_0}^\diamond$. Intuitively, estimating location and scale absorbs from the empirical process the components that look, to first order, like small location and scale perturbations of the null distribution. For mean/SD standardisation this gives the usual normal location-scale projection. For median/IQR standardisation the same mechanism applies, but the projection is driven by the quantile influence functions instead. This replacement has no effect on the form of the boundary limit after standardisation: the normalising constant in the statistic is simply computed with $K_\diamond$ through \eqref{eq:plugin_Omega}. In numerical implementation, the same quadrature scheme used in the fixed-scale case can be used after replacing the kernel. This is the familiar phenomenon in empirical-process goodness-of-fit theory that nuisance-parameter estimation changes the covariance structure of the limiting process, but not the basic form of the resulting projected test statistic.
\end{remark}

\begin{theorem}[Plug-in fit-improvement limit]
\label{thm:plugin_fit_improvement}
Suppose the Gaussian location-scale null holds and the conditions of Theorem~\ref{thm:fit_improvement} hold with $\widehat F_n$ replaced by $\widehat F_n^\diamond$. Then
\[
    \widetilde T_n^\diamond =
    \frac{2H_{c_0}}{\Omega_{c_0}^\diamond}
    \,n\{Q_n^\diamond(c_0)-Q_n^\diamond(\widehat c_n^\diamond)\}
    \;\Rightarrow\; \tfrac12\delta_0+\tfrac12\chi_1^2.
\]
An asymptotic level-$\alpha$ test rejects the Gaussian location-scale null when $\widetilde T_n^\diamond>\chi_{1,1-2\alpha}^2$.
\end{theorem}

\noindent The same substitution principle applies to the remaining minimum-distance results. Consistency of $\widehat c_n^\diamond$, interior asymptotic normality with the corresponding projected score variance, and directed consistency against alternatives satisfying
$\Delta^\diamond=Q^\diamond(c_0)-Q^\diamond(c_\diamond^\star)>0$ follow from the arguments of Theorems~\ref{thm:consistency}, \ref{thm:interior}, and \ref{thm:directed}, with the fixed-scale empirical-process variance replaced by the appropriate projected variance $\Omega^\diamond$.

\begin{corollary}[Moments]\label{cor:plugin_moments}
Under the conditions of Theorem~\ref{thm:plugin_fit_improvement}, $\E[\widetilde T_\infty^\diamond]=1/2$ and $\Var(\widetilde T_\infty^\diamond)=5/4$ for each $\diamond\in\{P,R\}$, by the same calculation as Corollary~\ref{cor:moments_Ttilde}.
\end{corollary}

\begin{corollary}[Plug-in local power]\label{cor:plugin_local_power}
Under local alternatives $F_n^\diamond(z)=\Phi(z)+\eta_\diamond(z)/
\sqrt{n} + o(n^{-1/2})$ satisfying the conditions of
Theorem~\ref{thm:local_power} with $K$ replaced by $K_\diamond$,
\[
    \pi_\nu^\diamond(\eta_\diamond) = \Phi\!\left\{
    \delta_\nu^\diamond(\eta_\diamond)-\Phi^{-1}(1-\alpha)
    \right\}, \quad \delta_\nu^\diamond(\eta_\diamond) =
    \frac{2\int_\mathbb{R} \eta_\diamond(z)\,\dot G_{c_0}(z)\,d\nu(z)}
    {\sqrt{\Omega_{c_0}^\diamond(\nu)}}.
\]
\end{corollary}

\begin{corollary}[Plug-in consistency]\label{cor:plugin_consistency} For $\diamond\in\{P,R\}$, let $F_0^\diamond$ denote the distribution of $(X-\mu_0^\diamond)/\sigma_0^\diamond$, where $(\mu_0^\diamond,\sigma_0^\diamond)$ are the population limits of the corresponding standardising estimators, with $\sigma_0^\diamond>0$. Define \[ Q^\diamond(c) = \int_{\mathbb R}\{F_0^\diamond(z)-G_c(z)\}^2\,d\nu(z), \] and suppose that $Q^\diamond$ has a unique minimiser $c_\diamond^\star\in\mathcal C$. Then $\widehat c_n^\diamond\to c_\diamond^\star$ almost surely. In particular, under the Gaussian location-scale null, $F_0^\diamond=G_{c_0}=\Phi$ for both $\diamond\in\{P,R\}$, and hence $\widehat c_n^\diamond\to c_0$ almost surely.
\end{corollary}

\begin{corollary}[Plug-in directed consistency]\label{cor:plugin_directed} For $\diamond\in\{P,R\}$, let $\Delta^\diamond = Q^\diamond(c_0)-Q^\diamond(c_\diamond^\star)$, where $c_\diamond^\star$ is the unique minimiser of $Q^\diamond$ over $\mathcal C$. If $\Delta^\diamond>0$, then $ n\{Q_n^\diamond(c_0)-Q_n^\diamond(\widehat c_n^\diamond)\} \to_p \infty$. Consequently, $\Pr_{F_0}(\widetilde T_n^\diamond>k_\alpha)\to1$ for any fixed $k_\alpha<\infty$.
\end{corollary}

\begin{remark}[Plug-in versus joint estimation]\label{rem:plugin_vs_joint}
An alternative to the plug-in approach is to estimate $\theta=(\mu,\sigma,c)$ jointly by minimising
\[
    (\widehat\mu_n,\widehat\sigma_n,\widehat c_n)
    \in
    \arg\min_{\mu,\sigma,c}
    \int_{\mathbb{R}}\left\{
    \widehat F_n(z)
    -G_c\!\left(\frac{z-\mu}{\sigma}\right)
    \right\}^2 d\nu(z).
\]
The joint theory is developed in Appendix~\ref{app:joint}. I prefer the plug-in approach for three reasons. First, the criterion has a nearly degenerate score-variance Schur complement in the $c$-direction after optimising over $(\mu,\sigma)$: under the Gaussian null, the $c$-score direction $\dot G_{c_0}$ nearly lies in the span of $\{\phi(z),z\phi(z)\}$ in the Brownian bridge covariance metric, making the residual score variance
$\Omega_{cc\cdot\lambda}\approx0.000013$ very small
(Table~\ref{tab:null_size_tail_scaling_robust}). The curvature Schur complement $H_{cc\cdot\lambda}\approx0.0023$ is not similarly reduced, giving a normalisation constant $2H_{cc\cdot\lambda}/\Omega_{cc\cdot\lambda}
\approx355$ that amplifies any small spurious criterion improvement into a large test statistic; since this is structural, the size distortion worsens with $n$. Simulation evidence confirms the severity: the $10\%$ rejection rate rises from $0.089$ at $n=200$ to $0.309$ at $n=1600$, with the boundary rate falling away from the theoretical $1/2$
(Table~\ref{tab:null_size_tail_exact_profile_appendix}). Second, the plug-in approach preserves the modular structure of the theory: the limit distribution, curvature, and rejection rule are invariant across all three versions of the test; only the projected covariance kernel must be recomputed. Third, the two-step structure allows the location-scale estimator to be adapted to the operating conditions: median/IQR standardisation remains valid when the variance is infinite ($c>3/2$), while mean/SD standardisation does not.
\end{remark}

\section{Finite-sample performance}\label{sec:simulations}

The simulations examine the finite-sample size and power of three hypergeometric (HGM) statistics: the fixed-scale statistic $\widetilde T_n$, the mean/SD plug-in statistic $\widetilde T_n^P$, and the median/IQR plug-in statistic $\widetilde T_n^R$. All use the tail-weighted criterion $d\nu_{1/2}(z)\propto [\Phi(z)\{1-\Phi(z)\}]^{-1/2}\,d\Phi(z)$, and asymptotic critical values from $\frac12\delta_0+\frac12\chi_1^2$. The criterion is evaluated on the Gaussian quantile grid \eqref{eq:grid} with $J=301$ nodes and normalised weights \eqref{eq:weights}. The null constants $H_{c_0}$, $\Omega_{c_0}$, $\Omega_{c_0}^P$, and $\Omega_{c_0}^R$ are computed once and reported in Table~\ref{tab:null_size_tail_scaling_robust}. The parameter space is $\mathcal C=[3/2,8]$, searched over a fine grid concentrated near the boundary. Unless stated otherwise, size results are based on $R=10^5$ replications and power results on $R=10^4$ replications.

\subsection{Null size}\label{subsec:sim_size}

Table~\ref{tab:null_size_tail_main_robust} reports empirical rejection
frequencies under the Gaussian null $F_0=\Phi$ for
$n\in\{25,50,100,200,400,800,1600\}$ and nominal levels
$10\%$, $5\%$, and $1\%$. The boundary frequency
$\Pr(\widehat c=c_0)$ and the mean of the standardised statistic provide
diagnostics for the boundary limit, whose corresponding values are
$1/2$ and $1/2$. The fixed-scale statistic is accurately sized throughout the design. At the $5\%$ level, rejection frequencies range from $0.050$ to $0.052$, and both diagnostic columns are close to their limiting values even for $n=25$. Thus the fixed-scale boundary approximation is reliable in quite small samples.

The mean/SD plug-in statistic converges more slowly but is well calibrated by moderate sample sizes. It mildly over-rejects at the $10\%$ level for small $n$ and is conservative at the $1\%$ level for small $n$, but these distortions diminish as $n$ increases. The slower convergence of the boundary frequency reflects the additional first-order effect of estimating location and scale, which is accounted for asymptotically by the projected kernel $K_P$. The median/IQR plug-in statistic is the most demanding finite-sample case. It is liberal for small samples, with $5\%$ size equal to $0.095$ at $n=25$, and remains mildly oversized at the largest sample sizes reported. This is consistent with the greater sampling variability of quantile-based standardisation. In practice, the robust plug-in test should be used with caution at $n<200$, although its size distortion declines substantially as $n$ grows. Appendix~\ref{app:simulations} reports an additional check: under $\N(3,4)$ the robust plug-in rejection frequencies match those under $\N(0,1)$ up to simulation error, confirming location-scale invariance (Table~\ref{tab:null_size_tail_robust_plugin_invariance_N3_2sq}). Overall, the simulations support use of the boundary-mixture critical values for the fixed-scale and mean/SD plug-in statistics, and show that the robust median/IQR version requires larger samples for accurate null calibration.

\input{table_null_size_tail_main_with_robust}

\subsection{Power}\label{subsec:sim_power}

The test is directed: by Theorem~\ref{thm:local_power}, local power is determined by the projection of the departure onto the hypergeometric score direction $\dot G_{c_0}$, giving power against symmetric heavy-tailed alternatives and limited power against skewed alternatives by construction. Three tests are compared: the fixed-scale $G_c$ test, the robust plug-in $\widetilde T_n^R$, and the omnibus benchmarks Jarque--Bera (JB), Anderson--Darling (AD), and Shapiro--Wilk (SW). Mean/SD plug-in and size-adjusted results are in Appendix~\ref{app:simulations}, where the mean/SD plug-in is shown to fail completely under HGM alternatives with infinite variance.

\subsubsection{HGM alternatives}

Table~\ref{tab:power_hgm_alternatives_revised} and Figure~\ref{fig:power_curves_robust} (upper panel) report rejection frequencies against HGM alternatives $G_c$ with $c\in\{1.55,1.60,1.75,2.00,3.00\}$ at nominal level $5\%$. The fixed-scale test $\widetilde T_n$ has strong power against HGM alternatives, as the theory predicts. At $c=2.00$ power reaches $0.590$ at $n=25$ and $0.967$ by $n=100$; at $c=3.00$ it reaches $0.964$ at $n=25$. This confirms the directed consistency of Theorem~\ref{thm:directed} in finite samples. For the closest alternative, $c=1.55$, the local-power approximation in Corollary~\ref{cor:local_hgm}, $\pi_n(c)
\approx \Phi\left\{ 1.052 \sqrt{n}\left(c-\frac32\right)-1.645
\right\}$, is also quantitatively accurate: it predicts rejection frequencies of $0.132$, $0.184$, $0.277$, and $0.438$ for $n=100,200,400$, and $800$, respectively, compared with the simulated frequencies $0.135$, $0.186$, $0.275$, and $0.413$. The robust plug-in test $\widetilde T_n^R$ has substantially lower power. At $c=1.60$ it reaches only $0.167$ at $n=800$; at $c=3.00$ it reaches $0.986$ at $n=800$. The gap relative to the fixed-scale test reflects the loss of signal from the median/IQR standardisation: the projected kernel $K_R$ absorbs part of the HGM tail signal together with the location-scale nuisance, as Corollary~\ref{cor:plugin_local_power} predicts. The omnibus benchmarks substantially outperform the robust plug-in at close alternatives: JB reaches $0.931$ at $c=1.60$, $n=400$, against $0.129$ for $\widetilde T_n^R$, reflecting the high finite-sample kurtosis of HGM alternatives with $c>3/2$, to which JB's kurtosis component is directly sensitive; $\widehat c$ by contrast provides a fitted distribution with explicit tail-risk implications rather than a rejection decision.

\subsubsection{Standard non-HGM alternatives}

Table~\ref{tab:power_standard_alternatives_revised} and Figure~\ref{fig:power_curves_robust} (lower panel) report rejection frequencies against six alternatives: contaminated normal, Laplace, chi-squared $\chi^2_3$, $t_3$, $t_5$, and $t_{10}$. The fixed-scale test has essentially zero power against every non-HGM alternative. Against $t_3$ it
never exceeds zero. The standardised $t_3$ distribution has density tail index $4$ (survival $\propto|z|^{-3}$), which is lighter than the HGM tail-index-2 direction (density $\propto|z|^{-3}$, survival $\propto|z|^{-2}$) targeted by the test: the best HGM approximation to $t_3$ data is therefore at the Gaussian boundary $c_0$, giving $\Delta=0$ and no directed power. Against the chi-squared, power is essentially zero because the departure is in the skewness direction, which is orthogonal to $\dot G_{c_0}$ by the symmetry of the
HGM family.

The robust plug-in has useful power against symmetric heavy-tailed alternatives. Against Laplace it reaches $0.272$ at $n=25$ and $0.999$ at $n=800$; against $t_3$ it reaches $0.195$ at $n=25$ and $0.977$ at $n=800$. Power against the contaminated normal and $t_{10}$ is moderate, reflecting their relative proximity to the Gaussian. Against the chi-squared, power falls to zero at large $n$, consistent with Theorem~\ref{thm:local_power}: skewed departures are orthogonal to the even function $\dot G_{c_0}$. The omnibus benchmarks dominate the robust plug-in against non-HGM alternatives: JB, AD, and SW have substantially higher power against the contaminated normal, chi-squared, and $t_5$. At $n=25$ the robust plug-in reaches $0.122$ against $t_{10}$, slightly above JB ($0.100$) and AD ($0.080$). This reflects the local sensitivity of directed tests near their targeted direction; after size adjustment the advantage disappears (Table~\ref{tab:appendix_sizeadj_robust_standard_power}). The HGM test complements rather than competes with omnibus procedures: it isolates the symmetric heavy-tailed component of non-Gaussianity and provides a fitted parameter $\widehat c$ with direct tail-risk interpretation.

\input{table_power_hgm_alternatives_revised}

\input{table_power_standard_alternatives_revised}

\begin{figure}[!p]
\centering
\includegraphics[width=0.78\textwidth]{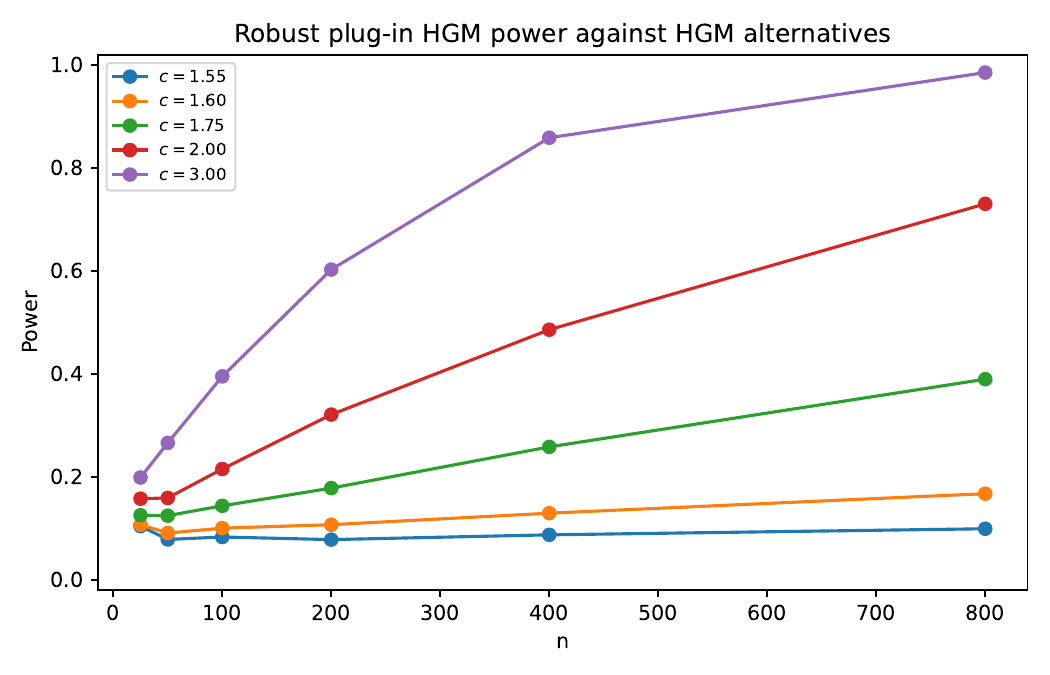}
\vspace{0.4cm}
\includegraphics[width=0.78\textwidth]{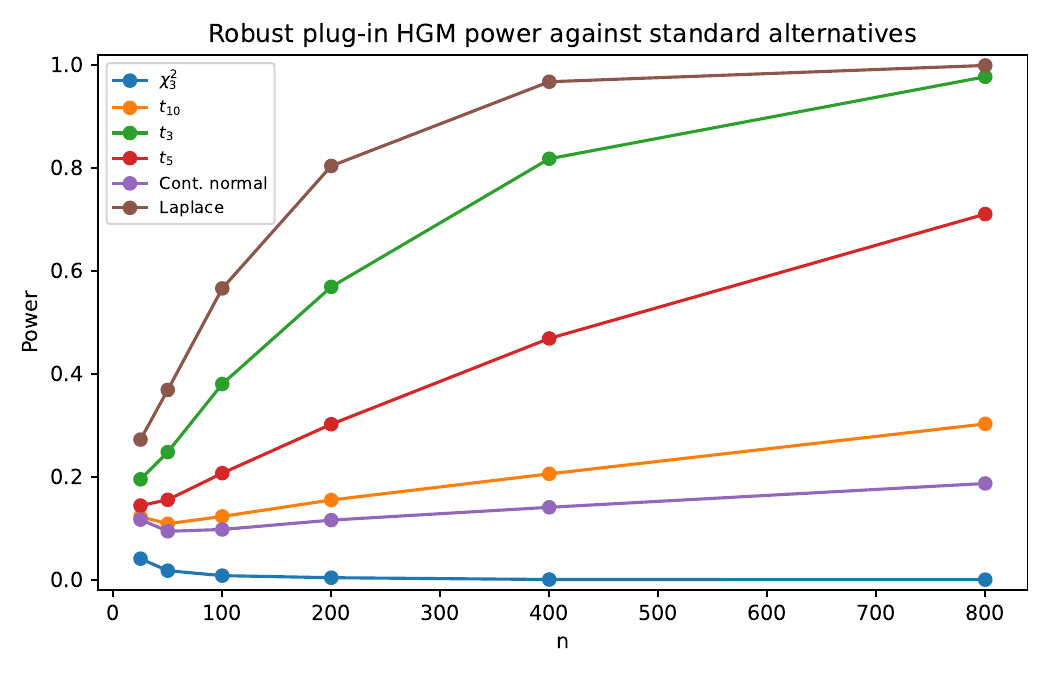}
\caption{Power curves for the robust plug-in statistic $\widetilde T_n^R$ at nominal level $5\%$, $R=10^4$ replications. Upper panel: HGM alternatives $G_c$, $c\in\{1.55,1.60,1.75,2.00,3.00\}$. Lower panel: standard alternatives standardised to mean zero and unit variance (contaminated normal, Laplace, chi-squared $\chi^2_3$, $t_3$, $t_5$, $t_{10}$).}
\label{fig:power_curves_robust}
\end{figure}

\section{Empirical illustration}\label{sec:application}

This section provides two descriptive illustrations using daily S\&P~500 index data from the FRED database \citep{fred_sp500}. Daily log returns $r_t=100\ln(P_t/P_{t-1})$ give $n=1255$ observations from 2~June~2021 to 1~June~2026. The first exercise applies the HGM test to raw returns; the second applies it to standardised residuals from a GARCH(1,1) model
\citep{bollerslev86},
\[
    r_t = \mu+\varepsilon_t, \quad
    h_t = \omega+\alpha\varepsilon_{t-1}^2+\beta h_{t-1},
    \quad z_t = \varepsilon_t/\sqrt{h_t},
\]
where $h_t$ is the conditional variance and $z_t$ is the standardised residual. The model is estimated by Gaussian MLE, with $\widehat\alpha=0.105$, $\widehat\beta=0.863$, $\widehat\alpha+\widehat\beta=0.968$. Both exercises are descriptive: the formal theory of Sections~\ref{sec:inference}--\ref{sec:plugin} assumes i.i.d. observations. The GARCH residual exercise additionally involves a generated-regressors problem; it serves as a robustness check on whether tail non-normality persists in the shape of the standardised residuals after filtering out volatility clustering. Figure~\ref{fig:garch_volatility} shows the estimated conditional volatility $\widehat h_t^{1/2}$: the 2022 equity decline and the tariff-driven market stress of early 2025 both produce sustained episodes of higher conditional volatility. Full descriptive statistics, goodness-of-fit tables, and GARCH diagnostics are in Tables \ref{tab:sp500_returns_desc}--\ref{tab:sp500_garch_residuals_tail} of Appendix~\ref{app:empirical}.

\subsection{Unconditional distribution of raw returns}\label{subsec:emp_raw}

Returns are standardised by $(\widehat\mu_n^R, \widehat\sigma_n^R)$ of~\eqref{eq:median_iqr} and $c$ is estimated by minimising $Q_n^R(c)$ with $d\nu_{1/2}$. Table~\ref{tab:sp500_returns_tests} reports normality test results: Jarque--Bera, Anderson--Darling, and Shapiro--Wilk all strongly reject the Gaussian null. The HGM statistic is $\widetilde T_n^R=26.65$ ($p\approx1.2\times10^{-7}$), with tail coefficient $\widehat c-3/2=0.135$, locating the departure in the symmetric heavy-tailed direction. Figure~\ref{fig:sp500_density_overlay} overlays the fitted densities; Figure~\ref{fig:sp500_tail_survival} plots two-sided survival functions on a logarithmic scale. In the observable range ($z\in[2,5]$) HGM lies below Student $t$ because the tail coefficient is small and the HGM asymptotic dominance sets in only beyond the sample range. By average log score, Student $t$ provides the best fit and HGM ranks third behind Laplace (Table~\ref{tab:sp500_returns_gof}). Lower-tail quantiles are reported in Table~\ref{tab:sp500_returns_var}. The models diverge materially at the $0.1\%$ level: Gaussian $-3.24$, Laplace $-4.60$, empirical $-4.82$, Student $t$ $-5.51$, HGM $-6.57$. The HGM quantile falls below the sample minimum ($-6.16$), illustrating that $\widehat c$ carries tail-risk implications that differ substantially from standard alternatives even when $\widehat c-3/2$ is small.

\subsection{Conditional innovation distribution: GARCH(1,1) residuals}\label{subsec:emp_garch}

The standardised residuals $\widehat z_t$ are treated as the sample, with $(\widehat\mu_n^R,\widehat\sigma_n^R)$ applied before fitting.  Excess kurtosis falls from $6.54$ (raw returns) to $1.56$, confirming that the GARCH filter accounts for most tail heaviness through time-varying volatility. Table~\ref{tab:sp500_garch_residuals_tests} reports test results. The HGM statistic falls to $\widetilde T_n^R=6.32$ ($p=0.006$), with $\widehat c =1.563$, much closer to the Gaussian boundary. Omnibus tests still reject normality, driven partly by residual left skewness ($-0.49$) that lies outside the scope of the symmetric HGM family. By log score, HGM now ranks second behind Student $t$ only, ahead of Laplace and Gaussian, reflecting closer alignment of the filtered residual tail structure with the hypergeometric family (Table~\ref{tab:sp500_garch_residuals_gof}). Figures~\ref{fig:sp500_garch_density} and~\ref{fig:sp500_garch_tail} display fitted densities and survival functions. Lower-tail quantiles at the $0.1\%$ level are (Table~\ref{tab:sp500_garch_residuals_var}): Gaussian $-3.12$, Student $t$ $-4.05$, Laplace $-4.68$, empirical $-4.70$, HGM $-5.02$. Student $t$ falls short of the empirical quantile; HGM gives the most extreme extrapolation, as for raw returns.

\subsection{Tail dominance and Value-at-Risk}\label{subsec:crossover}

For every $c>3/2$ and $z>0$, $1-G_c(z)>1-\Phi(z)$: the HGM distribution assigns strictly more probability to tail events than the Gaussian. The correction factor $\rho(z,c):=(1-G_c(z))/(1-\Phi(z))>1$ is strictly increasing in both $z$ and $c$ (Lemma~\ref{lem:R_monotone}), with exact values in Table~\ref{tab:correction}. HGM VaR quantiles are computed exactly by inverting $1-G_c(q_p)=1-p$; Table~\ref{tab:var_ratio} reports exact values alongside Gaussian quantiles. The ratio $q_p^{\mathrm{HGM}}/q_p^{\Phi}$ diverges as $p\to1^-$ for any $c>3/2$ (Proposition~\ref{prop:var_divergence}): any statistically significant rejection of $H_0:c=c_0$, however small $\widehat c-3/2$, implies understatement of extreme tail risk by the Gaussian model at sufficiently high confidence levels. Full details are in Appendix~\ref{app:crossover}.

\subsection{Limitations}\label{subsec:emp_limitations}

These exercises are not intended to establish that the HGM distribution fits better than all competitors: in both exercises Student $t$ provides a better overall log score, and omnibus tests have broader power. The purpose is to show what the fitted parameter $\widehat c$ implies for tail risk, and how that implication changes when volatility clustering is removed by GARCH filtering. Three limitations apply. First, the formal null theory assumes i.i.d.\ observations; raw returns exhibit serial dependence and volatility clustering \citep{engle82,bollerslev86}, and GARCH residuals are generated regressors, so size control is not formally guaranteed in either case. Second, the symmetric HGM family cannot capture the residual left skewness ($-0.49$) in the GARCH residuals, motivating future extensions to asymmetric families. Third, minimum-distance estimation rather than MLE means comparisons of log scores are descriptive rather than formal likelihood-ratio tests.

\begin{figure}[p]
\centering
\includegraphics[width=0.82\textwidth]{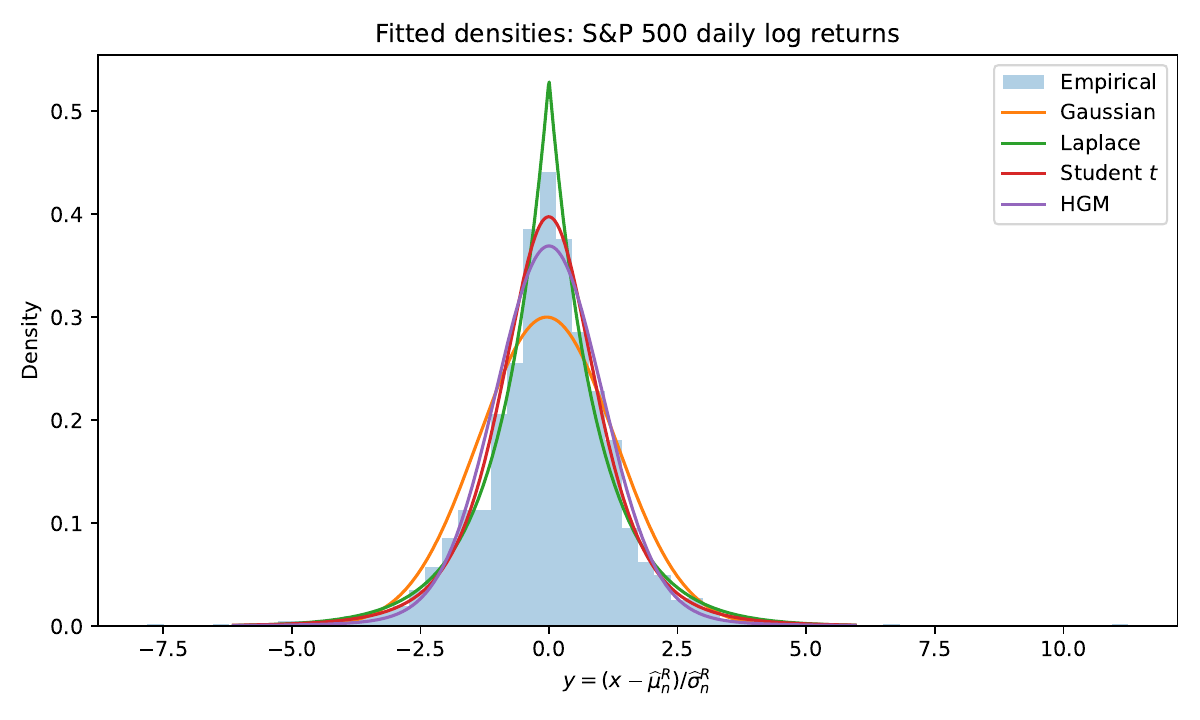}
\caption{Fitted densities of robust-standardised raw returns: HGM and Student $t$ are similar near the centre; both are more peaked than the Gaussian.}
\label{fig:sp500_density_overlay}
\end{figure}

\begin{figure}[p]
\centering
\includegraphics[width=0.82\textwidth]{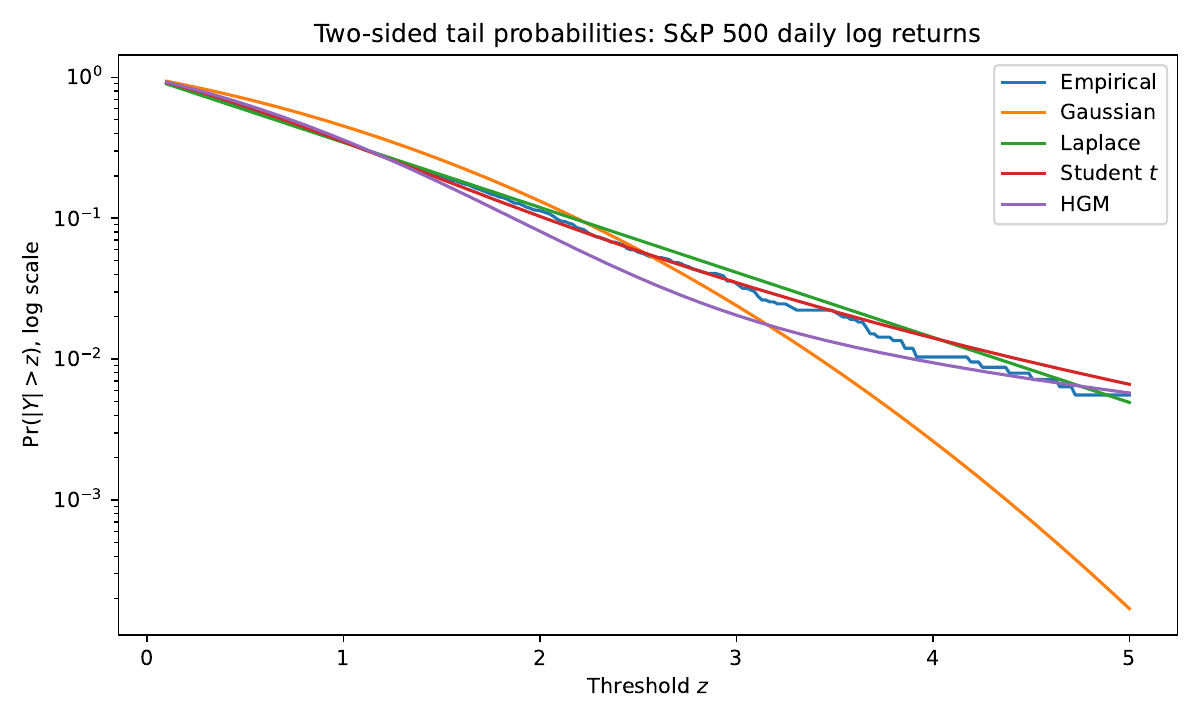}
\caption{Two-sided survival functions $\Pr(|Y|>z)$ for robust-standardised raw returns, logarithmic scale: HGM lies below Student $t$ in the 2--5 standard deviation range despite the heavier asymptotic tail index, reflecting the small tail coefficient $\widehat c-3/2=0.135$.}
\label{fig:sp500_tail_survival}
\end{figure}

\begin{figure}[p]
\centering
\includegraphics[width=0.82\textwidth]
    {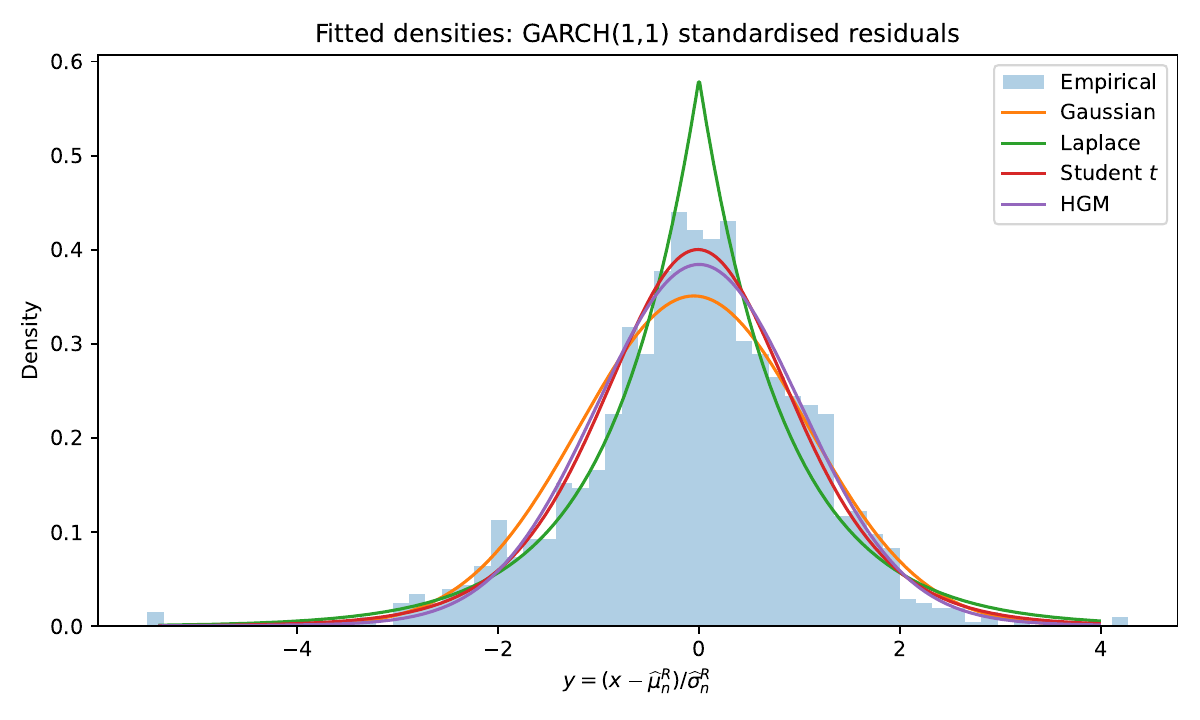}
\caption{Fitted densities of robust-standardised GARCH(1,1) residuals: HGM and Student $t$ are the closest competitors; excess kurtosis is reduced relative to raw returns.}
\label{fig:sp500_garch_density}
\end{figure}

\begin{figure}[p]
\centering
\includegraphics[width=0.82\textwidth]
    {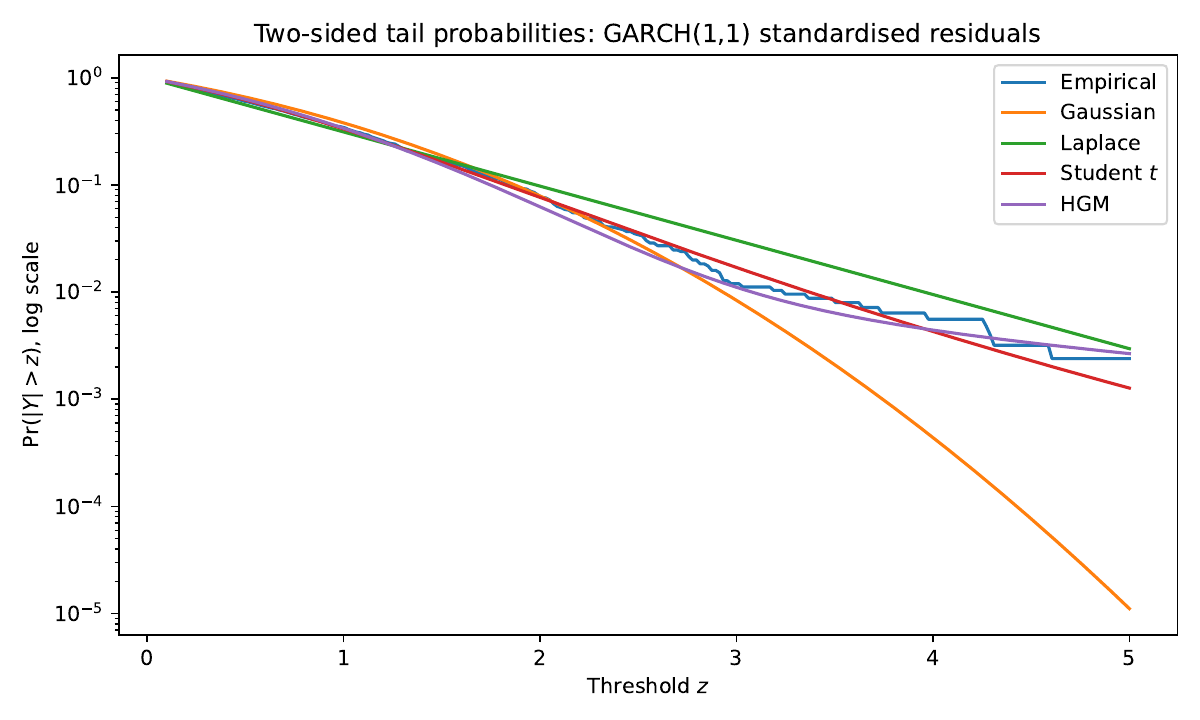}
\caption{Two-sided survival functions $\Pr(|Y|>z)$ for robust-standardised GARCH(1,1) residuals, logarithmic scale: HGM and Student $t$ track the empirical tail more closely than for raw returns; Laplace overstates tail mass at intermediate thresholds.}
\label{fig:sp500_garch_tail}
\end{figure}

\clearpage

\section{Conclusion}\label{sec:conclusion}

This paper develops a complete inferential framework for a one-parameter family of distribution functions built from the confluent hypergeometric function ${}_1F_1$. The family $\{G_c:c\geq3/2\}$ places the standard normal at an exact finite boundary, with algebraic tails of fixed index $2$ for $c>3/2$ and a beta-precision Gaussian scale-mixture representation. These properties are mathematical consequences of the endpoint constraints that determine the family from the ${}_1F_1$ representation of $\Phi$, not modelling assumptions. More broadly, the construction reverses the usual role of hypergeometric functions in econometrics. Rather than arising as analytic solutions to a model or inferential problem specified in advance, the hypergeometric representation is used here as the starting point for model construction: selected parameters are freed, the restrictions required of the econometric object are imposed, and the resulting admissible parameters are estimated from the data.

The central result is that the standardised fit-improvement statistic $\widetilde T_n$ converges to $\frac{1}{2}\delta_0+\frac{1}{2}\chi_1^2$ under the Gaussian null, with the explicit rejection rule $\widetilde T_n>\chi^2_{1,1-2\alpha}$. This limit holds for all three versions of the test: fixed-scale, mean/SD plug-in, and robust median/IQR plug-in, with only the projected covariance kernel changing across versions. The modular structure is deliberate: extensions to new standardisation procedures require only a new kernel, leaving the limit distribution, curvature, and rejection rule unchanged. The fitted parameter $\widehat c$ provides a constructive complement to omnibus normality tests: it locates the departure from normality in the hypergeometric tail direction, quantifies the departure as a single interpretable index, and implies tail-risk extrapolations that differ materially from Gaussian benchmarks even when $\widehat c-3/2$ is small.

The framework opens a natural research programme. The most immediate extension frees the tail index, generating a two-parameter family that nests the present symmetric case and extends the boundary inference to a two-dimensional constrained problem. A closely related development introduces asymmetry, allowing the family to capture signed skewness alongside heavy tails and generating a richer mixture limit at the boundary. A further line of research develops the formal two-step asymptotic theory for the test applied to residuals from parametric volatility models, where first-step GARCH estimation uncertainty propagates through an additional projected-process correction to the covariance kernel. A fourth extension pursues likelihood-based inference using the explicit density $g_c$ of Corollary~\ref{cor:density}, yielding an efficiency comparison and a likelihood-ratio counterpart to the fit-improvement statistic. A longer-horizon direction pursues multivariate normality testing via matrix-argument hypergeometric functions, replacing the scalar boundary $c_0=3/2$ with a matrix boundary at the identity. The present paper establishes the foundations on which this programme builds.

\newpage


\appendix

\section{Proofs}\label{app:proofs}

\begin{proof}[Proof of Lemma~\ref{lem:normalcdf}]
The identity is classical \citep[eq.~7.1.21]{abramowitz_stegun72}. Recall the series definition of the confluent hypergeometric function,
\[
    {}_1F_1(a;c;x) = \sum_{k=0}^{\infty}\frac{(a)_k}{(c)_k}\frac{x^k}{k!},
\]
where $(a)_k=a(a+1)\cdots(a+k-1)$ denotes the Pochhammer symbol. Therefore
\[
    {}_1F_1\!\left(\frac12;\frac32;-\frac{z^2}{2}\right) =
    \sum_{k=0}^{\infty}\frac{(1/2)_k}{(3/2)_k}\frac{(-z^2/2)^k}{k!}.
\]
For each $k\geq0$,
\[
    \frac{(1/2)_k}{(3/2)_k} = \prod_{j=0}^{k-1}\frac{2j+1}{2j+3} = \frac{1}{2k+1}.
\]
Hence
\[
    z\,{}_1F_1\!\left(\frac12;\frac32;-\frac{z^2}{2}\right) =
    \sum_{k=0}^{\infty}\frac{(-1)^k z^{2k+1}}{2^k k!(2k+1)}.
\]
The power series is entire, so it may be differentiated term by term on $\mathbb{R}$, giving
\[
    \frac{d}{dz}\left[z\,{}_1F_1\!\left(\frac12;\frac32;-\frac{z^2}{2}\right)\right]=\sum_{k=0}^{\infty}\frac{(-1)^k z^{2k}}{2^k k!} =
    e^{-z^2/2}.
\]
Define
\[
    H(z) = \frac12 + \frac{z}{\sqrt{2\pi}}\,
    {}_1F_1\!\left(\frac12;\frac32;-\frac{z^2}{2}\right).
\]
Then
\[
    H'(z) = \frac{1}{\sqrt{2\pi}}\,e^{-z^2/2} =
    \phi(z) = \Phi'(z), \quad H(0)=\frac12=\Phi(0).
\]
It follows that
\[
    H(z)=\Phi(z)
\]
for all \(z\in\mathbb R\).
\end{proof}

\begin{proof}[Proof of Lemma~\ref{lem:endpoint}]
Under the stated parameter exclusions, Kummer's function obeys the
large-argument asymptotic expansion \citep[eq.~13.1.5]{abramowitz_stegun72}
\[
    {}_1F_1(a;c;x) \sim \frac{\Gamma(c)}{\Gamma(c-a)}\,|x|^{-a}, \quad x\to-\infty,
\]
with finite, non-zero leading coefficient. Substituting $x=-z^2/2$,
so that $|x|^{-a}=2^{a}|z|^{-2a}$ and
$z\,|z|^{-2a}=\operatorname{sgn}(z)|z|^{1-2a}$, gives
\[
    G_{a,c}(z) \sim \frac12 + 2^{a}\,\kappa(a,c)\,
    \frac{\Gamma(c)}{\Gamma(c-a)}\,\operatorname{sgn}(z)\,|z|^{1-2a},
    \quad |z|\to\infty.
\]
The coefficient $2^{a}\kappa(a,c)\,\Gamma(c)/\Gamma(c-a)$ is finite and
non-zero. Three cases arise according to the value of $a$. If $a<1/2$, then $|z|^{1-2a}\to\infty$, so both endpoint limits are infinite. If $a>1/2$, then $|z|^{1-2a}\to 0$, so both endpoint limits equal $1/2$. Neither is compatible with the prescribed limits $0$ and $1$. Therefore $a=1/2$, in which case $|z|^{1-2a}=1$ and the two endpoint conditions reduce to
\[
    \sqrt{2}\,\kappa\!\left(\tfrac12,c\right)\,\frac{\Gamma(c)}{\Gamma(c-1/2)} = \frac12.
\]
Solving for $\kappa(1/2,c)$ yields~\eqref{eq:Khalf}.
\end{proof}

\begin{proof}[Proof of Lemma~\ref{lem:monotonicity}]
Using $\frac{d}{dx}\,{}_1F_1(a;c;x)=\frac{a}{c}\,{}_1F_1(a+1;c+1;x)$ and the chain rule,
\[
    G_c'(z) = \kappa(c)\left[{}_1F_1\!\left(\frac12;c;-\frac{z^2}{2}\right) - \frac{z^2}{2c}\,{}_1F_1\!\left(\frac32;c+1;-\frac{z^2}{2}\right)\right].
\]
Since $G_c'$ is even in $z$, it suffices to show $G_c'(z)\geq 0$
for $x:=z^2/2\geq 0$. Applying Kummer's transformation
\citep[eq.~13.1.27]{abramowitz_stegun72},
${}_1F_1(a;c;x)=e^x\,{}_1F_1(c-a;c;-x)$, to each term gives
\[
    G_c'(z) = \kappa(c)\,e^{-x}\left[{}_1F_1\!\left(c-\frac12;c;x\right) - \frac{x}{c}\,{}_1F_1\!\left(c-\frac12;c+1;x\right)\right].
\]
For $c>1/2$ the Euler integral representation
\citep[eq.~13.2.1]{abramowitz_stegun72} applies to both terms.
Substituting and simplifying using $\Gamma(c+1)=c\,\Gamma(c)$,
$\Gamma(1/2)=\sqrt{\pi}$, and $\Gamma(3/2)=\sqrt{\pi}/2$ gives
\[
    G_c'(z) = \frac{e^{-x}}{2\sqrt{2\pi}}\int_0^1 e^{xu}\,u^{c-3/2}\,(1-u)^{-1/2}\left[1-2x(1-u)\right]du.
\]
The identity
\[
    e^{xu}(1-u)^{-1/2}\left[1-2x(1-u)\right] = -2\,\frac{d}{du}\!\left[e^{xu}(1-u)^{1/2}\right]
\]
allows the integrand to be written as $-2\,u^{c-3/2}\,V'(u)$, where
$V(u):=e^{xu}(1-u)^{1/2}$, so
\[
    G_c'(z) = \frac{-e^{-x}}{\sqrt{2\pi}}\int_0^1 u^{c-3/2}\,V'(u)\,du.
\]
Integrate by parts on $[\delta,1-\varepsilon]$ and let
$\delta\downarrow 0$ and $\varepsilon\downarrow 0$. The upper
boundary term satisfies $u^{c-3/2}\,V(u)\big|_{u=1-\varepsilon}
=O(\varepsilon^{1/2})\to 0$ for every fixed $c$.

If $c>3/2$: the lower boundary term satisfies
$\delta^{c-3/2}\,V(\delta)\to 0$, and the resulting integral
converges since $c-5/2>-1$. Hence
\begin{equation}\label{eq:Gc_integral}
    G_c'(z) = \frac{(c-\tfrac32)\,e^{-x}}{\sqrt{2\pi}}
    \int_0^1 e^{xu}\,u^{c-5/2}\,(1-u)^{1/2}\,du > 0,
\end{equation}
since the integrand is strictly positive on $(0,1)$, the endpoints contribute nothing to the integral, and $c-3/2>0$.

If $c=3/2$: then $u^{c-3/2}\equiv 1$ and, evaluating directly,
\[
    G_c'(z) = -\frac{e^{-x}}{\sqrt{2\pi}}\bigl[V(u)\bigr]_0^1 = \frac{e^{-x}}{\sqrt{2\pi}} = \phi(z) > 0,
\]
where $\phi$ denotes the standard normal density. In both cases $G_c'(z)>0$, which establishes the stated inequality $G_c'(z)\geq 0$.

The more compact form of the density \eqref{eq:gc} can be used to give an  alternative proof of Lemma~\ref{lem:monotonicity}. For $c>3/2$, Euler's integral representation of ${}_1F_1(3/2;c;-z^2/2)$ has a strictly positive integrand, while at $c=3/2$ the identity ${}_1F_1(c;c;x)=e^x$ gives $g_{3/2}(z)=\phi(z)>0$. I retain the direct argument of
Lemma~\ref{lem:monotonicity}, since it may generalise more usefully in later work, and it does not depend on a fortunate reduction to a single contiguous function.
\end{proof}

\begin{proof}[Proof of Corollary~\ref{cor:Gc_distn}]
Lemma~\ref{lem:endpoint} gives $G_c(-\infty)=0$ and $G_c(+\infty)=1$; Lemma~\ref{lem:monotonicity} gives $G_c'\geq0$, so $G_c$ is non-decreasing. Continuity follows because ${}_{1}F_1(\frac12;c;\,\cdot\,)$
is entire and $z\mapsto-z^2/2$ is smooth, so $G_c$ is a composition of continuous functions. The three conditions of a distribution function
are satisfied.
\end{proof}

\begin{proof}[Proof of Corollary~\ref{cor:density}]
Since ${}_{1}F_1(\frac12;c;\,\cdot\,)$ is entire, $G_c$ is differentiable on $\mathbb{R}$ with derivative $g_c:=G_c'$. Applying the recurrence relation \citep[eq.~13.4.4]{abramowitz_stegun72} ${}_1F_1(a+1;c;x)={}_1F_1(a;c;x)+\frac{x}{c}\,{}_1F_1(a+1;c+1;x)$ to $G_c'(z)$ from the proof of Lemma~\ref{lem:monotonicity}, \eqref{eq:gc} follows immediately. The fundamental theorem of calculus and $G_c(0)=\frac12$ give $G_c(z)=\frac12+\int_0^z g_c(t)\,dt$, so $G_c$ is absolutely continuous with density $g_c$. Since the ${}_1F_1$ term in \eqref{eq:gc} depends on $z$ only through $-z^2/2$, it follows that $g_c$ is even. Finally, $g_{3/2}=G_{3/2}'=\Phi'=\phi$ by Lemma~\ref{lem:normalcdf}.
\end{proof}

\begin{proof}[Proof of Lemma~\ref{lem:continuity}]
From \eqref{eq:Gc_def},
\[
    G_c(z)-G_{c'}(z) = z\left[\kappa(c)\,{}_1F_1\!\left(\tfrac12;c;-\tfrac{z^2}{2}\right) - \kappa(c')\,{}_1F_1\!\left(\tfrac12;c';-\tfrac{z^2}{2}\right)\right].
\]
The map $c\mapsto \kappa(c)=\Gamma(c-1/2)/(2\sqrt{2}\,\Gamma(c))$
is continuous on $[3/2,\infty)$ since the Gamma function is
continuous and non-vanishing on $(0,\infty)$. The map $c\mapsto{}_1F_1(\frac12;c;x)$ is continuous for each fixed $x\in\mathbb{R}$ and $c\in[3/2,\infty)$. Each term
$(1/2)_k/[(c)_kk!]\cdot x^k$ is continuous in $c$ on $[3/2,\infty)$, since $(c)_k>0$ is a polynomial in $c$ with no zeros there. For any
compact subset $[c_L,c_U]\subset[3/2,\infty)$, $(c)_k\geq(c_L)_k>0$, so
$|(1/2)_k/(c)_k|\leq(1/2)_k/(c_L)_k$ and $\sum_{k=0}^\infty(1/2)_k/(c_L)_k\cdot
|x|^k/k!={}_1F_1(\tfrac12;c_L;|x|)<\infty$ since ${}_1F_1$ is entire; the series therefore converges absolutely and uniformly in $c$ on
$[c_L,c_U]$ by the Weierstrass M-test, and continuity of the sum follows by the uniform limit theorem. Pointwise continuity of
$c\mapsto G_c(z)$ then follows since $G_c(z)=\tfrac12+z\, \kappa(c)\,
{}_1F_1(\tfrac12;c;-z^2/2)$ is a product of continuous functions of $c$ for each fixed $z$.

For uniform convergence, write $G_c(z)-\Phi(z)=G_c(z)-G_{3/2}(z)$, so
\[
    \sup_{z\in\mathbb{R}}|G_c(z)-\Phi(z)| =
    \sup_{z\in\mathbb{R}}|z|\,\left|
    \kappa(c)\,{}_1F_1\!\left(\tfrac12;c;
    -\tfrac{z^2}{2}\right)
    -\frac{1}{\sqrt{2\pi}}\,{}_1F_1\!
    \left(\tfrac12;\tfrac32;-\tfrac{z^2}{2}
    \right)\right|.
\]
Applying the asymptotic expansion \citep[eq.~13.5.1]{abramowitz_stegun72}
\[
    {}_1F_1(a;c;-x) =
    \frac{\Gamma(c)}{\Gamma(c-a)}\,x^{-a}
    \left[1+\frac{a(1+a-c)}{x}+O(x^{-2})
    \right]
\]
with $x=z^2/2\to+\infty$ to each term, the leading $x^{-1/2}$ terms cancel identically since $\kappa(c)\frac{\Gamma(c)}{\Gamma(c-\frac12)}\sqrt{2}
\equiv\frac12$, leaving
\[
    \sup_{z\in\mathbb{R}}|G_c(z)-\Phi(z)| \leq
    \sup_{|z|>Z}\left|\tfrac12\bigl(\tfrac32-c\bigr)|z|^{-2}
    +o(|z|^{-2})\right| + \sup_{|z|\leq Z}|G_c(z)-\Phi(z)|.
\]
The first term is bounded by $\tfrac12|c-\tfrac32|Z^{-2}+o(Z^{-2})$,
which tends to zero both as $Z\to\infty$ for fixed $c$, and as $c\downarrow\tfrac32$ for fixed $Z$ since $|c-\tfrac32|\to0$. The second term tends to zero as $c\downarrow\tfrac32$ for any fixed $Z$, by
uniform convergence of $c\mapsto G_c(z)$ on the compact set $\{|z|\leq Z\}$, established in the pointwise continuity argument above. Hence $\sup_z|G_c(z)-\Phi(z)|\to0$ as $c\downarrow\tfrac32$, as required.
\end{proof}

\begin{proof}[Proof of Lemma~\ref{lem:differentiability}]
The series representation of ${}_1F_1$ is
\[
    {}_1F_1\!\left(\frac12;c;-\frac{z^2}{2}\right)
    = \sum_{k=0}^{\infty}\frac{(1/2)_k}{(c)_k\,k!}
    \left(-\frac{z^2}{2}\right)^k.
\]
Each term is analytic in $c$ on $(3/2,\infty)$, since $(c)_k=c(c+1)\cdots(c+k-1)$ is a polynomial in $c$ with no zeros there. The proof of Lemma~\ref{lem:continuity} establishes that the series converges absolutely and uniformly in $c$ on compact subsets of $(3/2,\infty)$ via the Weierstrass M-test with dominating series $\sum_{k=0}^\infty(1/2)_k/((c_L)_k\,k!)\cdot|z|^{2k}/2^k$. The $c$-derivatives of each term are bounded by the same dominating series up to a logarithmic factor in $k$ (from differentiating $(c)_k$ and using the asymptotic expansion $\psi(x) \sim \ln(x) - \tfrac{1}{2x} - \tfrac{1}{12x^2} + O(x^{-4})$ as $x \to \infty$ \citep[eq.~6.3.18]{abramowitz_stegun72}, so $\psi(c+k)-\psi(c)\sim\ln k$); since ${}_1F_1$ is entire, the terms of the original series decay faster than any geometric series in $k$, and this logarithmic growth is dominated and absolute convergence of the differentiated series is preserved for all $z\in\mathbb{R}$, justifying term-by-term differentiation to all orders and giving infinite differentiability of $c\mapsto{}_1F_1(1/2;c;-z^2/2)$ on $(3/2,\infty)$.

Differentiating $(c)_k=\prod_{j=0}^{k-1}(c+j)$ logarithmically with respect to $c$ gives
\[
    \frac{\partial}{\partial c}(c)_k
    = (c)_k\sum_{j=0}^{k-1}\frac{1}{c+j}
    = (c)_k\left[\psi(c+k)-\psi(c)\right],
\]
where $\psi=\Gamma'/\Gamma$ is the digamma function \citep[eq.~6.3.1]{abramowitz_stegun72} and the second equality follows from the digamma recurrence \citep[eq.~6.3.5]{abramowitz_stegun72} $\psi(c+1) = \psi(c) + \tfrac1c$ applied $k$ times. Hence,
\[
    \frac{\partial}{\partial c}\frac{(1/2)_k}{(c)_k\,k!}
    = -\frac{(1/2)_k}{(c)_k\,k!}\left[\psi(c+k)-\psi(c)\right],
\]
establishing \eqref{eq:1F1dot}. Differentiating $\kappa(c)=\exp(\ln\Gamma(c-\tfrac12)-\ln\Gamma(c)) /(2\sqrt{2})$ gives \eqref{eq:Kdot} by the definition of $\psi$. The product rule applied to $G_c(z)=\tfrac12+z\,\kappa(c)\,{}_1F_1(1/2;c;-z^2/2)$ then gives \eqref{eq:Gdot}; infinite differentiability of $c\mapsto G_c(z)$ follows since $\kappa(c)$ and ${}_1F_1(1/2;c;-z^2/2)$ are both infinitely differentiable in $c$ on $(3/2,\infty)$, with derivatives given by the polygamma functions $\psi^{(n)}(x)=\tfrac{d^n}{dx^n}\psi(x)$ \citep[eq.~6.4.1]{abramowitz_stegun72}.
\end{proof}

\begin{proof}[Proof of Lemma~\ref{lem:tail}]
Part (ii) is immediate from Corollary~\ref{cor:density}, which
gives $g_{3/2}=\phi$. For part (i), write $x=z^2/2\to+\infty$
and express $g_c$ from the proof of Lemma~\ref{lem:monotonicity} as
\[
    g_c(z) = \kappa(c)\left[{}_1F_1\!\left(\frac12;c;-x\right) -
    \frac{x}{c}\,{}_1F_1\!\left(\frac32;c+1;-x\right)\right].
\]
Applying the Kummer asymptotic expansion \citep[eq.~13.5.1]{abramowitz_stegun72}, as in the proof of Lemma~\ref{lem:continuity},
to each term, with $C:=\Gamma(c)/\Gamma(c-1/2)$, gives
\begin{align*}
    {}_1F_1\!\left(\frac12;c;-x\right)
    &=
    C\,x^{-1/2}\left[1+\frac{3/2-c}{2x} + O(x^{-2})\right],
    \\
    \frac{x}{c}\,{}_1F_1\!\left(\frac32;c+1;-x\right)
    &=
    C\,x^{-1/2}\left[1+\frac{3(3/2-c)}{2x} + O(x^{-2})\right].
\end{align*}
Subtracting, the leading $x^{-1/2}$ terms cancel, leaving
\[
    g_c(z) = \kappa(c)\,C\,x^{-1/2}\left[\frac{c-3/2}{x}+O(x^{-2})\right] = \frac{c-3/2}{2\sqrt{2}}\,x^{-3/2} + O(x^{-5/2}),
\]
using $\kappa(c)\,C=1/(2\sqrt{2})$. Since $x=z^2/2$, we have
$x^{-3/2}=2\sqrt{2}\,|z|^{-3}$ and $O(x^{-5/2})=o(|z|^{-3})$,
giving \eqref{eq:tail_algebraic}. Applying this argument directly to \eqref{eq:gc} gives the same result.
\end{proof}

\begin{proof}[Proof of Corollary~\ref{cor:survival}]
For $z>0$, $1-G_c(z)=\int_z^\infty g_c(t)\,dt$. By
Lemma~\ref{lem:tail}, $g_c(t)=(c-3/2)\,t^{-3}+o(t^{-3})$ as
$t\to\infty$. The leading term integrates to
\[
    \int_z^\infty (c-3/2)\,t^{-3}\,dt = \frac{c-3/2}{2}\,z^{-2}.
\]
For the remainder: given $\varepsilon>0$, choose $M$ such that
$|g_c(t)-(c-3/2)\,t^{-3}|\leq\varepsilon\, t^{-3}$ for all
$t\geq M$. Then for $z\geq M$,
\[
    \left|\int_z^\infty\!\left[g_c(t)-(c-3/2)\,t^{-3}\right]dt
    \right| \leq \varepsilon\int_z^\infty t^{-3}\,dt = \frac{\varepsilon}{2}\,z^{-2}.
\]
Since $\varepsilon>0$ was arbitrary, the remainder is
$o(z^{-2})$, giving \eqref{eq:survival}.
\end{proof}

\begin{proof}[Proof of Corollary~\ref{cor:moments}]
Since $g_c$ is even, $\E[|Z|^r]=2\int_0^\infty z^r g_c(z)\,dz$.
The integrand is continuous and bounded on $[0,M]$ for any
$M<\infty$, so convergence is determined by the tail integral
$\int_M^\infty z^r g_c(z)\,dz$. By Lemma~\ref{lem:tail},
$g_c(z)=(c-3/2)z^{-3}+o(z^{-3})$ as $z\to\infty$, so
$z^rg_c(z)/(c-3/2)z^{r-3}\to1$; by the limit comparison
test, the tail integral converges if and only if
$\int_M^\infty z^{r-3}\,dz<\infty$, if and only if
$r<2$. The case $r=1<2$ gives
$\E[|Z|]<\infty$, so $\E[Z]=0$ by symmetry; the case $r=2$
gives $\int_M^\infty z^{-1}\,dz=\infty$, so $\E[Z^2]=\infty$.
\end{proof}

\begin{proof}[Proof of Proposition~\ref{prop:mixture}]
Applying the change of variable $u=1-\lambda$ to the integral form \eqref{eq:Gc_integral} obtained in the proof of Lemma~\ref{lem:monotonicity} gives
\[
    g_c(z) = \frac{c-\tfrac32}{\sqrt{2\pi}}
    \int_0^1\lambda^{1/2}\,(1-\lambda)^{c-5/2}
    \,e^{-\lambda z^2/2}\,d\lambda.
\]
Identifying $(\sqrt{\lambda}/\sqrt{2\pi}) e^{-\lambda z^2/2}$ as the
$\mathrm{N}(0,\lambda^{-1})$ density and $(c-\tfrac32)(1-\lambda)^{c-5/2}$ as the $\mathrm{Beta}(1,c-\tfrac32)$ density gives \eqref{eq:mixture_density}, and the stochastic representation $Z=\varepsilon/\sqrt{\Lambda}$ follows immediately. Working directly from the ${}_1F_1$ series via Lemma~\ref{lem:monotonicity} avoids the need to justify differentiation under the integral sign that would arise from applying the Euler integral representation of ${}_1F_1$ to \eqref{eq:Gc_def} directly.
\end{proof}

\begin{proof}[Proof of Corollary~\ref{cor:unimodal}]
For $c=\tfrac32$, $g_{3/2}=\phi$ is strictly unimodal with unique mode at zero since $\phi(0)>\phi(z)$ for all $z\neq0$. For $c>\tfrac32$, \eqref{eq:mixture_density} expresses $g_c(z)$ as a mixture of $\mathrm{N}(0,\lambda^{-1})$ densities with strictly positive weights $(c-\tfrac32)(1-\lambda)^{c-5/2}$ on $(0,1)$. Each $\mathrm{N}(0,\lambda^{-1})$ density
attains its unique maximum at $z=0$ and is strictly decreasing in $|z|$; since the weights are strictly positive on $(0,1)$, the mixture inherits both properties: $g_c(z)<g_c(0)$ for all $z\neq0$, and $g_c$ is strictly decreasing on $(0,\infty)$.
\end{proof}

\begin{proof}[Proof of Corollary~\ref{cor:charfun}]
By Proposition~\ref{prop:mixture}, $Z=\varepsilon/\sqrt{\Lambda}$ with $\varepsilon\sim\N(0,1)$ and $\Lambda\sim\Beta(1,c-3/2)$ independent. Conditionally, $Z \mid (\Lambda=\lambda)\sim\N(0,\lambda^{-1})$, so $\E[e^{itZ}\mid(\Lambda=\lambda)]=e^{-t^2/(2\lambda)}$. The density of $\Lambda$ is $\lambda \mapsto (c-3/2)(1-\lambda)^{c-5/2}$ on $(0,1)$, giving
\[
    \varphi_c(t) = (c-3/2)\int_0^1 \exp\!\left(-\tfrac{t^2}{2\lambda}\right)\,
    (1-\lambda)^{c-5/2}\,d\lambda,
\]
which proves \eqref{eq:charfun_integral}. Setting
$x=t^2/2$ and substituting $y=(1-\lambda)/\lambda$ gives
\[
    \varphi_c(t) = (c-3/2)\,e^{-x} \int_0^\infty
    e^{-xy}\,y^{c-5/2}(1+y)^{1/2-c}\,dy.
\]
Applying the integral form of Tricomi's confluent hypergeometric function \citep[eq.~13.2.5]{abramowitz_stegun72} proves \eqref{eq:charfun}. The boundary case follows from $\Lambda\sim\Beta(1,c-3/2)\Rightarrow\delta_1$ as $c\downarrow \tfrac32$, so $Z=\varepsilon/\sqrt{\Lambda}\to\varepsilon\sim\N(0,1)$
and $\varphi_c(t)\to e^{-t^2/2}$ pointwise.
\end{proof}

\begin{proof}[Proof of Corollary~\ref{cor:location_scale_family}]
Since $G_c$ is a distribution function by Corollary~\ref{cor:Gc_distn} and
$x\mapsto(x-\mu)/\sigma$ is strictly increasing for $\sigma>0$, the composition $G_{\mu,\sigma,c}(x)=G_c\bigl((x-\mu)/\sigma\bigr)$ is a distribution function on $\mathbb{R}$. Differentiating gives
$g_{\mu,\sigma,c}(x)=\sigma^{-1}g_c\bigl((x-\mu)/\sigma\bigr)$. The Gaussian nesting follows from $G_{3/2}=\Phi$. Symmetry about $\mu$ follows from the evenness of $g_c$ (Corollary~\ref{cor:density}): for all
$t\in\mathbb{R}$, $g_{\mu,\sigma,c}(\mu+t)=\sigma^{-1}g_c(t/\sigma)
=\sigma^{-1}g_c(-t/\sigma)=g_{\mu,\sigma,c}(\mu-t)$. For the tail, apply Lemma~\ref{lem:tail} with $z=(x-\mu)/\sigma$: as $|x|\to\infty$,
\[
    g_{\mu,\sigma,c}(x) = \sigma^{-1}g_c\!\left(\tfrac{x-\mu}{\sigma}\right)
    \sim \sigma^{-1}\bigl(c-\tfrac32\bigr) \left|\tfrac{x-\mu}{\sigma}\right|^{-3} = \bigl(c-\tfrac32\bigr)\sigma^2|x-\mu|^{-3}.
\]
For the moments, let $Z=(X-\mu)/\sigma\sim G_c$, so $\E[|X-\mu|^r]=\sigma^r\E[|Z|^r]$. By Corollary~\ref{cor:moments}, $\E[|Z|^r]<\infty\Leftrightarrow r<2$, and the same equivalence holds for $\E[|X-\mu|^r]$. In particular $r=1<2$ gives $\E[|X-\mu|]<\infty$, and
$\E[X]=\mu$ follows by symmetry. The variance does not exist since $r=2$ is excluded.
\end{proof}

\begin{proof}[Proof of Theorem~\ref{thm:consistency}]
The proof has two steps: uniform convergence of $Q_n$ to $Q$
on $\mathcal{C}$, and passage from uniform convergence to
convergence of the minimiser. Since $\widehat F_n$, $F_0$, and $G_c$ all take values in $[0,1]$,
\[
    |Q_n(c)-Q(c)| \leq 2\int_{\mathbb{R}}|\widehat F_n(z)-F_0(z)|\,d\nu(z) \leq 2\,\nu(\mathbb{R})\sup_{z\in\mathbb{R}}
    |\widehat F_n(z)-F_0(z)|,
\]
where the second inequality uses \textup{(A2)}. Taking the
supremum over $c\in\mathcal{C}$,
\[
    \sup_{c\in\mathcal{C}}|Q_n(c)-Q(c)| \leq 2\,\nu(\mathbb{R})
    \sup_{z\in\mathbb{R}}|\widehat F_n(z)-F_0(z)|.
\]
By the Glivenko--Cantelli theorem, the right-hand side converges
to zero almost surely, so $Q_n\to Q$ uniformly on $\mathcal{C}$
almost surely. For each $z\in\mathbb{R}$, the map
$c\mapsto\{F_0(z)-G_c(z)\}^2$ is continuous by \textup{(A0)}
and bounded by $1$. Since $\nu(\mathbb{R})<\infty$ by
\textup{(A2)}, dominated convergence gives that $Q$ is
continuous on $\mathcal{C}$. Since $\mathcal{C}$ is compact by
\textup{(A1)} and $Q$ has a unique minimiser $c^\star$ by
\textup{(A3)}, uniform convergence of $Q_n$ to $Q$ implies
$\widehat c_n\to c^\star$ almost surely by the argmin theorem
\citep[Theorem~2.1]{newey_mcfadden94}.
\end{proof}

\begin{proof}[Proof of Lemma~\ref{lem:identifiability}]
The direction $c=c_0\Rightarrow G_c=G_{c_0}$ $\nu$-almost
everywhere is immediate. For the converse, let $c\neq c_0$
in $\mathcal{C}$, without loss of generality $c>c_0$.
By Corollary~\ref{cor:survival},
\[
    1-G_c(z)\sim\frac{c-3/2}{2}\,z^{-2}, \quad
    1-G_{c_0}(z)\sim\frac{c_0-3/2}{2}\,z^{-2},
    \quad z\to\infty,
\]
where the second asymptotic is read as $o(z^{-2})$ when
$c_0=3/2$, since $1-\Phi(z)$ decays super-exponentially.
In either case,
\[
    \frac{1-G_c(z)}{1-G_{c_0}(z)}\to
    \begin{cases}
        \dfrac{c-3/2}{c_0-3/2}\neq 1 & \text{if } c_0>3/2,\\[6pt]
        +\infty & \text{if } c_0=3/2,
    \end{cases}
\]
so $G_c(z)\neq G_{c_0}(z)$ for all sufficiently large $z$,
and by \textup{(A4)}, $\nu((M,\infty))>0$ for every $M$,
so $G_c$ and $G_{c_0}$ differ on a set of positive
$\nu$-measure and are not equal $\nu$-almost everywhere.
If $F_0=G_{c_0}$, then $Q(c_0)=0$ while $Q(c)>0$ for all
$c\neq c_0$, so $c^\star=c_0$ is the unique minimiser of $Q$
and \textup{(A3)} holds.
\end{proof}

\begin{proof}[Proof of Corollary~\ref{cor:consistency_correct}]
By Lemma~\ref{lem:identifiability}, \textup{(A4)} implies
\textup{(A3)} with $c^\star=c_0$. The conclusion then follows
from Theorem~\ref{thm:consistency} under \textup{(A0)--(A3)}.
The Gaussian null is the case $c_0=3/2$, using
$G_{3/2}=\Phi$ by Corollary~\ref{cor:density}.
\end{proof}

\begin{proof}[Proof of Lemma~\ref{lem:B1}]
Infinite differentiability of $c\mapsto G_c(z)$ on $(3/2,\infty)$ is established in Lemma~\ref{lem:differentiability}. It remains to show real-analyticity. Fix $z\in\mathbb{R}$ and let $\mathcal{N}$ be a neighbourhood of $c^\star$ in $(1/2,\infty)$. The function $K(c)=\Gamma(c-1/2)/(2\sqrt{2}\,\Gamma(c))$ is analytic on $\mathcal{N}$ since $\Gamma$ is analytic and non-vanishing on $(0,\infty)$. The map
$c\mapsto{}_1F_1(\frac12;c;-z^2/2)$ is meromorphic in $c$ with poles only at $c\in\mathbb{Z}_{0,-}$; since $\mathcal{N}\subset(1/2,\infty)$, no pole lies in $\mathcal{N}$, so it is analytic there. Hence $G_c(z)$, a sum and product of functions analytic in $c$, is real-analytic on $\mathcal{N}$.
\end{proof}

\begin{proof}[Proof of Lemma~\ref{lem:B2}]
Let $c_L=\inf_{c\in\mathcal{N}}c>1/2$. By the Euler integral representation
\citep[eq.~13.2.1]{abramowitz_stegun72},
\[
    G_c(z) = \frac12 + \frac{z}{2\sqrt{2\pi}}
    \int_0^1 e^{-z^2\lambda/2}\,\lambda^{-1/2}
    \,(1-\lambda)^{c-3/2}\,d\lambda.
\]
Since $(1-\lambda)^{c-3/2}\leq(1-\lambda)^{c_L-3/2}$ for $c\geq c_L$, the $c$-derivatives of the integrand are bounded uniformly in $c\in\mathcal{N}$ by
\[
    B_m(z) = \frac{|z|}{2\sqrt{2\pi}}\int_0^1
    e^{-z^2\lambda/2}\,\lambda^{-1/2}\,
    (1-\lambda)^{c_L-3/2}
    |\ln(1-\lambda)|^m\,d\lambda,
\]
which is finite for each fixed $z$: near $\lambda=0$ the integrand behaves like $\lambda^{m-1/2}$, which is integrable; near $\lambda=1$,
$(1-\lambda)^{c_L-3/2}|\ln(1-\lambda)|^m\to0$ since $c_L>1/2$ implies $c_L-3/2>-1$ and the power dominates the logarithm. Dominated convergence with dominator $B_m(z)$ therefore justifies differentiating under the
integral sign, giving
\begin{align*}
    \dot G_c(z)
    &= \frac{z}{2\sqrt{2\pi}} \int_0^1
    e^{-z^2\lambda/2}\,\lambda^{-1/2}
    \,(1-\lambda)^{c-3/2}\ln(1-\lambda)\,d\lambda,
    \\
    \ddot G_c(z)
    &= \frac{z}{2\sqrt{2\pi}} \int_0^1
    e^{-z^2\lambda/2}\,\lambda^{-1/2}
    \,(1-\lambda)^{c-3/2}\{\ln(1-\lambda)\}^2\,d\lambda,
\end{align*}
with $|\dot G_c(z)|\leq B_1(z)$ and $|\ddot G_c(z)|\leq B_2(z)$ uniformly in $c\in\mathcal{N}$. It remains to show that $\sup_z B_m(z)<\infty$. As $z\to0$, $B_m(z)\to0$ since $|z|\to0$ and the integral is bounded. As $|z|\to\infty$, $e^{-z^2\lambda/2}$ is negligible for all $\lambda$ bounded away from zero, so replacing $|\ln(1-\lambda)|^m\sim\lambda^m$ and
$(1-\lambda)^{c_L-3/2}\sim 1$ near $\lambda=0$, extending the upper limit to $\infty$, and using Euler's integral form of the Gamma function $\Gamma(x)=\int_0^\infty u^{x-1}e^{-u}du$ \citep[eq.~6.1.1]{abramowitz_stegun72}, it follows that
\[
    B_m(z) \sim \frac{|z|}{2\sqrt{2\pi}}\int_0^\infty
    \lambda^{m-1/2}\,e^{-z^2\lambda/2}\,d\lambda
    = \frac{|z|}{2\sqrt{2\pi}}
    \left(\frac{2}{z^2}\right)^{m+1/2}
    \Gamma\!\left(m+\tfrac12\right)
    = O(|z|^{-2m})\to0.
\]
Since $B_m$ is continuous, tends to zero at $z=0$ and as $|z|\to\infty$, it attains its supremum: $\sup_z B_m(z)\leq C_m<\infty$. Setting $M_1(z)=C_1$ and $M_2(z)=C_2$, both integrability conditions of \textup{(B2)} follow since $M_1$ and $M_2$ are constants and $\nu(\mathbb{R})<\infty$ by
\textup{(A2)}.
\end{proof}

\begin{proof}[Proof of Lemma~\ref{lem:B3}]
For positivity, note that $2\int_{\mathbb{R}}\dot
G_{c^\star}(z)^2\,d\nu(z)=0$ if and only if $\dot
G_{c^\star}=0$ $\nu$-almost everywhere. From
Lemma~\ref{lem:differentiability}, and expanding near $z=0$,
${}_1F_1(\tfrac12;c^\star;-z^2/2)=1+O(z^2)$ and
$\tfrac{\partial}{\partial c}\,{}_1F_1(\tfrac12;c^\star;
-z^2/2)=O(z^2)$, so
\[
    \dot G_{c^\star}(z) = z\dot \kappa(c^\star)+O(z^3),
\]
where $\dot \kappa(c^\star)=\kappa(c^\star)[\psi(c^\star-\tfrac12)
-\psi(c^\star)]<0$ since $\psi$ is strictly increasing on
$(0,\infty)$ by $\psi'(x)=\sum_{k=0}^{\infty}(x+k)^{-2}$
\citep[eq.~6.4.10]{abramowitz_stegun72}, and $c^\star-\tfrac12
<c^\star$. Hence $\dot G_{c^\star}$ is non-zero on a punctured
neighbourhood of zero, which has positive $\nu$-measure by
assumption, so $2\int_{\mathbb{R}}\dot G_{c^\star}(z)^2\,
d\nu(z)>0$. For the identification with $H_{c^\star}$: differentiating
$Q(c)=\int_{\mathbb{R}}\{F_0(z)-G_c(z)\}^2\,d\nu(z)$ twice
under the integral sign, justified by \textup{(B2)}, gives
\[
    Q''(c) = 2\int_{\mathbb{R}}\dot G_c(z)^2\,d\nu(z) -
    2\int_{\mathbb{R}}\{F_0(z)-G_c(z)\}\,\ddot G_c(z)\,
    d\nu(z).
\]
Under correct specification $F_0=G_{c^\star}$, the second
term vanishes, giving $H_{c^\star}=Q''(c^\star)=
2\int_{\mathbb{R}}\dot G_{c^\star}(z)^2\,d\nu(z)$, positive
by the above.
\end{proof}

\begin{proof}[Proof of Lemma~\ref{lem:Omega_pos}]
By construction (proof of Theorem~\ref{thm:interior}), $\Omega_{c^\star}=\Var(S)$ for $S=\int_{\mathbb{R}}
\mathbb{B}_{G_{c^\star}}(z)\,\dot G_{c^\star}(z)\,d\nu(z)$, a mean-zero Gaussian random variable; hence $\Omega_{c^\star} \geq0$ automatically. Equality would require $S=0$ almost surely, which (since $\mathbb{B}_{G_{c^\star}}$ is a non-degenerate Gaussian process whenever $G_{c^\star}$ has a strictly positive density) holds only if $\dot G_{c^\star}=0$ $\nu$-almost everywhere. By Lemma~\ref{lem:B3}, $\dot G_{c^\star}$ is non-zero on a punctured neighbourhood of zero, which has positive $\nu$-measure by assumption: so $\dot G_{c^\star}$ is not
$\nu$-almost everywhere zero, giving a contradiction. Hence $\Omega_{c^\star}>0$ as claimed.
\end{proof}

\begin{proof}[Proof of Theorem~\ref{thm:interior}]
Since $c^\star\in\operatorname{int}(\mathcal{C})$, consistency (Theorem~\ref{thm:consistency}) gives $\widehat c_n\in\operatorname{int}(\mathcal{C})$ with probability tending to one, so $Q_n'(\widehat c_n)=0$
eventually. A mean-value expansion about $c^\star$ gives
\begin{equation}\label{eq:mvt}
    \sqrt{n}\,(\widehat c_n-c^\star) = -\{Q_n''(\widetilde c_n)\}^{-1}
    \sqrt{n}\,Q_n'(c^\star),
\end{equation}
where $\widetilde c_n$ lies between $\widehat c_n$ and $c^\star$. By (B1)--(B2), differentiation under the integral sign gives
\[
    Q_n'(c^\star) = -2\int_{\mathbb{R}}
    \{\widehat F_n(z)-G_{c^\star}(z)\}
    \,\dot G_{c^\star}(z)\,d\nu(z).
\]
From the population first-order condition $Q'(c^\star)=0$ it follows that the score
\[
    \sqrt{n}\,Q_n'(c^\star) = -2\int_{\mathbb{R}}
    \sqrt{n}\,\{\widehat F_n(z)-F_0(z)\}
    \,\dot G_{c^\star}(z)\,d\nu(z).
\]
By Donsker's theorem \citep{donsker52}, $\sqrt{n}(\widehat F_n-F_0)
\Rightarrow\mathbb{B}_{F_0}$ in the space of bounded functions equipped with the supremum norm $\ell^\infty(\mathbb{R})$; $\mathbb{B}_{F_0}$ is an $F_0$-Brownian bridge, a mean-zero Gaussian process with covariance kernel $F_0(z\wedge y)-F_0(z)F_0(y)$, where $z\wedge y:=\min(z,y)$. The map $f\mapsto\int f\,\dot G_{c^\star}\,d\nu$ is a continuous linear functional on $\ell^\infty(\mathbb{R})$ under \textup{(B2)}, so the continuous mapping theorem gives $\sqrt{n}\,Q_n'(c^\star)\Rightarrow\N(0,\Omega_{c^\star})$, non-degenerate by Lemma~\ref{lem:Omega_pos}, with $\Omega_{c^\star}$ as in \eqref{eq:Omega}. By \textup{(B1)--(B2)} and the Glivenko--Cantelli theorem, $Q_n''(c)\to Q''(c)$ almost surely, uniformly on the neighbourhood $\mathcal{N}$ of $c^\star$ introduced in
\textup{(B1)}. Since $\widetilde c_n \to c^\star$ almost surely and $Q''$ is continuous, $Q_n''(\widetilde c_n)\to H_{c^\star}>0$ almost surely by
\textup{(B3)}, and hence in probability, which is all that
Slutsky's theorem requires. Slutsky's theorem applied to
\eqref{eq:mvt} then gives $\sqrt{n}(\widehat c_n-c^\star)
\Rightarrow\N(0,\Omega_{c^\star}/H_{c^\star}^2)$.
\end{proof}

\begin{proof}[Proof of Lemma~\ref{lem:curvature}]
Since $F_0=\Phi=G_{c_0}$, the Gaussian null is the correctly
specified case with $c^\star=c_0$, so \eqref{eq:Hc0} follows
from the identification in Lemma~\ref{lem:B3}, and positivity
from the bound in Lemma~\ref{lem:B3} applied with
$c^\star=c_0$. For \eqref{eq:Qn_uniform}, write
\[
    Q_n''(c)-Q''(c) = -2\int_{\mathbb{R}}
    \{\widehat F_n(z)-F_0(z)\}\,\ddot G_c(z)\,d\nu(z),
\]
so
\[
    \sup_{c\in\mathcal{N}}|Q_n''(c)-Q''(c)|\leq
    2\sup_{z \in \mathbb{R}}|\widehat F_n(z)-F_0(z)|
    \int_{\mathbb{R}}\sup_{c\in\mathcal{N}}|\ddot G_c(z)|\,d\nu(z).
\]
The integral is finite by (B2); $\sup_z|\widehat F_n(z)-F_0(z)|\to 0$ almost surely by Glivenko--Cantelli. The product converges to zero almost surely, which is stronger than \eqref{eq:Qn_uniform} requires: only convergence in probability is used in the proofs of Theorems~\ref{thm:interior} and~\ref{thm:boundary}.
\end{proof}

\begin{proof}[Proof of Lemma~\ref{lem:tightness}]
From the score argument of Theorem~\ref{thm:interior},  $\sqrt{n}\,Q_n'(c_0)\Rightarrow S\sim\N(0,\Omega_{c_0})$; in particular $\sqrt{n}\,Q_n'(c_0)=O_p(1)$. Since $Q''$ is continuous on a neighbourhood $\mathcal{N}$ of $c_0$ by \textup{(B1)--(B2)}, and $Q''(c_0)=H_{c_0}>0$ by Lemma~\ref{lem:curvature}, by shrinking $\mathcal N$ if necessary we may assume $Q''(c)\geq H_{c_0}/2$ for all $c\in\mathcal N$. By \eqref{eq:Qn_uniform}, with probability tending to one, $Q_n''(c)\geq H_{c_0}/4$ for all $c\in\mathcal N$. By Taylor's theorem with Lagrange remainder, for any $c\in \mathcal N$, with probability tending to one,
\[
    Q_n(c)-Q_n(c_0) \geq Q_n'(c_0)(c-c_0) +
    \frac{H_{c_0}}{8}(c-c_0)^2.
\]
By Corollary~\ref{cor:consistency_correct},
$\widehat c_n\to c_0$ almost surely, so $\widehat c_n\in \mathcal N$ with probability tending to one; since $\widehat c_n$ minimises $Q_n$ over $\mathcal C$ and $c_0 \in \mathcal{C}$, $Q_n(\widehat c_n)\leq Q_n(c_0)$. Setting $c=\widehat c_n$ above gives, with probability tending to one,
\[
    \frac{H_{c_0}}{8}(\widehat c_n-c_0)^2 \leq
    -Q_n'(c_0)(\widehat c_n-c_0) \leq
    |Q_n'(c_0)|\,|\widehat c_n-c_0|.
\]
When $\widehat c_n=c_0$ the bound is trivial; when
$\widehat c_n>c_0$, dividing by $\widehat c_n-c_0>0$ gives
$|\widehat c_n-c_0|\leq(8/H_{c_0})\,|Q_n'(c_0)|$, and hence
\[
    \sqrt{n}\,|\widehat c_n-c_0| \leq \frac{8}{H_{c_0}}\,
    \sqrt{n}\,|Q_n'(c_0)| = O_p(1),
\]
and the claim follows.
\end{proof}

\begin{proof}[Proof of Theorem~\ref{thm:boundary}]
Reparametrise locally by $c=c_0+h/\sqrt{n}$ with $h\geq 0$,
and define the localised criterion
\[
    L_n(h) = n\bigl\{Q_n(c_0+h/\sqrt{n})-Q_n(c_0)\bigr\}.
\]
By Taylor's theorem with Lagrange remainder, with the
remainder term controlled uniformly for $h$ in compact sets
by Lemma~\ref{lem:curvature},
\begin{equation}\label{eq:taylor_boundary}
    L_n(h) = h\,\sqrt{n}\,Q_n'(c_0) + \tfrac12 h^2\,Q_n''(c_0) + o_p(1).
\end{equation}
By the score argument of Theorem~\ref{thm:interior},
$\sqrt{n}\,Q_n'(c_0)\Rightarrow S\sim\N(0,\Omega_{c_0})$, non-degenerate by Lemma~\ref{lem:Omega_pos}, where $\Omega_{c_0}$ is the correct-specification form of \eqref{eq:Omega}. By Lemma~\ref{lem:curvature}, $Q_n''(c_0)\to_p H_{c_0}>0$, with $H_{c_0}$ as in \eqref{eq:Hc0}. Substituting into \eqref{eq:taylor_boundary}, $L_n\Rightarrow L$ uniformly on compact sets, where
\[
    L(h) = hS + \tfrac12 H_{c_0}h^2, \quad h\geq 0.
\]
Since $H_{c_0}>0$, $L$ is strictly convex on $[0,\infty)$ with unique minimiser
\[
    h^\star = \arg\min_{h\geq 0}L(h) =
    \max\!\left\{0,\,-\frac{S}{H_{c_0}}\right\}.
\]
The conditions of the argmax continuous mapping theorem
\citep[Lemma~3.2.1]{vandervaart_wellner23} are now in place:
$L_n\Rightarrow L$ uniformly on compact sets and $L$ has the
unique minimiser $h^\star$, while $\sqrt{n}(\widehat c_n-c_0)$, the exact minimiser of $L_n$ over $h\geq 0$, is asymptotically tight by
Lemma~\ref{lem:tightness}. Applying the theorem to $-L_n$ and $-L$ gives
\[
    \sqrt{n}\,(\widehat c_n-c_0)\;\Rightarrow\;
    h^\star = \max\!\left\{0,\,-\frac{S}{H_{c_0}}\right\}.
\]
Writing $Z=-S/H_{c_0}\sim\N(0,\Omega_{c_0}/H_{c_0}^2)$ gives
$h^\star=\max\{0,Z\}$. Since $Z$ is symmetric about zero,
$\Pr(Z\leq0)=1/2$, so the limit places mass $1/2$ at $\{0\}$;
and since $\max\{0,Z\}$ restricted to $(0,\infty)$ coincides
in distribution with $Z$ restricted to $(0,\infty)$, the
remaining mass is half-normal.
\end{proof}

\begin{proof}[Proof of Theorem~\ref{thm:fit_improvement}]
By the localisation argument in the proof of
Theorem~\ref{thm:boundary}, with the localised criterion
$L_n(h)=n\{Q_n(c_0+h/\sqrt{n})-Q_n(c_0)\}$,
\[
    L_n\;\Rightarrow\; L(h)=hS+\tfrac12 H_{c_0}h^2,
    \quad h\geq 0,
\]
uniformly on compact sets, where $S\sim\N(0,\Omega_{c_0})$.
Since $c=c_0+h/\sqrt{n}$ is an increasing reparametrisation
and $\widehat c_n$ minimises $Q_n$ over $c\geq c_0$,
$h_n^\star:=\sqrt{n}(\widehat c_n-c_0)$ minimises $L_n$ over
$h\geq 0$, and
\[
    T_n = n\{Q_n(c_0)-Q_n(\widehat c_n)\} =
    -L_n(h_n^\star) = -\min_{h\geq 0}L_n(h)
\]
exactly. By Lemma~\ref{lem:tightness}, $h_n^\star=O_p(1)$, so
for any $\varepsilon>0$ there is an $M<\infty$, taken large
enough that also $h^\star\leq M$, with $\Pr(h_n^\star\leq M)>1-\varepsilon$ for $n$ large. On this event, $\min_{h\geq 0}L_n(h)=\min_{h\in[0,M]}L_n(h)$;
and since $M$ was also chosen so that $h^\star\leq M$,
$\min_{h\in[0,M]}L(h)=\min_{h\geq 0}L(h)$ holds always, not
just with probability tending to one. Since $L_n\Rightarrow L$ uniformly on $[0,M]$,
\[
    \Bigl|\min_{h\in[0,M]}L_n(h)-\min_{h\in[0,M]}L(h)\Bigr|
    \leq \sup_{h\in[0,M]}|L_n(h)-L(h)| \to_p 0,
\]
so $T_n\Rightarrow-\min_{h\geq 0}L(h)$. The strictly convex limit $L$ has minimiser $h^\star=\max\{0,-S/H_{c_0}\}$, as in the proof of
Theorem~\ref{thm:boundary}. If $S\geq 0$, the minimum is at $h^\star=0$ and the limiting improvement is zero. If $S<0$, the minimum is at $h^\star=-S/H_{c_0}>0$, giving
\[
    \min_{h\geq 0}L(h) = -\frac{S}{H_{c_0}}\cdot S +
    \frac{H_{c_0}}{2}\cdot\frac{S^2}{H_{c_0}^2} =
    -\frac{S^2}{2H_{c_0}}.
\]
Hence
\[
    T_n \;\Rightarrow\; \frac{S^2}{2H_{c_0}}\,\mathbf{1}\{S<0\}.
\]
Writing $S=\Omega_{c_0}^{1/2}\,U$ with $U\sim\N(0,1)$,
\[
    T_n \;\Rightarrow\; \frac{\Omega_{c_0}}{2H_{c_0}}\,
    U^2\,\mathbf{1}\{U<0\}.
\]
Since $\Pr(U<0)=1/2$ and $U^2\sim\chi_1^2$, the limit
distributes as $\frac12\delta_0+\frac12(\Omega_{c_0}/
2H_{c_0})\chi_1^2$, establishing \eqref{eq:Tn_limit}.
Multiplying by $2H_{c_0}/\Omega_{c_0}$ gives
\eqref{eq:Tntilde_limit}, which is well-defined since $\Omega_{c_0}>0$ by Lemma~\ref{lem:Omega_pos}.
\end{proof}

\begin{proof}[Proof of Corollary~\ref{cor:moments_Ttilde}]
Write $\widetilde T_\infty=BX$ where $X\sim\chi_1^2$ and
$B\sim\mathrm{Bernoulli}(1/2)$ are independent. Since
$\E[X]=1$ and $\E[X^2]=3$,
\[
    \E[\widetilde T_\infty] = \E[B]\E[X] = \frac12, \quad
    \E[\widetilde T_\infty^2] = \E[B]\E[X^2] = \frac32,
\]
giving $\Var(\widetilde T_\infty)=\frac32-\frac14=\frac54$.
\end{proof}

\begin{proof}[Proof of Theorem~\ref{thm:directed}]
Uniform convergence of $Q_n$ to $Q$ on $\mathcal{C}$,
established in the proof of Theorem~\ref{thm:consistency},
gives
\[
    Q_n(c_0)\to_p Q(c_0), \quad Q_n(\widehat c_n) =
    \inf_{c\in\mathcal{C}}Q_n(c) \to_p
    \inf_{c\in\mathcal{C}}Q(c) = Q(c^\star).
\]
Hence $Q_n(c_0)-Q_n(\widehat c_n)\to_p Q(c_0)-Q(c^\star)
=\Delta$. Writing $R_n:=Q_n(c_0)-Q_n(\widehat c_n)-\Delta
\to_p 0$, for any fixed $M$, eventually $M/n-\Delta<-\Delta/2$,
so $\{n(\Delta+R_n)\leq M\}\subseteq\{R_n\leq-\Delta/2\}$, and
$\Pr(R_n\leq-\Delta/2)\to0$; hence $n\{Q_n(c_0)-Q_n(\widehat
c_n)\}\to_p\infty$ whenever $\Delta>0$. The constants $H_{c_0}=2\int_\mathbb{R}\dot G_{c_0}(z)^2\,d\nu(z)$ and $\Omega_{c_0}$ in \eqref{eq:Omega_c0} are determined entirely by $\Phi$, $\dot G_{c_0}$, and $\nu$: the same fixed quantities appearing in Theorem~\ref{thm:boundary}; and so remain well-defined, finite, and strictly positive regardless of the true $F_0$. Positivity of $H_{c_0}$ is exactly Lemma~\ref{lem:B3} applied with $c^\star=c_0$, which carries no hypothesis on $F_0$. Hence $2H_{c_0}/\Omega_{c_0}$ is a fixed, finite, strictly positive constant, and $\widetilde T_n\to_p\infty$, giving
$\Pr_{F_0}(\widetilde T_n>k_\alpha)\to 1$ for any fixed $k_\alpha<\infty$.
\end{proof}

\begin{proof}[Proof of Theorem~\ref{thm:local_power}]
The proof modifies the boundary argument of Theorem~\ref{thm:fit_improvement} to account for the local drift. The score at $c_0$ is
\[
    \sqrt{n}\,Q_n'(c_0) = -2\int_{\mathbb{R}}
    \sqrt{n}\,\{\widehat F_n(z)-\Phi(z)\}
    \,\dot G_{c_0}(z)\,d\nu(z).
\]
Since $f\mapsto-2\int f\,\dot G_{c_0}\,d\nu$ is a continuous linear functional on $\ell^\infty(\mathbb{R})$ under \textup{(B2)} (as in the proof of Theorem~\ref{thm:interior}), the assumed limit $\sqrt{n}(\widehat F_n-\Phi)\Rightarrow \mathbb{B}_\Phi+\eta$ and the continuous mapping theorem give
\[
    \sqrt{n}\,Q_n'(c_0)\;\Rightarrow\;S_\eta :=
    -2\int_{\mathbb{R}}\{\mathbb{B}_\Phi(z)+\eta(z)\}
    \,\dot G_{c_0}(z)\,d\nu(z) = S_0 - 2A_\nu(\eta),
\]
where $S_0:=-2\int_\mathbb{R}\mathbb{B}_\Phi(z)\,\dot G_{c_0}(z)\, d\nu(z)$. Since $\mathbb{B}_\Phi$ has covariance kernel $\Phi(z\wedge y)-\Phi(z)\Phi(y)$, $S_0$ is Gaussian with mean zero and variance $\Omega_{c_0}$ (matching
\eqref{eq:Omega_c0}), positive by Lemma~\ref{lem:Omega_pos}. We write $H_{c_0}(\nu)$ and $\Omega_{c_0}(\nu)$ to make explicit the dependence on $\nu$ suppressed elsewhere. Hence $S_\eta\sim\N(-2A_\nu(\eta),\Omega_{c_0}(\nu))$, and writing $S_\eta=\sqrt{\Omega_{c_0}(\nu)}\,(Z-\delta_\nu(\eta))$ with $Z\sim\N(0,1)$ recovers \eqref{eq:delta}. Writing $\widehat F_n-\Phi=(\widehat F_n-F_n)+(F_n-\Phi)$: the second term is deterministic and $o(1)$ by \eqref{eq:local_alt}; the first is $O_p(n^{-1/2})$, since $\sqrt{n}(\widehat F_n-F_n)\Rightarrow \mathbb{B}_\Phi$ follows by subtracting the deterministic limit $\sqrt{n}(F_n-\Phi)\to\eta$ from $\sqrt{n}(\widehat F_n-\Phi)\Rightarrow \mathbb{B}_\Phi+\eta$. Hence $\sup_z|\widehat F_n(z)-\Phi(z)|\to_p0$ as under the fixed null, and the argument of Lemma~\ref{lem:curvature}
gives $Q_n''(c_0)\to_p H_{c_0}(\nu)$, the same constant as under
$\Phi$.

Contiguity of $\{F_n\}$ to $\Phi$ holds under the same regularity underlying the assumed local empirical process limit. Since $\sqrt{n}(\widehat c_n-c_0)=O_p(1)$ under $\Phi$ (from Lemma~\ref{lem:tightness}), contiguity transfers this to $O_p(1)$ under $F_n$ \citep[Ch.~6]{vandervaart98}. With the score limit $S_\eta$ in place of $S$, $H_{c_0}(\nu)$ unchanged, and tightness confirmed, the localisation, Taylor expansion, and argmax continuous mapping argument of Theorem~\ref{thm:fit_improvement} apply directly to give
\[
    \widetilde T_n\;\Rightarrow\;\frac{S_\eta^2}{\Omega_{c_0}(\nu)}
    \mathbf{1}\{S_\eta<0\} =
    (Z-\delta_\nu(\eta))^2\,\mathbf{1}\{Z-\delta_\nu(\eta)<0\}, \quad Z \sim \N(0,1),
\]
which is \eqref{eq:Tn_local}. The rejection probability
\eqref{eq:power} follows since $\Pr\{Z-\delta_\nu(\eta)<-\sqrt{k_\alpha}\}=
\Phi\{\delta_\nu(\eta)-\sqrt{k_\alpha}\}$ and $\sqrt{k_\alpha}=\Phi^{-1}(1-\alpha)$.
\end{proof}

\begin{proof}[Proof of Corollary~\ref{cor:local_hgm}]
The expansion $G_{c_0+\tau/\sqrt{n}}(z)=G_{c_0}(z)+(\tau/\sqrt{n})\,\dot G_{c_0}(z)+o(n^{-1/2})$ gives $\eta(z)=\tau\,\dot G_{c_0}(z)$, so $A_\nu(\eta)=\tau\int_{\mathbb{R}}\dot G_{c_0}(z)^2\,d\nu(z)=(\tau/2)H_{c_0}(\nu)$. Substituting into \eqref{eq:delta} and \eqref{eq:power} gives the result.
\end{proof}

\begin{proof}[Proof of Lemma~\ref{lem:plugin_ep}]
Let $W_i=(X_i-\mu_0)/\sigma_0\stackrel{\text{i.i.d}}{\sim}\N(0,1)$, so $Y_i^P=(W_i-\bar W_n)/\widehat s_n$, where $\bar W_n=n^{-1} \sum_{i=1}^n W_i$ and $\widehat s_n^2=n^{-1}\sum_{i=1}^n(W_i-\bar W_n)^2$, and $\widehat F_n^P(z)=\widehat F_n^W(\bar W_n+\widehat s_n z)$. We apply \citet[corollary~1]{durbin73} to the normal location-scale family
\[
    F(w,\theta)=\Phi\left(\frac{w-\mu}{s}\right),
    \quad
    \theta=(\mu,s)' ,
\]
with true value $\theta_0=(0,1)'$ and estimator $\widehat\theta_n=(\bar W_n,\widehat s_n)'$. For $t=\Phi(z)$,
\[
    F(W_i,\widehat\theta_n)\leq t
    \quad\Longleftrightarrow\quad
    \frac{W_i-\bar W_n}{\widehat s_n}\leq z
    \quad\Longleftrightarrow\quad
    Y_i^P\leq z.
\]
Thus Durbin's estimated empirical process satisfies
\[
    \widehat y_n(\Phi(z)) = \sqrt n(\widehat F_n^P(z)-\Phi(z)).
\]
The estimator $\widehat\theta_n$ satisfies Durbin's asymptotic linearity condition, since
\[
    \sqrt n
    \begin{pmatrix}
        \bar W_n\\
        \widehat s_n-1
    \end{pmatrix}
    =
    \frac1{\sqrt n}\sum_{i=1}^n
    \ell(W_i)+o_p(1),
    \quad
    \ell(w)=
    \begin{pmatrix}
        w\\
        \frac12(w^2-1)
    \end{pmatrix}.
\]
Moreover,
\[
    L:=\E[\ell(W)\ell(W)']
    =
    \begin{pmatrix}
        1 & 0\\
        0 & 1/2
    \end{pmatrix}.
\]
For $t=\Phi(z)$, the $t$-quantile under $\theta_0$ is $z$. Therefore Durbin's derivative vector is
\[
    g(\Phi(z)) =
    \left.
    \frac{\partial F(w,\theta)}{\partial\theta}
    \right|_{w=z,\theta=\theta_0}
    =
    \begin{pmatrix}
        -\phi(z)\\
        -z\phi(z)
    \end{pmatrix}
\]
and Durbin's $h$-function is
\[
    h(\Phi(z)) = \int_{-\infty}^z \ell(w)\phi(w)\,dw =
    \begin{pmatrix}
        \int_{-\infty}^z w\phi(w)\,dw\\[3pt]
        \frac12\int_{-\infty}^z (w^2-1)\phi(w)\,dw
    \end{pmatrix}
    =
    \begin{pmatrix}
        -\phi(z)\\
        -\frac12 z\phi(z)
    \end{pmatrix}.
\]
By \citet[corollary~1]{durbin73}, under the null the process
$\widehat y_n(t)$ converges weakly to a mean-zero Gaussian process with covariance
\[
    \min(t_1,t_2) - t_1t_2 - h(t_1)'g(t_2) - h(t_2)'g(t_1) + g(t_1)'Lg(t_2).
\]
Putting $t_1=\Phi(z)$ and $t_2=\Phi(y)$ gives
\[
    K_P(z,y) = \Phi(z\wedge y)-\Phi(z)\Phi(y) 
    -\phi(z)\phi(y)-\frac12 zy\,\phi(z)\phi(y).
\]
Therefore
\[
    \sqrt n\{\widehat F_n^P(z)-\Phi(z)\}\Rightarrow\mathbb Z_P(z)
\]
in $\ell^\infty(\mathbb R)$, where $\mathbb Z_P$ is the mean-zero Gaussian process with covariance kernel $K_P$, proving \eqref{eq:ZP} and \eqref{eq:KP}. The ordinary empirical process satisfies
\[
    \sqrt n(\widehat F_n^W(z)-\Phi(z)) \Rightarrow \mathbb B_\Phi(z),
\]
jointly with
\[
    \sqrt n\,\bar W_n\Rightarrow Z_1,\quad
    \sqrt n(\widehat s_n-1)\Rightarrow Z_2,
\]
where
\[
    Z_1\sim\N(0,1), \quad Z_2\sim\N(0,1/2), \quad \Cov(Z_1,Z_2)=0.
\]
The cross-covariances are
\[
    \Cov(\mathbb B_\Phi(z),Z_1) =
    \E[(\mathbf{1}\{W\leq z\}-\Phi(z))W] =
    \int_{-\infty}^z w\phi(w)\,dw = -\phi(z),
\]
and
\[
    \Cov(\mathbb B_\Phi(z),Z_2) =
    \frac12\E[(\mathbf{1}\{W\leq z\}-\Phi(z))(W^2-1)] =
    \frac12\int_{-\infty}^z (w^2-1)\phi(w)\,dw = -\frac12 z\phi(z).
\]
Consequently, the Gaussian process with kernel $K_P$ admits the representation
\[
    \mathbb Z_P(z) = \mathbb B_\Phi(z)+\phi(z)Z_1+z\phi(z)Z_2,
\]
with the stated joint Gaussian structure, as claimed.
\end{proof}

\begin{proof}[Proof of Lemma~\ref{lem:robust_ep}]
Let $W_i=(X_i-\mu_0)/\sigma_0\stackrel{\text{i.i.d}}{\sim}\N(0,1)$, so $Y_i^R=(W_i-\widetilde \mu_n^R)/\widetilde \sigma_n^R$, where $\widetilde\mu_n^R=\widehat q_n^W(1/2)$ and $\widetilde\sigma_n^R = (\widehat q_n^W(3/4)-\widehat q_n^W(1/4))/2a$, and $\widehat F_n^R(z) = \widehat F_n^W(\widetilde\mu_n^R+\widetilde\sigma_n^R z)$. The Bahadur representation for sample quantiles gives, for $p\in\{1/4,1/2,3/4\}$,
\[
    \sqrt n(\widehat q_n^W(p)-q_p) = \frac1{\sqrt n}\sum_{i=1}^n
    \frac{p-\mathbf{1}\{W_i\leq q_p\}}{\phi(q_p)} +o_p(1),
\]
where $q_p=\Phi^{-1}(p)$; see \citet{bahadur66} and
\citet[Thm.~7.3]{dasgupta08}. Since $q_{1/2}=0$, $q_{3/4}=a$, and $q_{1/4}=-a$, it follows that
\[
    \sqrt n\,\widetilde\mu_n^R =
    \frac1{\sqrt n}\sum_{i=1}^n \IF_m(W_i)+o_p(1), \quad \sqrt n(\widetilde\sigma_n^R-1) =
    \frac1{\sqrt n}\sum_{i=1}^n \IF_s(W_i)+o_p(1).
\]
Thus the estimator $\widehat\theta_n^R = (\widetilde\mu_n^R,\widetilde\sigma_n^R)'$ satisfies Durbin's asymptotic linearity condition with influence vector $\ell_R(w) = (\IF_m(w), \IF_s(w))'$. We now apply \citet[corollary~1]{durbin73} to the normal
location-scale family
\[
    F(w,\theta)=\Phi\left(\frac{w-\mu}{s}\right), \quad \theta=(\mu,s)',
\]
with true value $\theta_0=(0,1)'$ and estimator $\widehat\theta_n^R$. For $t=\Phi(z)$,
\[
    F(W_i,\widehat\theta_n^R)\leq t
    \quad\Longleftrightarrow\quad
    \frac{W_i-\widetilde\mu_n^R}{\widetilde\sigma_n^R}\leq z
    \quad\Longleftrightarrow\quad
    Y_i^R\leq z.
\]
It follows that Durbin's estimated empirical process satisfies
\[
    \widehat y_n(\Phi(z)) = \sqrt n(\widehat F_n^R(z)-\Phi(z)).
\]
For $t=\Phi(z)$, the $t$-quantile under $\theta_0$ is $z$, and Durbin's derivative vector is
\[
    g(\Phi(z)) = \left.
    \frac{\partial F(w,\theta)}{\partial\theta}
    \right|_{w=z,\theta=\theta_0} =
    \begin{pmatrix}
        -\phi(z)\\
        -z\phi(z)
    \end{pmatrix}.
\]
Moreover,
\[
    L_R := \E[\ell_R(W)\ell_R(W)'] = \E
    \begin{bmatrix}
        \IF_m(W)^2 & \IF_m(W)\IF_s(W)\\
        \IF_m(W)\IF_s(W) & \IF_s(W)^2
    \end{bmatrix},
\]
and Durbin's $h$-function is
\[
    h_R(\Phi(z)) =
    \int_{-\infty}^z \ell_R(w)\phi(w)\,dw =
    \begin{pmatrix}
        \E[\mathbf{1}\{W\leq z\}\IF_m(W)]\\
        \E[\mathbf{1}\{W\leq z\}\IF_s(W)]
    \end{pmatrix}.
\]
Since $\E[\IF_m(W)]=\E[\IF_s(W)]=0$, this is equivalently
\[
    h_R(\Phi(z)) =
    \begin{pmatrix}
        \E[(\mathbf{1}\{W\leq z\}-\Phi(z))\IF_m(W)]\\
        \E[(\mathbf{1}\{W\leq z\}-\Phi(z))\IF_s(W)]
    \end{pmatrix}.
\]
By \citet[corollary~1]{durbin73} we therefore have weak convergence in
$\ell^\infty(\mathbb R)$ of
\[
    \sqrt n(\widehat F_n^R(z)-\Phi(z))
\]
to a mean-zero Gaussian process. Equivalently, this process admits the
 representation
\[
    \mathbb Z_R(z) = \mathbb B_\Phi(z) + \phi(z)Z_m + z\phi(z)Z_s,
\]
where $(\mathbb B_\Phi,Z_m,Z_s)$ is jointly Gaussian and $(Z_m,Z_s)' \sim \N_2(0,L_R)$. The cross-covariances with the Brownian bridge are precisely the components of $h_R(\Phi(z))$, namely
\[
    \Cov(\mathbb B_\Phi(z),Z_m) =
    \E[(\mathbf{1}\{W\leq z\}-\Phi(z))\IF_m(W)],
\]
and
\[
    \Cov(\mathbb B_\Phi(z),Z_s) =
    \E[(\mathbf{1}\{W\leq z\}-\Phi(z))\IF_s(W)].
\]
Finally, collecting the first-order summand gives
\[
    \xi_z(W) = \mathbf{1}\{W\leq z\}-\Phi(z) + \phi(z)\IF_m(W)
    + z\phi(z)\IF_s(W).
\]
Thus
\[
    K_R(z,y) = \Cov(\mathbb Z_R(z),\mathbb Z_R(y)) =
    \E[\xi_z(W)\xi_y(W)],
\]
because $\E[\xi_z(W)]=0$. This proves the claim.
\end{proof}

\begin{proof}[Proof of Theorem~\ref{thm:plugin_fit_improvement}]
The argument is identical to Theorem~\ref{thm:fit_improvement}
with the Brownian bridge score limit $S\sim\N(0,\Omega_{c_0})$
replaced by the projected score limit $S^\diamond\sim\N(0,
\Omega_{c_0}^\diamond)$, which follows from
Lemma~\ref{lem:plugin_ep} or Lemma \ref{lem:robust_ep} and the
continuous mapping theorem. The curvature $Q_n^{\diamond\prime
\prime}(c_0)\to_p H_{c_0}$ by the same uniform convergence
argument as Lemma~\ref{lem:curvature}, with $\widehat F_n$
replaced by $\widehat F_n^\diamond$.
\end{proof}

\begin{proof}[Proof of Corollary~\ref{cor:plugin_moments}]
The standardised statistic $\widetilde T_n^\diamond$ has the same
limiting law $\frac12\delta_0+\frac12\chi_1^2$ for each
$\diamond\in\{P,R\}$ by Theorem~\ref{thm:plugin_fit_improvement}.
The moment calculation is identical to
Corollary~\ref{cor:moments_Ttilde}.
\end{proof}

\begin{proof}[Proof of Corollary~\ref{cor:plugin_local_power}]
The argument is identical to Theorem~\ref{thm:local_power} with the Brownian bridge kernel $K$ replaced by $K_\diamond$ and $\Omega_{c_0}$ replaced by $\Omega_{c_0}^\diamond$. The score limit under the local alternative acquires the drift $-2A_\nu^\diamond(\eta_\diamond)$ as in the fixed-scale case, giving $\delta_\nu^\diamond(\eta_\diamond)$ and $\pi_\nu^\diamond(\eta_\diamond)$ as stated.
\end{proof}

\begin{proof}[Proof of Corollary~\ref{cor:plugin_consistency}]
Strong consistency of the standardising estimators and Glivenko--Cantelli imply
\[
    \sup_z|\widehat F_n^\diamond(z)-F_0^\diamond(z)|\to0
\]
almost surely. Thus $Q_n^\diamond$ converges uniformly to $Q^\diamond$ by the same argument as in Theorem~\ref{thm:consistency}. The argmin theorem then gives $\widehat c_n^\diamond\to c_\diamond^\star$ almost surely. Under
the Gaussian location-scale null, standardisation gives $F_0^\diamond=\Phi=G_{c_0}$, so $c_\diamond^\star=c_0$ by
Lemma~\ref{lem:identifiability}.
\end{proof}

\begin{proof}[Proof of Corollary~\ref{cor:plugin_directed}] By the uniform convergence established in Corollary~\ref{cor:plugin_consistency}, \[ Q_n^\diamond(c_0)-Q_n^\diamond(\widehat c_n^\diamond) \to_p Q^\diamond(c_0)-Q^\diamond(c_\diamond^\star) = \Delta^\diamond . \] If $\Delta^\diamond>0$, the $n$-scaled improvement diverges in probability. Since $2H_{c_0}/\Omega_{c_0}^\diamond$ is finite and strictly positive, $\widetilde T_n^\diamond\to_p\infty$, and the rejection probability tends to one for any fixed $k_\alpha$. \end{proof}

\begin{proof}[Proof of Lemma~\ref{lem:R_monotone}]

\textit{Monotonicity in $c$.}
By the beta-precision normal scale mixture representation of Proposition~\ref{prop:mixture},
\[
    1-G_c(z) = \E\bigl[1-\Phi(z\sqrt{\Lambda})\bigr],
    \quad \Lambda\sim\mathrm{Beta}(1,c-3/2),
\]
and $1-\Phi(z)$ does not depend on $c$. Throughout this proof, $\E[\cdot]$ denotes expectation with respect to $\Lambda$. The $\mathrm{Beta}(1,c-3/2)$ distribution is stochastically decreasing in $c$ (larger $c$ shifts mass toward smaller $\lambda$), and $\lambda\mapsto1-\Phi(z\sqrt{\lambda})$ is decreasing in $\lambda$ for $z>0$. Therefore $\E[1-\Phi(z\sqrt{\Lambda})]$ is strictly increasing in $c$, equivalently $\dot G_c(z)<0$ for all $z>0$ and $c>3/2$, giving $\partial \rho(z,c)/\partial c>0$.

\textit{Monotonicity in $z$.}
Write $g_c(z)=\E[\phi(z\sqrt{\Lambda})\sqrt{\Lambda}]$. Then $\rho'(z)>0$ if and only if
\[
    \E\bigl[\phi(z)\{1-\Phi(z\sqrt{\Lambda})\}
    -\{1-\Phi(z)\}\phi(z\sqrt{\Lambda})\sqrt{\Lambda}\bigr]>0.
\]
Setting $u=z\sqrt{\lambda}$ so that $\phi(z\sqrt{\lambda})\sqrt{\lambda}=\phi(u)\cdot u/z$, the integrand is strictly positive for each $\lambda\in(0,1)$ if and only if
\[
    z\cdot\frac{\phi(z)}{1-\Phi(z)} >
    u\cdot\frac{\phi(u)}{1-\Phi(u)},
    \quad u=z\sqrt{\lambda}<z,
\]
that is, if $s\mapsto s\,\phi(s)/(1-\Phi(s))$ is strictly increasing, where $\phi(s)/(1-\Phi(s))$ is the Gaussian hazard rate. Differentiating,
\[
    \frac{d}{ds}
    \!\left[\frac{s\,\phi(s)}{1-\Phi(s)}\right] =
    \frac{\phi(s)\bigl[(1-s^2)(1-\Phi(s))+s\phi(s)\bigr]}
    {(1-\Phi(s))^2}.
\]
For $s\in(0,1]$ both terms in the numerator bracket are non-negative and $s\phi(s)>0$. For $s>1$, Mill's inequality
$1-\Phi(s)<\phi(s)/s$ gives $(s^2-1)(1-\Phi(s))<(s^2-1)\phi(s)/s<s\phi(s)$,
so the bracket is again strictly positive. Hence $s\,\phi(s)/(1-\Phi(s))$ is strictly increasing for all $s>0$. Since $\Lambda\in(0,1)$ almost surely, the integrand is strictly positive almost surely, giving $\rho'(z)>0$ as required.
\end{proof}

\begin{proof}[Proof of Proposition~\ref{prop:var_divergence}]
From Corollary~\ref{cor:survival}, $1-G_c(z)\sim(c-3/2)/(2z^2)$ as $z\to\infty$, giving $q_p^{\mathrm{HGM}}\sim[(c-3/2)/2]^{1/2}(1-p)^{-1/2}$ as $p\to1^-$. The Gaussian quantile satisfies $q_p^\Phi\sim[2\ln(1-p)^{-1}]^{1/2}$, which grows more slowly than $(1-p)^{-1/2}$.
\end{proof}

\begin{proof}[Proof of Theorem~\ref{thm:joint_consistency}]
For any $\lambda$ and $z$, $|\widehat F_n(\mu+\sigma z)-F_0(\mu+\sigma z)|
\leq\sup_{x \in \mathbb{R}}|\widehat F_n(x)-F_0(x)|$, so
\[
    \sup_{\theta \in \Theta}|Q_n(\theta)-Q(\theta)|\leq 2\nu(\mathbb{R})\,
    \sup_{x \in  \mathbb{R}}|\widehat F_n(x)-F_0(x)|\to0
\]
almost surely by the Glivenko--Cantelli theorem. The argmin theorem then gives the result; see the proof of Theorem~\ref{thm:consistency} for the standard argument.
\end{proof}

\begin{proof}[Proof of Theorem~\ref{thm:joint_boundary}]
Introduce $\lambda=\lambda_0+t/\sqrt{n}$, $c=c_0+h/\sqrt{n}$, $h\geq0$, with $v=(t',h)'\in\mathbb{R}^2\times[0,\infty)$. Since $\sigma_0>0$, the constraint $\sigma>0$ is not locally binding and $t$ may be taken in
$\mathbb{R}^2$. The local criterion $L_n(v)=n[Q_n(\theta_0+v/\sqrt{n})-Q_n(\theta_0)]$ satisfies $L_n(v)=v'S_n+\frac{1}{2}v'Hv+o_p(1)$ uniformly over bounded $v$, by the same quadratic approximation argument as Theorem~\ref{thm:boundary}. For fixed $h$, minimising over $t$ yields
\[
    t(h)=-H_{\lambda\lambda}^{-1}S_\lambda
    -H_{\lambda\lambda}^{-1}H_{\lambda c}\,h.
\]
Substituting leaves the reduced boundary criterion
\[
    L_{\mathrm{red}}(h)
    = h\,S_{c\cdot\lambda}
    +\tfrac{1}{2}H_{cc\cdot\lambda}\,h^2,
    \quad h\geq0,
\]
up to a constant in $h$. This has the same form as the local criterion in Theorem~\ref{thm:boundary}, with $S_c$ and $H_{cc}$ replaced by their
nuisance-adjusted counterparts $S_{c\cdot\lambda}$ and $H_{cc\cdot\lambda}$. The constrained minimiser is $h_n=\max\{0,-S_{c\cdot\lambda}/H_{cc\cdot\lambda}\}$. Tightness of $\sqrt{n}(\widehat c_n-c_0)$ follows by the same strong-convexity argument as
Lemma~\ref{lem:tightness}, and the remainder of the proof is identical to that of Theorem~\ref{thm:boundary}, with the Gaussianity of $S_{c\cdot\lambda}\sim\N(0,\Omega_{cc\cdot\lambda})$ giving the stated distribution.
\end{proof}

\begin{proof}[Proof of Theorem~\ref{thm:joint_fit_improvement}]
The improvement from allowing $h\geq0$ rather than fixing $h=0$ satisfies
\[
    T_n^J\Rightarrow{-}\min_{h\geq0}\!\left\{hS_{c\cdot\lambda}
    +\tfrac{1}{2}H_{cc\cdot\lambda}h^2\right\}
    =\frac{S_{c\cdot\lambda}^2}{2H_{cc\cdot\lambda}}
    \mathbf{1}\{S_{c\cdot\lambda}<0\}.
\]
Writing $S_{c\cdot\lambda}=\Omega_{cc\cdot\lambda}^{1/2}U$ with $U\sim\N(0,1)$, this equals $0$ with probability $1/2$ and equals $(\Omega_{cc\cdot\lambda}/2H_{cc\cdot\lambda}) \,\chi_1^2$ with probability $1/2$. Multiplying by $2H_{cc\cdot\lambda}/\Omega_{cc\cdot\lambda}$ gives $\frac{1}{2}\delta_0 +\frac{1}{2}\chi_1^2$, and the rejection rule follows as in Theorem~\ref{thm:fit_improvement}.
\end{proof}

\newpage

\section{Simulation details}\label{app:simulations}

This appendix provides a complete description of the simulation design used in Section~\ref{sec:simulations}, sufficient for independent replication. It gives all numerical choices, every formula used in implementation, and the designs of the null-size, power, and diagnostic experiments.

\subsection*{Common numerical ingredients}

All hypergeometric (HGM) statistics are based on the same
minimum-distance criterion, applied to the appropriate empirical
distribution function:
\[
    Q_n(c) = \int_{\mathbb{R}}\{\widehat F_n(z)-G_c(z)\}^2\,d\nu(z),
    \quad c\in[c_0,c_U],
\]
where
\[
    G_c(z) = \frac12 + z\,\frac{\Gamma(c-\frac12)}{2\sqrt2\,\Gamma(c)}
    \,{}_1F_1\!\left(\frac12;c;-\frac{z^2}{2}\right), \quad c\geq \frac32,
\]
with $c_0=3/2$, $c_U=8$, and $d\nu_{1/2}(z)\propto[\Phi(z)\{1-\Phi(z)\}]^{-1/2}d\Phi(z)$, and $\Phi$ is the distribution function of the standard normal. The integral is approximated by
\begin{equation}\label{eq:Qn_discrete}
    Q_n(c) \approx \sum_{j=1}^J
    w_j\{\widehat F_n(z_j)-G_c(z_j)\}^2
\end{equation}
on the Gaussian quantile grid
\begin{equation}\label{eq:grid}
    z_j = \Phi^{-1}(u_j), \quad u_j = \frac{j-1/2}{J}, \quad j=1,\ldots,J,
\end{equation}
with $J=301$, and tail-weighted normalised weights
\begin{equation}\label{eq:weights}
    w_j = \frac{[u_j(1-u_j)]^{-1/2}}
    {\sum_{\ell=1}^J[u_\ell(1-u_\ell)]^{-1/2}}.
\end{equation}
As a robustness check, calculations were repeated using a
finer uniform grid with $J=501$ nodes, and separately by using Gauss--Chebyshev
quadrature. The resulting rejection frequencies were almost unchanged. For computational speed, the precomputed identity
\[
    Q_n(c) = \sum_{j=1}^J w_j\widehat F_n(z_j)^2 +
    \sum_{j=1}^J w_jG_c(z_j)^2 -
    2\sum_{j=1}^J w_j\widehat F_n(z_j)G_c(z_j)
\]
is used, with all $G_c(z_j)$ values precomputed for every $(z_j,c)$ pair before the Monte Carlo loop begins.

\subsection*{Parameter grid}

Minimisation over $c$ uses a deterministic grid
$\mathcal C_{\mathrm{grid}}$ concentrated near the boundary:
\[
\begin{aligned}
    \mathcal{C}_{\mathrm{grid}}
    =
    \{c_0\}
    &\cup
    \bigl\{c_0+\exp(a_m):
    a_m\in[\ln(10^{-8}),\ln(10^{-3})],\;
    m=1,\ldots,90\bigr\}\\
    &\cup
    \bigl\{c_0+t_m:
    t_m\in[0.0012,0.10],\;
    m=1,\ldots,100\bigr\}\\
    &\cup
    \bigl\{c_0+s_m:
    s_m\in[0.105,0.50],\;
    m=1,\ldots,100\bigr\}\\
    &\cup
    \bigl\{r_m:
    r_m\in[2.02,c_U],\;
    m=1,\ldots,100\bigr\}.
\end{aligned}
\]
This gives 391 grid points. The logarithmically spaced block near $c_0$ is used because the null value is a boundary point and the fitted value can lie extremely close to $c_0$. The remaining blocks give progressively coarser coverage of moderate and large departures from the Gaussian boundary. For each simulated sample, $\widehat c_n$ is the grid point minimising \eqref{eq:Qn_discrete}. The estimator is classified as on the lower boundary whenever $\widehat c_n\leq c_0+10^{-4}$, and as on the upper boundary whenever $\widehat c_n\geq c_U-10^{-5}$. The fit improvement $\Delta_n=Q_n(c_0)-Q_n(\widehat c_n)$ is truncated at zero if negative, and set to zero if $\Delta_n\leq 10^{-10}$:
\[
    \Delta_n \leftarrow \max\{\Delta_n,\,0\}, \quad \Delta_n \leftarrow
    \Delta_n\,\mathbf{1}\{\Delta_n>10^{-10}\}.
\]

\subsection*{Null constants}

The null constants are computed on the same Gaussian quantile grid and
with the same weights used for the criterion. The score derivative $\dot G_{c_0}(z)= \partial G_c(z)/\partial c\,|_{c=c_0}$ is evaluated analytically using \eqref{eq:Gdot}--\eqref{eq:1F1dot}. Writing
\[
    \kappa(c)=\frac{\Gamma(c-\frac12)}{2\sqrt2\,\Gamma(c)},
    \quad x=-\frac{z^2}{2},
\]
we have
\[
    \dot G_c(z) =
    z\left[
        \dot \kappa(c)\,{}_1F_1\!\left(\frac12;c;x\right)
        + \kappa(c)\,\frac{\partial}{\partial c}
        {}_1F_1\!\left(\frac12;c;x\right)
    \right],
\]
where
\[
    \dot \kappa(c) = \kappa(c)\{\psi(c-\tfrac12)-\psi(c)\},
\]
and $\psi$ denotes the digamma function. The derivative of
${}_1F_1$ with respect to its second parameter is computed from
\[
    \frac{\partial}{\partial c}\,
    {}_1F_1\!\left(\frac12;c;x\right) = - \sum_{k=1}^{\infty}
    \frac{(1/2)_k}{(c)_k\,k!}x^k[\psi(c+k)-\psi(c)].
\]
The series is summed until the next term is negligible relative to the
running sum. The curvature constant is discretised as
\[
    H_{c_0} = 2\sum_{j=1}^J w_j\,\dot G_{c_0}(z_j)^2.
\]
The fixed-scale variance uses the Brownian-bridge kernel
\[
    K_{jk} = \min(u_j,u_k)-u_ju_k,
\]
giving
\begin{equation}\label{eq:Omega_discrete}
    \Omega_{c_0} = 4\sum_{j=1}^J\sum_{k=1}^J w_j\,w_k\,K_{jk}\,
    \dot G_{c_0}(z_j)\,\dot G_{c_0}(z_k).
\end{equation}
For the mean/SD plug-in statistic, the Brownian-bridge kernel is
replaced by the projected kernel
\[
    K^P_{jk} = K_{jk} - \phi(z_j)\phi(z_k) -
    \frac12 z_jz_k\phi(z_j)\phi(z_k),
\]
consistent with \eqref{eq:KP}. The corresponding variance
$\Omega_{c_0}^P$ is obtained from \eqref{eq:Omega_discrete} with
$K_{jk}$ replaced by $K^P_{jk}$. The projected kernel for robust median/IQR standardisation, consistent with \eqref{eq:KR}, is
$K^R_{jk}=\E[\xi^R_{z_j}(W)\xi^R_{z_k}(W)]$, $W\sim\N(0,1)$, where the projected influence function is
\[
    \xi^R_z(x) = \mathbf{1}\{x\leq z\}-\Phi(z) + \phi(z)\IF_m(x)
    + z\phi(z)\IF_s(x),
\]
with influence functions $\IF_m$ and $\IF_s$ from Lemma~\ref{lem:robust_ep}. Setting
$a=\Phi^{-1}(3/4)$:
\[
    \IF_m(x) = \frac{1/2-\mathbf{1}\{x\leq0\}}{\phi(0)}, \quad
    \IF_s(x) = \frac{1}{2a} \left[
    \frac{3/4-\mathbf{1}\{x\leq a\}}{\phi(a)} -
    \frac{1/4-\mathbf{1}\{x\leq -a\}}{\phi(a)} \right].
\]
The expectation defining $K^R_{jk}$ is approximated by deterministic quadrature on a fine probability grid $v_\ell=(\ell-1/2)/M$, $x_\ell=\Phi^{-1}(v_\ell)$, $\ell=1,\ldots,M$, with $M=4001$:
\[
    K^R_{jk} \approx \frac{1}{M} \sum_{\ell=1}^M
    \xi^R_{z_j}(x_\ell)\,\xi^R_{z_k}(x_\ell).
\]
The corresponding variance $\Omega_{c_0}^R$ is obtained from \eqref{eq:Omega_discrete} with
$K_{jk}$ replaced by $K^R_{jk}$.

\subsection*{Three HGM statistics}

Each Monte Carlo sample is used to compute three HGM statistics, using
the same $z$-grid, $c$-grid, weights, and boundary critical values. Given the raw sample $X_1,\ldots,X_n$, the empirical distribution function is
$\widehat F_n(z)=n^{-1}\sum_{i=1}^n\mathbf{1}\{X_i\leq z\}$,
evaluated at each $z_j$. Let $\widehat c_n=\arg\min_{c \in \mathcal{C}_{\mathrm{grid}}}Q_n(c)$ and set $\Delta_n = Q_n(c_0)-Q_n(\widehat c_n)$. The fixed-scale statistic is
\[
    \widetilde T_n = \frac{2H_{c_0}}{\Omega_{c_0}}\,n\,\Delta_n .
\]
For the mean/SD plug-in statistic, each sample is standardised by
\[
    \widehat\mu_n^P=\bar X_n, \quad \widehat\sigma_n^P =
    \left\{
        \frac1n\sum_{i=1}^n (X_i-\bar X_n)^2
    \right\}^{1/2},
\]
and $Y_i^P=(X_i-\widehat\mu_n^P)/\widehat\sigma_n^P$. The empirical distribution of $Y_1^P,\ldots,Y_n^P$ defines $Q_n^P(c)$, with minimiser $\widehat c_n^P$, improvement $\Delta_n^P=Q_n^P(c_0) - Q_n^P(\widehat c_n^P)$, and mean/SD plug-in statistic
\[
    \widetilde T_n^P = \frac{2H_{c_0}}{\Omega_{c_0}^P}\,n\,\Delta_n^P.
\]
For the robust median/IQR plug-in statistic, each sample is standardised by the sample median $\widehat\mu_n^R$ and the normal-calibrated IQR scale
\[
    \widehat\sigma_n^R =
    \frac{\widehat q_n(3/4)-\widehat q_n(1/4)}
    {2\Phi^{-1}(3/4)},
\]
where $\widehat q_n(p)$ is the $p$-th sample quantile, and $\widehat\sigma_n^R$ is truncated below at $10^{-12}$ in computation. Let $Y_i^R=(X_i-\widehat\mu_n^R)/\widehat\sigma_n^R$. The empirical distribution of $Y_1^R,\ldots,Y_n^R$ defines $Q_n^R(c)$, with minimiser $\widehat c_n^R$, improvement $\Delta_n^R = Q_n^R(c_0) - Q_n^R(\widehat c_n^R)$, and robust median/IQR plug-in statistic
\[
    \widetilde T_n^R = \frac{2H_{c_0}}{\Omega_{c_0}^R}\,n\,\Delta_n^R.
\]
The statistics use boundary critical value $k_\alpha=\chi^2_{1,1-2\alpha}$, where $\Pr(\frac12\delta_0+\frac12\chi_1^2>k_\alpha)=\alpha$.

\subsection*{Scaling constants}

Table~\ref{tab:null_size_tail_scaling_robust} reports the numerical
constants used to scale the fixed-scale, mean/SD plug-in, and
median/IQR plug-in statistics. The same tail-weighted measure
$\nu_{1/2}$ and Gaussian quantile grid are used throughout. The plug-in constants differ only through the projected score variance. The mean/SD projection gives a substantially smaller variance than the fixed-scale statistic and therefore a larger normalising factor, while
the median/IQR projection produces an intermediate scaling. The profiled
joint constants are included only as a diagnostic for Remark~\ref{rem:plugin_vs_joint}; the very small $\Omega_{cc\cdot\lambda}$ explains the large profiled scaling factor and the instability of the exact-profiled statistic in the simulations.

\input{table_null_size_tail_scaling_with_robust}

\subsection*{Null-size simulations}

Null-size simulations draw $X_i\stackrel{\mathrm{i.i.d.}}{\sim}\N(0,1)$ for $n\in\{25,50,100,200,400,800,1600\}$, using $R=10^5$ replications and random-number seed 46387492; see Table~\ref{tab:null_size_tail_main_robust} for results. For each replication, the same sample is used to compute $\widetilde T_n$, $\widetilde T_n^P$, and $\widetilde T_n^R$. Rejection at level $\alpha$ uses $k_{0.10}=1.642$, $k_{0.05}=2.706$, and $k_{0.01}=5.412$ from Table~\ref{tab:null_size_tail_scaling_robust}. For each statistic, sample size, and nominal level, empirical size and Monte Carlo standard error are
\[
    \widehat{\mathrm{size}}_\alpha = \frac1R\sum_{r=1}^R
    \mathbf 1\{\widetilde T_{n,r}>k_\alpha\}, \quad
    \widehat{\mathrm{mcse}}_\alpha = \left[
    \frac{\widehat{\mathrm{size}}_\alpha
    \{1-\widehat{\mathrm{size}}_\alpha\}}{R}
    \right]^{1/2},
\]
where $\widetilde T_{n,r}$ denotes the panel-specific statistic. For
example, for the robust statistic $\widetilde T_n^R$ at the $5\%$ level,
the Monte Carlo standard errors are below $0.001$ for all sample sizes
reported. Diagnostics recorded for each panel and sample size include rejection frequencies at $\alpha\in\{0.10,0.05,0.01\}$, the lower-boundary rate $\Pr(\widehat c_n\leq c_0+10^{-4})$, the upper-boundary rate
$\Pr(\widehat c_n\geq c_U-10^{-5})$, the mean and median of $\widehat c_n$, and summary quantiles of the corresponding statistic. For the plug-in versions, the mean estimated scale is also recorded.

\subsection*{Location-scale invariance}

\noindent The robust plug-in statistic is also evaluated
under $X_i\stackrel{\mathrm{i.i.d.}}{\sim}\N(3,4)$ for $n\in\{100,400,1600\}$ with $R=10^5$ replications; see Table~\ref{tab:null_size_tail_robust_plugin_invariance_N3_2sq}. The rejection frequencies under $\N(0,1)$ and $\N(3,4)$ are very close in absolute terms, with differences no larger than $0.002$ for the reported rejection probabilities. This confirms approximate location-scale invariance of the implemented robust plug-in procedure.

\input{table_null_size_tail_robust_plugin_invariance_N3_2sq}

\clearpage

\subsection*{Exact-profiled statistic}

The exact-profiled statistic minimises the criterion jointly over $(\mu,\sigma,c)$, which, as noted in Appendix~\ref{app:joint}, is equivalent to optimising $(\mu,\sigma)$ exactly for each fixed $c$ and then minimising the resulting profile criterion over $c$:
\[
    Q_n^{EP}(\lambda,c) = \sum_{j=1}^J w_j\bigl\{\widehat F_n(\mu+\sigma z_j) - G_c(z_j)\bigr\}^2,\quad\lambda=(\mu,\sigma)'.
\]
The restricted (Gaussian boundary) and unrestricted profiles are
\[
    Q_{n,0}^{EP} = \inf_{\mu\in\mathbb{R},\,\sigma>0}
    Q_n^{EP}(\lambda,c_0),\quad Q_{n,1}^{EP} =
    \inf_{\mu\in\mathbb{R},\,\sigma>0,\,
    c\in[c_0,c_U]} Q_n^{EP}(\lambda,c).
\]
The improvement is $\Delta_n^{EP}=Q_{n,0}^{EP}-Q_{n,1}^{EP}$, set to zero if negative. Both profiles are approximated by a three-stage nested grid search in $(\mu,\ln\sigma)$, with grids centred on the sample mean $\bar X$ and sample standard deviation $s_X$. For the restricted profile at $c=c_0$:
\begin{align*}
    \text{Stage~1 ($9\times9$):}\quad&
    \mu\in\bar X+s_X[-0.55,0.55],\quad
    \sigma\in s_X\exp([-0.55,0.55]);\\
    \text{Stage~2 ($7\times7$):}\quad&
    \mu\in\widehat\mu_1+s_X[-0.16,0.16],\quad
    \sigma\in\widehat\sigma_1\exp([-0.18,0.18]);\\
    \text{Stage~3 ($5\times5$):}\quad&
    \mu\in\widehat\mu_2+s_X[-0.05,0.05],\quad
    \sigma\in\widehat\sigma_2\exp([-0.06,0.06]).
\end{align*}
Each stage centres on the previous stage's minimiser. For the unrestricted profile, initialised at $(\widehat\mu_0,\widehat\sigma_0)$ and searched jointly over $(\mu,\ln\sigma,c)$ on $\mathcal{C}_{\mathrm{grid}}$:
\begin{align*}
    \text{Stage~1 ($9\times9\times|\mathcal{C}|$):}
    \quad&
    \mu\in\widehat\mu_0+s_X[-0.35,0.35],\quad
    \sigma\in\widehat\sigma_0\exp([-0.35,0.35]);\\
    \text{Stage~2 ($7\times7\times|\mathcal{C}|$):}
    \quad&
    \mu\in\widehat\mu_1+s_X[-0.12,0.12],\quad
    \sigma\in\widehat\sigma_1\exp([-0.14,0.14]);\\
    \text{Stage~3 ($5\times5\times|\mathcal{C}|$):}
    \quad&
    \mu\in\widehat\mu_2+s_X[-0.04,0.04],\quad
    \sigma\in\widehat\sigma_2\exp([-0.05,0.05]).
\end{align*}
If the unrestricted criterion exceeds the restricted criterion due to numerical error, the improvement is set to zero. The statistic is
$\widetilde T_n^{EP}=(2H_{cc\cdot\lambda}/\Omega_{cc\cdot\lambda})\,n\,\Delta_n^{EP}$,
with scaling constants computed from the $3\times3$ derivative matrix
$D=(D_1,\ldots,D_J)'$,
\[
    D_j = \bigl(\phi(z_j),\;z_j\phi(z_j),\;{-\dot G_{c_0}(z_j)}\bigr)',
\]
evaluated at $\sigma_0=1$ (the simulation null). Let $W=\mathrm{diag}(w_1,\ldots,w_J)$ and let $K_0$ be the $J\times J$ Brownian bridge kernel
matrix with $(K_0)_{jk}=\Phi(z_j\wedge z_k)-\Phi(z_j)\Phi(z_k)$. Then
\[
    H^{EP}=2D'WD,\quad\Omega^{EP}=4D'WK_0WD.
\]
Partitioning conformably with $\lambda=(\mu,\ln\sigma)'$ and $c$,
\[
    H_{cc\cdot\lambda} = H_{cc}-H_{c\lambda}H_{\lambda\lambda}^{-1}
    H_{\lambda c},\quad a'=\bigl({-H_{c\lambda}
    H_{\lambda\lambda}^{-1}},\;1\bigr), \quad
    \Omega_{cc\cdot\lambda}=a'\Omega^{EP}a.
\]
The exact-profile simulation experiment uses $R=10^3$ replications under $\N(0,1)$ for $n\in\{50,100,200,400,800,1600\}$. The near-degeneracy of $\Omega_{cc\cdot\lambda}\approx0.000013$ (Table~\ref{tab:null_size_tail_scaling_robust}) is structural: as explained in Remark~\ref{rem:plugin_vs_joint}, the $c$-score direction $\dot G_{c_0}$ nearly lies in the span of $\{\phi(z),z\phi(z)\}$ in the Brownian bridge covariance metric, making the residual score variance very small and the normalisation constant very large. The size distortion worsens with $n$ over any practically feasible range; the plug-in approach avoids this problem by construction.

\input{table_null_size_tail_exact_profile_appendix}

\newpage

\subsection*{Power simulation design}

Power is assessed at nominal level $\alpha=0.05$ with $R=10^4$ replications and seed 72921253, for $n\in\{25,50,100,200,400,800\}$. The HGM statistics use the asymptotic boundary critical value $k_{0.05}=\chi^2_{1,0.90}=2.706$. The main power tables report rejection frequencies for the fixed-scale test $\widetilde T_n$, the robust plug-in $\widetilde T_n^R$, and the omnibus benchmarks Jarque--Bera (JB), Anderson--Darling (AD), and Shapiro--Wilk (SW). The mean/SD plug-in $\widetilde T_n^P$ appears only in the
appendix tables.\\
\\
\noindent \textit{Jarque--Bera.} The statistic is
\[
    {JB}_n = \frac{n}{6} \left(\widehat S_n^2+\tfrac14\widehat K_n^2\right),
\]
where $\widehat S_n$ is sample skewness and $\widehat K_n$ is sample excess kurtosis; rejection occurs at ${JB}_n>\chi^2_{2,0.95}=5.991$.\\
\\
\noindent \textit{Anderson--Darling.} Each sample is standardised
by its sample mean and standard deviation. For sorted standardised observations $Z_{(1)}\leq\cdots\leq Z_{(n)}$ and $P_i=\Phi(Z_{(i)})$,
\[
    A_n^2 = -n - \frac1n\sum_{i=1}^n(2i-1)
    \bigl\{\ln P_i+\ln(1-P_{n+1-i})\bigr\}.
\]
The \citet{stephens74} finite-sample adjustment gives $A_n^{2,*}=A_n^2(1+0.75/n+2.25/n^2)$, with $5\%$ critical value $0.787$.\\
\\
\textit{Shapiro--Wilk.} The statistic is
\[
    W = \frac{\bigl(\sum_{i=1}^n a_i X_{(i)}\bigr)^2}
             {\sum_{i=1}^n(X_i-\bar X)^2},
\]
where $X_{(1)}\leq\cdots\leq X_{(n)}$ are the order statistics, $m_i=\E[Z_{(i)}]$ for $Z_1,\ldots,Z_n\stackrel{\mathrm{i.i.d.}}{\sim}\N(0,1)$, $V$ is the covariance matrix of $(Z_{(1)},\ldots,
Z_{(n)})'$, and $a\propto m'V^{-1}$. Under normality, $W\approx1$; departures from normality reduce $W$ below $1$. The statistic is computed via the \citet{royston92} approximation as implemented in SciPy; rejection occurs at $p<0.05$.\\
\\
For each test, alternative, and sample size,
\[
    \widehat\pi = \frac{1}{R}\sum_{r=1}^R
    \mathbf{1}\{\text{test rejects in replication }r\},
    \quad
    \widehat{\mathrm{mcse}} =
    \bigl[\widehat\pi(1-\widehat\pi)/R\bigr]^{1/2}.
\]
For the HGM procedures the simulations also record $\mathrm{mean}(\widehat c_n)$, $\mathrm{median}(\widehat c_n)$, the boundary rate $\Pr(\widehat c_n\leq c_0+10^{-4})$, and the mean of $\widetilde T_n$.

\subsection*{HGM alternatives}

HGM alternatives $X_i\sim G_{c_A}$ with $c_A\in\{1.55,1.60,1.75,2.00,3.00\}$ are generated by inverse-CDF interpolation. For each $c_A$, a monotone grid is built as follows. Set $s_{\max}=\sinh^{-1}(x_{\max})$ with $x_{\max}=20{,}000$ and let $s$ be equally spaced in $[-s_{\max},s_{\max}]$ with $N_{\mathrm{CDF}}= 120{,}000$ points. The $x$-grid is $x(s)=\sinh(s)$, whose Jacobian $\cosh(s)$ equals $1$ at $s=0$ and grows as $|x|$ for large $|x|$, concentrating grid density near the centre and giving logarithmic spacing in the tails. Probabilities $p(s)=G_{c_A}(x(s))$ are clipped to $[0,1]$ to handle floating-point rounding at extreme arguments, then made strictly
monotone by taking running maxima and removing resulting ties. To draw a sample of size $n$, independent uniforms on the retained probability range $[p_{\min},p_{\max}]$ are linearly interpolated through the resulting $(p,x)$ grid, inverting the CDF without extrapolation.

\subsection*{Standard non-HGM alternatives}

Standard alternatives are standardised to mean zero and unit variance so that the fixed-scale test, which assumes $\N(0,1)$ with known location and scale, is evaluated against a shape departure rather than a
scale mismatch:
\begin{align*}
    \text{Student }t_k\;(k\in\{3,5,10\}):\quad
    &X\sim\sqrt{(k-2)/k}\cdot t_k,\\
    \text{Laplace}:\quad
    &X\sim\mathrm{Laplace}(0,1/\sqrt{2}),\\
    \text{Cont.\ normal}:\quad
    &X=X^\ast/\sqrt{1.40},\quad
    X^\ast\sim0.95\,\N(0,1)+0.05\,\N(0,9),\\
    \text{Chi-squared}:\quad
    &X=(\chi^2_3-3)/\sqrt{6}.
\end{align*}

\subsection*{Power: mean/SD plug-in}

Tables~\ref{tab:appendix_mean_plugin_hgm_power} and~\ref{tab:appendix_mean_plugin_standard_power} confirm that $\widetilde T_n^P$ has essentially zero power in all scenarios. Against HGM alternatives with $c>3/2$, the sample standard deviation diverges as $n\to\infty$ because the variance is infinite, collapsing the standardised observations toward zero and destroying all tail signal: power against $G_{3.0}$ is $0.004$ at $n=25$ and zero by $n=100$, and the pattern is decreasing in $n$ for every alternative. Against standard alternatives, power is similarly negligible. These results validate the recommendation of the robust median/IQR version.

\subsection*{Power: size-adjusted robust plug-in}

The size-adjusted robust plug-in replaces the asymptotic critical value $\chi^2_{1,0.90}$ with the empirical $95\%$ quantiles of $\widetilde T_n^R$ under $\N(0,1)$, computed from the null-size simulations and passed as precomputed constants. The values used are:
\[
\begin{array}{lllllll}
n: & 25 & 50 & 100 & 200 & 400 & 800 \\[2pt]
\widehat k_{n,0.05}^R: &
4.184 & 3.494 & 3.321 & 3.115 & 3.009 & 2.952.
\end{array}
\]
This removes the mild liberalness of Panel~C in Table~\ref{tab:null_size_tail_main_robust} and provides a conservative benchmark for genuine power. Against HGM alternatives
(Table~\ref{tab:appendix_sizeadj_robust_hgm_power}), size-adjusted power rises clearly with $c$ and $n$: at $c=3.00$ it reaches $0.981$ at $n=800$; at $c=2.00$ it reaches $0.705$; against $c=1.55$ it is $0.053$--$0.087$ across all $n$. Against standard alternatives (Table~\ref{tab:appendix_sizeadj_robust_standard_power}), it reaches $0.999$ against Laplace and $0.974$ against $t_3$ at $n=800$; against $\chi^2_3$ it is $0.018$ or below at all $n$, consistent with the symmetry of the HGM family.

\subsection*{Reproducibility}

All simulations use Python (v.~3.12.11) with NumPy (v.~2.5.0) and SciPy (v.~1.18.0). Pseudo-random draws use NumPy's default bit generator (PCG64) with seeds $46387492$ (null-size simulations) and $72921253$ (power simulations). All $G_c$ values are evaluated using \texttt{scipy.special.hyp1f1} and precomputed on the $z$- and $c$-grids before the Monte Carlo loop. Code and data are available from the author on request.

\input{table_appendix_mean_plugin_hgm_power}
\input{table_appendix_mean_plugin_standard_power}

\clearpage

\input{table_appendix_sizeadjusted_robust_hgm_power}
\input{table_appendix_sizeadjusted_robust_standard_power}

\clearpage

\section{Empirical illustration}\label{app:empirical}

This appendix collects descriptive statistics,
fitted parameters, normality test results,
goodness-of-fit measures, lower-tail quantiles,
and tail probabilities for both exercises of
Section~\ref{sec:application}. Tables for the
raw returns appear first, followed by the GARCH
parameter estimates, the estimated conditional
volatility series, and the corresponding tables
for the standardised residuals.

\clearpage

\input{table_sp500_returns_desc}
\input{table_sp500_returns_params}
\input{table_sp500_returns_tests}
\input{table_sp500_returns_gof}

\clearpage

\input{table_sp500_returns_var}
\input{table_sp500_returns_tail}

\clearpage

\input{table_sp500_garch_residuals_desc}
\input{table_sp500_garch11_parameters}

\begin{figure}[p]
\centering
\includegraphics[width=0.82\textwidth]
    {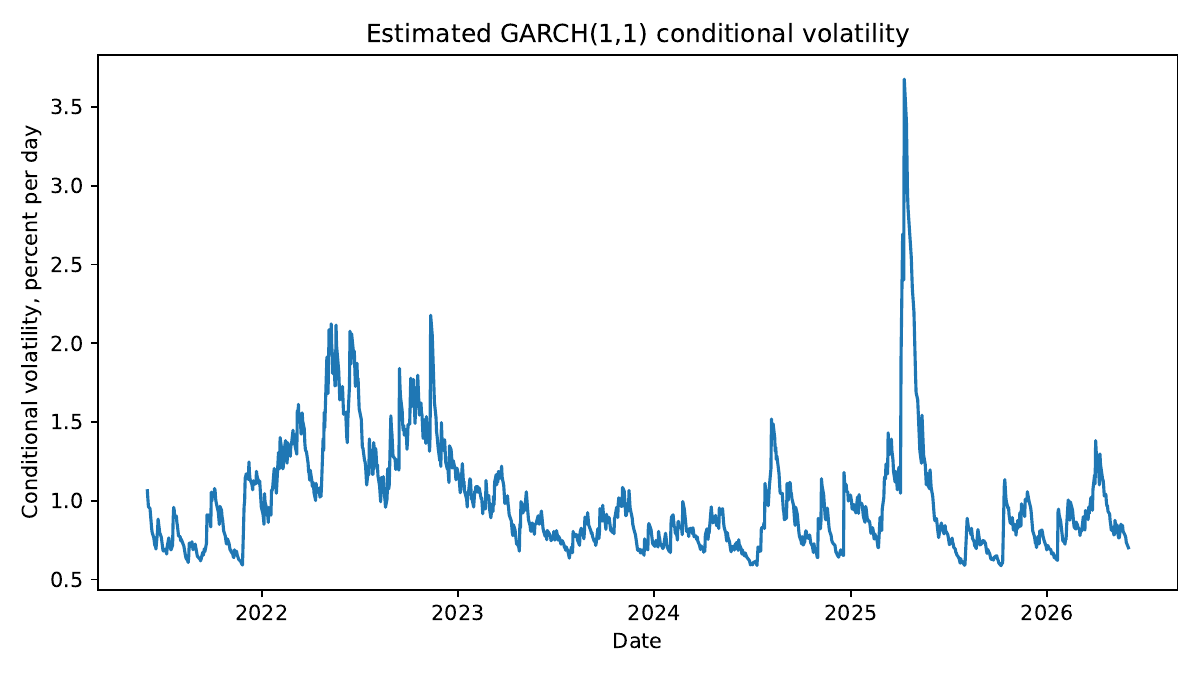}
\caption{Estimated GARCH(1,1) conditional volatility
$\widehat h_t^{1/2}$, 2~June~2021 to 1~June~2026.}
\label{fig:garch_volatility}
\end{figure}

\clearpage

\input{table_sp500_garch_residuals_params}
\input{table_sp500_garch_residuals_tests}
\input{table_sp500_garch_residuals_gof}

\clearpage

\input{table_sp500_garch_residuals_var}
\input{table_sp500_garch_residuals_tail}

\clearpage

\section{Tail dominance and Value-at-Risk (VaR)}\label{app:crossover}

For every $c>3/2$ and $z>0$, $1-G_c(z)>1-\Phi(z)$. The correction factor
\[
    \rho(z,c) := \frac{1-G_c(z)}{1-\Phi(z)} > 1
    \quad\text{for all }z>0,\;c>3/2
\]
satisfies the following monotonicity property.

\begin{lemma}[Monotonicity of the correction factor]\label{lem:R_monotone}
For every $c>3/2$, $\rho(z,c)$ is strictly increasing in $c$ for each fixed $z>0$, and strictly increasing in $z$ for each fixed $c>3/2$.
\end{lemma}

\noindent Table~\ref{tab:correction} reports exact values of $\rho$ computed from the HGM survival function. For $X\sim G_c$, the VaR at confidence level $p$ is the quantile $q_p>0$ satisfying $1-G_c(q_p)=1-p$, computed by numerical inversion of $G_c$; Table~\ref{tab:var_ratio} reports exact values alongside Gaussian quantiles for $\widehat c-3/2 =0.30$. The ratio $q_p^{\mathrm{HGM}}/ q_p^\Phi$ is increasing in both $c$ and $p$:

\begin{proposition}[VaR divergence]\label{prop:var_divergence}
For any $c>3/2$, $q_p^{\mathrm{HGM}}/q_p^\Phi\to\infty$ as $p\to1^-$.
\end{proposition}

\noindent Proposition~\ref{prop:var_divergence} implies that any rejection of $H_0:c=c_0$, however small $\widehat c-3/2$, entails diverging Gaussian understatement of VaR at sufficiently extreme confidence levels.

\input{table_correction_factor}
\input{table_var_ratio}

\clearpage

\section{Joint location-scale estimation}\label{app:joint}

This appendix provides a self-contained treatment of joint minimum-distance estimation of $\theta=(\mu,\sigma,c)'$, included as a theoretical reference for the plug-in procedures of Section~\ref{sec:plugin}. Joint estimation is
equivalent to exact profiling: for each fixed $c$, minimising the criterion over $\lambda=(\mu,\sigma)'$ yields a profile criterion whose minimiser over $c$ coincides with the joint minimiser $\widehat c_n$. I use the term \textit{joint} in the theoretical development below and \textit{exact-profiled} in the computational implementation of Section~\ref{app:simulations}, where it serves as a diagnostic. All results are valid under the stated conditions, but the near-degeneracy of $\Omega_{cc\cdot\lambda}\approx0.000013$ (Table~\ref{tab:null_size_tail_scaling_robust}) means the asymptotic approximation requires very large $n$; see Section~\ref{app:simulations}.

\subsection*{Setup}

Let $\Theta=\mathcal{M}\times\mathcal{S}\times\mathcal{C}$ be compact, with
$\mathcal{M}=[\mu_L,\mu_U]$, $\mathcal{S}=[\sigma_L,\sigma_U]$, and
$\mathcal{C}=[c_0,c_U]$. Writing $\lambda=(\mu,\sigma)'$ and working in
standardised coordinates $z=(x-\mu)/\sigma$, define
\[
    Q_n(\lambda,c) = \int_{\mathbb{R}}\bigl\{\widehat F_n(\mu+\sigma z)
    -G_c(z)\bigr\}^2\,d\nu(z),
\]
with population counterpart $Q(\lambda,c)$. Let $\widehat\theta_n=(\widehat\lambda_n', \widehat c_n)'\in\arg\min_\Theta Q_n$.

\subsection*{Consistency}

\begin{theorem}[Joint consistency]\label{thm:joint_consistency}
Suppose $\nu(\mathbb{R})<\infty$, $\Theta$ is compact, and $Q$ has a unique minimiser in $\Theta$. Then $\widehat\theta_n$ converges to that minimiser almost surely.
\end{theorem}

\subsection*{Boundary distribution}

Under the Gaussian location-scale null $F_0(\mu_0+\sigma_0 z)=\Phi(z)=G_{c_0}(z)$, write $\theta_0=(\lambda_0',c_0)'$ and partition
\[
    S = \begin{pmatrix}S_\lambda\\S_c\end{pmatrix}, \quad
    H =
    \begin{pmatrix}
        H_{\lambda\lambda} & H_{\lambda c}\\
        H_{c\lambda} & H_{cc}
    \end{pmatrix},
\]
where $S_n=\sqrt{n}\,\nabla Q_n(\theta_0)\Rightarrow S\sim\N(0,\Omega)$,
$H=\nabla^2 Q(\theta_0)$, and $\Omega$ is the limiting score covariance matrix. Under the standardised-coordinate criterion, the null derivative vector is
\[
    D_0(z) = \bigl(\phi(z)/\sigma_0,\; z\phi(z)/\sigma_0,\;
    {-\dot G_{c_0}(z)}\bigr)',
\]
so $H = 2\int_\mathbb{R} D_0(z) D_0(z)'\,d\nu(z)$ and
\[
    \Omega=4\iint_{\mathbb{R}^2}[\Phi(z\wedge y)-\Phi(z)\Phi(y)]
    D_0(z)D_0(y)'\,d\nu(z)\,d\nu(y).
\]
Define the nuisance-adjusted score and curvature:
\[
    S_{c\cdot\lambda} = S_c-H_{c\lambda}H_{\lambda\lambda}^{-1}
    S_\lambda, \quad H_{cc\cdot\lambda} = H_{cc}-H_{c\lambda}H_{\lambda\lambda}^{-1}H_{\lambda c},
\]
and $\Omega_{cc\cdot\lambda}=a'\Omega a$, where $a'=(-H_{c\lambda}H_{\lambda\lambda}^{-1},\;1)$.

\begin{theorem}[Joint boundary distribution]\label{thm:joint_boundary}
Suppose the Gaussian location-scale null holds, $H_{\lambda\lambda}$ is nonsingular, and $H_{cc\cdot\lambda}>0$. Then
\[
    \sqrt{n}(\widehat c_n-c_0)\Rightarrow\max\bigl\{0,\;Z_\lambda\bigr\},
    \quad Z_\lambda\sim\N\!\left(0,\,\frac{\Omega_{cc\cdot\lambda}}
    {H_{cc\cdot\lambda}^2}\right).
\]
\end{theorem}

\begin{remark}[Connection with the fixed-scale result]
Theorem~\ref{thm:joint_boundary} has the same form as Theorem~\ref{thm:boundary}: the projection of $\sqrt{n}(\widehat c_n-c_0)$
onto $[0,\infty)$ of a centred Gaussian. The only change is that the raw $c$-direction is replaced by its component orthogonal, in the local quadratic metric, to the nuisance $(\mu,\sigma)$ directions.
\end{remark}

\subsection*{Fit-improvement statistic}

Let $\widehat\lambda_{0n}\in\arg\min_\lambda Q_n(\lambda,c_0)$ and $(\widehat\lambda_n, \widehat c_n)\in\arg\min_{\lambda,c\geq c_0}
Q_n(\lambda,c)$. Define
\[
    T_n^J=n\{Q_n(\widehat\lambda_{0n},c_0)
    -Q_n(\widehat\lambda_n,\widehat c_n)\}.
\]

\begin{theorem}[Fit-improvement statistic]\label{thm:joint_fit_improvement}
Under the assumptions of Theorem~\ref{thm:joint_boundary},
\[
    \widetilde T_n^J = \frac{2H_{cc\cdot\lambda}}{\Omega_{cc\cdot\lambda}}
    \,T_n^J \Rightarrow \tfrac12\delta_0+\tfrac12\chi_1^2.
\]
An asymptotic level-$\alpha$ test rejects $H_0:c=c_0$ when
$\widetilde T_n^J>\chi^2_{1,1-2\alpha}$.
\end{theorem}

\subsection*{Why joint estimation fails in practice}

The results above require $H_{cc\cdot\lambda}>0$ and $\Omega_{cc\cdot\lambda}>0$. Both hold: from Table~\ref{tab:null_size_tail_scaling_robust}, $H_{cc\cdot\lambda}\approx0.0023$ and $\Omega_{cc\cdot\lambda}\approx0.000013$, giving normalisation constant
$2H_{cc\cdot\lambda}/\Omega_{cc\cdot\lambda} \approx355$. The near-degeneracy of $\Omega_{cc\cdot\lambda}$ and the resulting size
distortion are explained in Remark~\ref{rem:plugin_vs_joint}; simulation
evidence is in Section~\ref{app:simulations}.

\newpage

\section{Logical dependencies among the main results}\label{app:dependencies}

This appendix maps the logical structure of the paper's lemmas, propositions, theorems, and corollaries. Tables~\ref{tab:deps_one}
and~\ref{tab:deps_two} list all numbered results in order of appearance, with their direct (explicitly invoked) inputs. Three structural observations are worth noting.

\textit{Single foundational chain.}
Every result ultimately traces back to Lemma~\ref{lem:normalcdf} and
Lemma~\ref{lem:endpoint}. Lemma~\ref{lem:normalcdf} is a classical identity. Lemma~\ref{lem:endpoint} establishes that the endpoint constraints force $a=\frac{1}{2}$ and determine $\kappa(c)$ uniquely. Every result in
Section~\ref{sec:hgm_family} is a direct or indirect consequence of
Lemma~\ref{lem:endpoint}, and every result in Sections~\ref{sec:inference}--\ref{sec:plugin} depends on the family properties established there.

\textit{Theorem~\ref{thm:fit_improvement} as the logical hub.}
The central result has three direct inputs (Theorem~\ref{thm:boundary} and Lemmas~\ref{lem:Omega_pos},\,\ref{lem:tightness}) and feeds directly into four downstream results: Corollary~\ref{cor:moments_Ttilde}, Theorem~\ref{thm:local_power}, Theorem~\ref{thm:plugin_fit_improvement}, and Theorem~\ref{thm:joint_fit_improvement}. The limit distribution, the rejection rule, and the modular inheritance of the plug-in theory all follow from this result.

\textit{Clean modularity of the plug-in extension.}
Lemmas~\ref{lem:plugin_ep} and~\ref{lem:robust_ep} derive the projected kernels $K_P$ and $K_R$. All four plug-in corollaries follow by the same argument as the fixed-scale case, with $K$ replaced by $K_\diamond$. Extending the framework to a new standardisation procedure requires only verifying the corresponding projected kernel lemma; the remainder of the theory is inherited automatically.

\begin{table}[p]
\small
\centering
\setlength{\tabcolsep}{4pt}
\caption{Dependency map, Part~I: family construction
and supporting lemmas}
\label{tab:deps_one}
\begin{tabular}{p{2.8cm}p{6.4cm}p{6.0cm}}
\toprule
Result & Statement (brief) & Direct inputs \\
\midrule
\multicolumn{3}{l}{\textit{Section~\ref{sec:hgm_family}:
The hypergeometric family}}\\[3pt]
Lemma~\ref{lem:normalcdf}
  & ${}_1F_1$ representation of $\Phi$
  & Classical identity \\[2pt]
Lemma~\ref{lem:endpoint}
  & Endpoint constraints force $a=\frac12$
    and determine $\kappa(c)$ uniquely
  & --- \\[2pt]
Lemma~\ref{lem:monotonicity}
  & Monotonicity: $G_c'(z)\geq0$ for all $z$
  & Lemma~\ref{lem:endpoint} \\[2pt]
Corollary~\ref{cor:Gc_distn}
  & $G_c$ is a distribution function
  & Lemmas~\ref{lem:endpoint},
    \ref{lem:monotonicity} \\[2pt]
Corollary~\ref{cor:density}
  & Density $g_c$ exists;\ $g_{3/2}=\phi$
  & Lemmas~\ref{lem:normalcdf},
    \ref{lem:monotonicity} \\[2pt]
Lemma~\ref{lem:continuity}
  & $c\mapsto G_c(z)$ continuous;\
    $G_c\to\Phi$ uniformly as $c\downarrow\tfrac32$
  & Lemmas~\ref{lem:normalcdf},
    \ref{lem:endpoint} \\[2pt]
Lemma~\ref{lem:differentiability}
  & Differentiability of $G_c$ in $c$; score $\dot G_c(z)$
  & Lemma~\ref{lem:continuity} \\[2pt]
Lemma~\ref{lem:tail}
  & Algebraic tail:
    $g_c(z)\sim(c-\tfrac32)|z|^{-3}$
  & Lemma~\ref{lem:monotonicity}; Corollary~\ref{cor:density}\\[2pt]
Corollary~\ref{cor:survival}
  & Survival:
    $1-G_c(z)\sim\tfrac{c-3/2}{2}z^{-2}$
  & Lemma~\ref{lem:tail} \\[2pt]
Corollary~\ref{cor:moments}
  & $\E[|Z|^r]<\infty\Leftrightarrow r<2$;\
    $\E[Z]=0$
  & Lemma~\ref{lem:tail} \\[2pt]
Proposition~\ref{prop:mixture}
  & Beta-precision scale mixture:
    $Z=\varepsilon/\sqrt{\Lambda}$,\
    $\varepsilon \sim \N(0,1), \Lambda\sim\Beta(1,c-\tfrac32)$
  & Lemma~\ref{lem:monotonicity} \\[2pt]
Corollary~\ref{cor:unimodal}
  & Unimodality: $g_c$ unique mode at
    $z=0$
  & Proposition~\ref{prop:mixture} \\[2pt]
Corollary~\ref{cor:charfun}
  & Characteristic function of $G_c$
  & Proposition~\ref{prop:mixture} \\[2pt]
Corollary~\ref{cor:location_scale_family}
  & Location-scale extension $G_{\mu,\sigma,c}$
  & Lemma~\ref{lem:tail}; \\&&Corollaries~\ref{cor:Gc_distn}, ~\ref{cor:density}, ~\ref{cor:moments}\\[6pt]
\multicolumn{3}{l}{\textit{Section~\ref{sec:inference}:
Supporting lemmas}}\\[3pt]
Lemma~\ref{lem:identifiability}
  & Identification, correct specification
  & Corollary~\ref{cor:survival} \\[2pt]
Corollary~\ref{cor:consistency_correct}
  & Consistency, correct specification
  & Theorem~\ref{thm:consistency};\
    Lemma~\ref{lem:identifiability};\\ & & Corollary~\ref{cor:density} \\[2pt]
Lemma~\ref{lem:B1}
  & Verify \textup{(B1)}: $c\mapsto G_c(z)$
    real-analytic on $(1/2,\infty)$
  & Lemma~\ref{lem:differentiability} \\[2pt]
Lemma~\ref{lem:B2}
  & Verify \textup{(B2)}: $|\dot G_c|,|\ddot G_c|$
    uniformly bounded
  & --- \\[2pt]
Lemma~\ref{lem:B3}
  & $2\int_\mathbb{R}\dot G_{c^\star}(z)^2\,d\nu(z)>0$;\
    equals $H_{c^\star}$ and verifies \textup{(B3)}
    if correct specification
  & Lemma~\ref{lem:differentiability} \\[2pt]
Lemma~\ref{lem:Omega_pos}
  & $\Omega_{c^\star}>0$
  & Lemma~\ref{lem:B3} \\[2pt]
Lemma~\ref{lem:curvature}
  & $H_{c_0}>0$;\
    $Q_n''(c)\to Q''(c)$ uniformly in
    $c\in\mathcal{N}$
  & Lemma~\ref{lem:B3} \\[2pt]
Lemma~\ref{lem:tightness}
  & Tightness: $\sqrt{n}(\widehat c_n-c_0)=O_p(1)$
  & Theorem~\ref{thm:interior};\
    Lemma~\ref{lem:curvature};\\&&
    Corollary~\ref{cor:consistency_correct} \\[6pt]
\bottomrule
\end{tabular}
\begin{minipage}{\linewidth}
\smallskip\footnotesize
Notes: ``Direct inputs'' lists results whose
proofs are explicitly invoked; continued in Table~\ref{tab:deps_two}.
\end{minipage}
\end{table}

\begin{table}[p]
\small
\centering
\setlength{\tabcolsep}{4pt}
\caption{Dependency map, Part~II: inference,
plug-in procedures, and appendix results}
\label{tab:deps_two}
\begin{tabular}{p{2.8cm}p{6.6cm}p{5.8cm}}
\toprule
Result & Statement (brief) & Direct inputs \\
\midrule
\multicolumn{3}{l}{\textit{Section~\ref{sec:inference}:
Main inference results}}\\[3pt]
Theorem~\ref{thm:consistency}
  & Consistency: $\hat{c}_n\to c^*$ a.s.
  & Lemma~\ref{lem:continuity};\
    Corollary~\ref{cor:Gc_distn} \\[2pt]
Theorem~\ref{thm:interior}
  & Interior: $\sqrt{n}(\hat{c}_n-c^\star)
    \Rightarrow \N(0,\Omega_{c^\star}/H_{c^\star}^2)$
  & Theorem~\ref{thm:consistency};\\&&
    Lemmas~\ref{lem:B1}, \ref{lem:B2},
    \ref{lem:B3}, \ref{lem:Omega_pos} \\[2pt]
Theorem~\ref{thm:boundary}
  & Boundary: $\sqrt{n}(\hat{c}_n-c_0)
    \Rightarrow\max\{0,Z\}$
  & Theorem~\ref{thm:interior};\\&&
    Lemmas~\ref{lem:Omega_pos}, 
    \ref{lem:curvature}, \ref{lem:tightness} \\[2pt]
Theorem~\ref{thm:fit_improvement}
  & $\widetilde{T}_n\Rightarrow
    \tfrac12\delta_0+\tfrac12\chi_1^2$
  & Theorem~\ref{thm:boundary};\
    Lemmas~\ref{lem:Omega_pos},\,\ref{lem:tightness} \\[2pt]
Corollary~\ref{cor:moments_Ttilde}
  & $\E[\widetilde{T}_\infty]=\tfrac12$,\
    $\Var(\widetilde{T}_\infty)=\tfrac54$
  & Theorem~\ref{thm:fit_improvement} \\[2pt]
Theorem~\ref{thm:directed}
  & Directed consistency: & Theorem~\ref{thm:consistency};\\
  &
    $\Delta>0\Rightarrow\widetilde{T}_n\to_p\infty$
  & Lemmas~\ref{lem:B3}, \ref{lem:Omega_pos} \\[2pt]
Theorem~\ref{thm:local_power}
  & Local power: & Theorem~\ref{thm:fit_improvement};\\
  & $\pi_\nu(\eta)=\Phi\{\delta_\nu(\eta)
    -\Phi^{-1}(1-\alpha)\}$
  & Lemmas~\ref{lem:Omega_pos}, \ref{lem:curvature}, \ref{lem:tightness}\\[2pt]
Corollary~\ref{cor:local_hgm}
  & $\delta_\nu(\tau)=\tau \,H_{c_0}(\nu)/
    \sqrt{\Omega_{c_0}(\nu)}$
  & Theorem~\ref{thm:local_power} \\[6pt]
\multicolumn{3}{l}{\textit{Section~\ref{sec:plugin}:
Plug-in location-scale procedures}}\\[3pt]
Lemma~\ref{lem:plugin_ep}
  & Mean/SD projected empirical process;\
    kernel $K_P$
  & --- \\[2pt]
Lemma~\ref{lem:robust_ep}
  & Robust projected empirical process;\
    kernel $K_R$
  & ---  \\[2pt]
Theorem~\ref{thm:plugin_fit_improvement}
  & $\widetilde{T}_n^\diamond\Rightarrow
    \tfrac12\delta_0+\tfrac12\chi_1^2$,
    $\diamond\in\{P,R\}$
  & Theorem~\ref{thm:fit_improvement};\\&&
    Lemmas~\ref{lem:curvature},
    \ref{lem:plugin_ep},
    \ref{lem:robust_ep} \\[2pt]
Corollary~\ref{cor:plugin_moments}
  & $\E[\widetilde T_{\infty}^\diamond]=\tfrac12$,\
    $\Var(\widetilde T_{\infty}^\diamond)=\tfrac54$, $\diamond\in\{P,R\}$
  & Theorem~\ref{thm:plugin_fit_improvement}; Corollary~\ref{cor:moments_Ttilde} \\[2pt]
Corollary~\ref{cor:plugin_local_power}
  & Plug-in local power
  & Theorem~\ref{thm:local_power}\\[2pt]
Corollary~\ref{cor:plugin_consistency}
  & Plug-in consistency
  & Theorem~\ref{thm:consistency}; Lemma~\ref{lem:identifiability} \\[2pt]
Corollary~\ref{cor:plugin_directed}
  & Plug-in directed consistency
  & Corollary~\ref{cor:plugin_consistency} \\[6pt]
\multicolumn{3}{l}{\textit{Appendix~\ref{app:crossover}:
Tail dominance and Value-at-Risk (VaR)}}\\[3pt]
Lemma~\ref{lem:R_monotone}
  & $\rho(z,c)$
    strictly increasing in $z$ and $c$
  & Proposition~\ref{prop:mixture} \\[2pt]
Proposition~\ref{prop:var_divergence}
  & VaR ratio $q_p^{\mathrm{HGM}}/
    q_p^{\Phi}\to\infty$
    as $p\to1^-$
  & Corollary~\ref{cor:survival} \\[6pt]
\multicolumn{3}{l}{\textit{Appendix~\ref{app:joint}:
Joint estimation}}\\[3pt]
Theorem~\ref{thm:joint_consistency}
  & Consistency
  & Theorem~\ref{thm:consistency}
    (same argument) \\[2pt]
Theorem~\ref{thm:joint_boundary}
  & Boundary distribution
  & Theorem~\ref{thm:boundary}
    (same argument) \\[2pt]
Theorem~\ref{thm:joint_fit_improvement}
  & Fit-improvement
    statistic
  & Theorem~\ref{thm:fit_improvement}
    (same argument) \\[6pt]
\bottomrule
\end{tabular}
\begin{minipage}{\linewidth}
\smallskip\footnotesize
Notes: ``Direct inputs'' lists results whose
proofs are explicitly invoked.
\end{minipage}
\end{table}

\newpage

\bibliographystyle{plainnat}

\bibliography{ms}

\end{document}

%% file: table_null_size_tail_main_with_robust.tex
\begin{table}[!p]
\centering
\caption{Null size of the HGM tests with tail-weighted criterion}
\label{tab:null_size_tail_main_robust}
\begin{tabular}{rcccccc}
\toprule
$n$ & $10\%$ & $5\%$ & $1\%$ & $\Pr(\widehat c_n=c_0)$ & mean $\widehat c_n$ & mean $\widetilde T_n$ \\
\midrule
\multicolumn{7}{l}{Panel A: fixed-scale $G_c$ test of $\Phi$} \\
25 & 0.100 & 0.052 & 0.011 & 0.511 & 1.5898 & 0.503 \\
50 & 0.100 & 0.051 & 0.010 & 0.505 & 1.5603 & 0.500 \\
100 & 0.100 & 0.051 & 0.011 & 0.506 & 1.5413 & 0.506 \\
200 & 0.101 & 0.051 & 0.011 & 0.507 & 1.5284 & 0.503 \\
400 & 0.101 & 0.050 & 0.010 & 0.502 & 1.5198 & 0.502 \\
800 & 0.100 & 0.050 & 0.010 & 0.502 & 1.5138 & 0.502 \\
1600 & 0.100 & 0.050 & 0.010 & 0.504 & 1.5096 & 0.499 \\
\addlinespace
\multicolumn{7}{l}{Panel B: mean/standard-deviation plug-in location-scale test} \\
25 & 0.116 & 0.043 & 0.002 & 0.373 & 1.5443 & 0.565 \\
50 & 0.114 & 0.047 & 0.004 & 0.412 & 1.5290 & 0.546 \\
100 & 0.112 & 0.050 & 0.006 & 0.442 & 1.5194 & 0.537 \\
200 & 0.109 & 0.050 & 0.008 & 0.457 & 1.5133 & 0.531 \\
400 & 0.106 & 0.050 & 0.008 & 0.473 & 1.5090 & 0.517 \\
800 & 0.104 & 0.050 & 0.009 & 0.483 & 1.5062 & 0.514 \\
1600 & 0.103 & 0.050 & 0.009 & 0.491 & 1.5043 & 0.509 \\
\addlinespace
\multicolumn{7}{l}{Panel C: robust median/IQR plug-in location-scale test} \\
25 & 0.161 & 0.095 & 0.029 & 0.412 & 1.6145 & 0.815 \\
50 & 0.131 & 0.073 & 0.020 & 0.438 & 1.5659 & 0.668 \\
100 & 0.126 & 0.068 & 0.017 & 0.455 & 1.5429 & 0.632 \\
200 & 0.119 & 0.065 & 0.016 & 0.472 & 1.5282 & 0.600 \\
400 & 0.114 & 0.060 & 0.014 & 0.477 & 1.5190 & 0.569 \\
800 & 0.111 & 0.058 & 0.014 & 0.484 & 1.5130 & 0.558 \\
1600 & 0.107 & 0.057 & 0.013 & 0.492 & 1.5090 & 0.543 \\
\bottomrule
\end{tabular}
\begin{minipage}{0.95\linewidth}
\footnotesize
Notes: The table reports empirical rejection frequencies under the Gaussian null $F_0=\Phi$. All tests use the tail-weighted measure $d\nu_{1/2}(z)\propto[\Phi(z)\{1-\Phi(z)\}]^{-1/2}d\Phi(z)$. Critical values are from the boundary limit $\frac12\delta_0+\frac12\chi_1^2$. The mean--standard-deviation plug-in statistic uses the projected-process variance $\Omega_{c_0}^P(\nu_{1/2})$; the robust plug-in statistic uses the median/IQR projected-process variance $\Omega_{c_0}^R(\nu_{1/2})$. The boundary column reports $\Pr(\widehat c_n\le c_0+10^{-4})$ with $c_0=3/2$.
\end{minipage}
\end{table}

%% file: table_power_hgm_alternatives_revised.tex
\begin{table}[!p]
\centering
\caption{Power against HGM alternatives}
\label{tab:power_hgm_alternatives_revised}
\begin{tabular}{llcccccc}
\toprule
Alternative & Test & $n=25$ & $n=50$ & $n=100$ & $n=200$ & $n=400$ & $n=800$ \\
\midrule
$c=1.55$ & Fixed scale & 0.088 & 0.108 & 0.135 & 0.186 & 0.275 & 0.413 \\
 & Robust plug-in & 0.104 & 0.078 & 0.083 & 0.078 & 0.087 & 0.099 \\
 & JB & 0.117 & 0.207 & 0.335 & 0.535 & 0.745 & 0.922 \\
 & AD & 0.103 & 0.143 & 0.200 & 0.308 & 0.465 & 0.684 \\
 & SW & 0.129 & 0.199 & 0.303 & 0.484 & 0.700 & 0.899 \\
\addlinespace
$c=1.60$ & Fixed scale & 0.146 & 0.187 & 0.262 & 0.398 & 0.618 & 0.855 \\
 & Robust plug-in & 0.107 & 0.091 & 0.100 & 0.107 & 0.129 & 0.167 \\
 & JB & 0.184 & 0.329 & 0.518 & 0.745 & 0.931 & 0.995 \\
 & AD & 0.153 & 0.236 & 0.344 & 0.527 & 0.759 & 0.943 \\
 & SW & 0.188 & 0.309 & 0.474 & 0.702 & 0.906 & 0.992 \\
\addlinespace
$c=1.75$ & Fixed scale & 0.319 & 0.467 & 0.694 & 0.910 & 0.993 & 1.000 \\
 & Robust plug-in & 0.125 & 0.124 & 0.144 & 0.178 & 0.258 & 0.390 \\
 & JB & 0.308 & 0.533 & 0.778 & 0.944 & 0.997 & 1.000 \\
 & AD & 0.259 & 0.408 & 0.622 & 0.841 & 0.977 & 1.000 \\
 & SW & 0.307 & 0.501 & 0.736 & 0.920 & 0.995 & 1.000 \\
\addlinespace
$c=2.00$ & Fixed scale & 0.590 & 0.814 & 0.967 & 1.000 & 1.000 & 1.000 \\
 & Robust plug-in & 0.157 & 0.159 & 0.215 & 0.321 & 0.486 & 0.730 \\
 & JB & 0.419 & 0.682 & 0.902 & 0.990 & 1.000 & 1.000 \\
 & AD & 0.364 & 0.572 & 0.810 & 0.967 & 0.999 & 1.000 \\
 & SW & 0.417 & 0.648 & 0.879 & 0.985 & 1.000 & 1.000 \\
\addlinespace
$c=3.00$ & Fixed scale & 0.964 & 0.999 & 1.000 & 1.000 & 1.000 & 1.000 \\
 & Robust plug-in & 0.199 & 0.266 & 0.395 & 0.603 & 0.859 & 0.986 \\
 & JB & 0.524 & 0.807 & 0.967 & 0.999 & 1.000 & 1.000 \\
 & AD & 0.495 & 0.749 & 0.942 & 0.997 & 1.000 & 1.000 \\
 & SW & 0.535 & 0.790 & 0.959 & 0.999 & 1.000 & 1.000 \\
\addlinespace
\bottomrule
\end{tabular}
\begin{minipage}{0.95\linewidth}
\footnotesize
Notes: Entries are rejection frequencies at nominal level $5\%$ based on $R=10^4$ replications, against HGM alternatives $G_c: c \in \{1.55, 1.60, 1.75, 2.00, 3.00\}$. The reported HGM procedures are the fixed-scale test $\widetilde{T}_n$ and the robust median/IQR plug-in test $\widetilde{T}_n^R$; both tests use the tail-weighted criterion $d\nu_{1/2}(z)\propto[\Phi(z)\{1-\Phi(z)\}]^{-1/2}d\Phi(z)$ and asymptotic boundary critical values. Benchmark tests are Jarque–Bera (JB), Anderson–Darling (AD), and Shapiro–Wilk (SW).
\end{minipage}
\end{table}

%% file: table_power_standard_alternatives_revised.tex
\begin{table}[!p]
\centering
\caption{Power against standard non-HGM alternatives}
\label{tab:power_standard_alternatives_revised}
\begin{tabular}{llcccccc}
\toprule
Alternative & Test & $n=25$ & $n=50$ & $n=100$ & $n=200$ & $n=400$ & $n=800$ \\
\midrule
$\chi^2_3$ & Fixed scale & 0.008 & 0.004 & 0.001 & 0.000 & 0.000 & 0.000 \\
 & Robust plug-in & 0.041 & 0.018 & 0.008 & 0.004 & 0.000 & 0.000 \\
 & JB & 0.466 & 0.860 & 1.000 & 1.000 & 1.000 & 1.000 \\
 & AD & 0.675 & 0.958 & 1.000 & 1.000 & 1.000 & 1.000 \\
 & SW & 0.782 & 0.987 & 1.000 & 1.000 & 1.000 & 1.000 \\
\addlinespace
$t_{10}$ & Fixed scale & 0.031 & 0.024 & 0.014 & 0.008 & 0.004 & 0.000 \\
 & Robust plug-in & 0.122 & 0.109 & 0.123 & 0.155 & 0.206 & 0.303 \\
 & JB & 0.100 & 0.172 & 0.293 & 0.443 & 0.667 & 0.894 \\
 & AD & 0.080 & 0.098 & 0.148 & 0.216 & 0.373 & 0.660 \\
 & SW & 0.109 & 0.152 & 0.233 & 0.357 & 0.569 & 0.834 \\
\addlinespace
$t_{3}$ & Fixed scale & 0.000 & 0.000 & 0.000 & 0.000 & 0.000 & 0.000 \\
 & Robust plug-in & 0.195 & 0.248 & 0.380 & 0.569 & 0.818 & 0.977 \\
 & JB & 0.394 & 0.661 & 0.896 & 0.992 & 1.000 & 1.000 \\
 & AD & 0.364 & 0.592 & 0.837 & 0.982 & 1.000 & 1.000 \\
 & SW & 0.403 & 0.639 & 0.874 & 0.989 & 1.000 & 1.000 \\
\addlinespace
$t_{5}$ & Fixed scale & 0.011 & 0.006 & 0.002 & 0.000 & 0.000 & 0.000 \\
 & Robust plug-in & 0.144 & 0.155 & 0.207 & 0.302 & 0.469 & 0.711 \\
 & JB & 0.211 & 0.393 & 0.626 & 0.863 & 0.983 & 0.999 \\
 & AD & 0.170 & 0.277 & 0.451 & 0.716 & 0.936 & 0.998 \\
 & SW & 0.217 & 0.354 & 0.561 & 0.816 & 0.971 & 0.999 \\
\addlinespace
Cont. normal & Fixed scale & 0.010 & 0.006 & 0.001 & 0.000 & 0.000 & 0.000 \\
 & Robust plug-in & 0.117 & 0.094 & 0.098 & 0.116 & 0.141 & 0.187 \\
 & JB & 0.232 & 0.411 & 0.642 & 0.855 & 0.977 & 0.999 \\
 & AD & 0.174 & 0.269 & 0.405 & 0.599 & 0.831 & 0.975 \\
 & SW & 0.230 & 0.381 & 0.587 & 0.815 & 0.964 & 0.999 \\
\addlinespace
Laplace & Fixed scale & 0.010 & 0.004 & 0.001 & 0.000 & 0.000 & 0.000 \\
 & Robust plug-in & 0.272 & 0.369 & 0.566 & 0.804 & 0.968 & 0.999 \\
 & JB & 0.290 & 0.507 & 0.783 & 0.968 & 0.999 & 1.000 \\
 & AD & 0.302 & 0.517 & 0.807 & 0.982 & 1.000 & 1.000 \\
 & SW & 0.321 & 0.520 & 0.798 & 0.977 & 1.000 & 1.000 \\
\addlinespace
\bottomrule
\end{tabular}
\begin{minipage}{0.95\linewidth}
\footnotesize
Notes: Entries are rejection frequencies at nominal level $5\%$ based on $R=10^4$ replications. Continuous alternatives are standardised to mean zero and unit variance before testing; this is relevant for the fixed-scale test, which assumes known location and scale, but immaterial for the plug-in and omnibus tests, which are location-scale invariant. The contaminated normal is the unit-variance version of $0.95\N(0,1)+0.05\N(0,9)$.
\end{minipage}
\end{table}

%% file: table_null_size_tail_scaling_with_robust.tex
\begin{table}[!h]
\centering
\caption{Scaling constants and asymptotic critical values}
\label{tab:null_size_tail_scaling_robust}
\begin{tabular}{lc}
\toprule
Quantity & Value \\
\midrule
$H_{c_0}(\nu_{1/2})$ & 0.036295 \\
$\Omega_{c_0}(\nu_{1/2})$ & 0.001191 \\
$2H_{c_0}/\Omega_{c_0}$ & 60.952944 \\
$\Omega_{c_0}^P(\nu_{1/2})$ & 0.000231 \\
$2H_{c_0}/\Omega_{c_0}^P$ & 313.567929 \\
$\Omega_{c_0}^R(\nu_{1/2})$ & 0.000925 \\
$2H_{c_0}/\Omega_{c_0}^R$ & 78.469455 \\
$H_{cc\cdot\lambda}$ & 0.002282 \\
$\Omega_{cc\cdot\lambda}$ & 0.000013 \\
$2H_{cc\cdot\lambda}/\Omega_{cc\cdot\lambda}$ & 355.448213 \\
$\chi^2_{1,0.80}$ & 1.642374 \\
$\chi^2_{1,0.90}$ & 2.705543 \\
$\chi^2_{1,0.98}$ & 5.411894 \\
\bottomrule
\end{tabular}
\begin{minipage}{0.85\linewidth}
\footnotesize
Notes: Constants are computed for $d\nu_{1/2}(z)\propto[\Phi(z)\{1-\Phi(z)\}]^{-1/2}d\Phi(z)$. The plug-in and robust plug-in statistics use their corresponding projected-process variances.
\end{minipage}
\end{table}

%% file: table_null_size_tail_robust_plugin_invariance_N3_2sq.tex
\begin{table}[!htbp]
\centering
\caption{Robustness check: robust plug-in HGM test under a non-standard Gaussian null}
\label{tab:null_size_tail_robust_plugin_invariance_N3_2sq}
\begin{tabular}{rcccccc}
\toprule
& \multicolumn{3}{c}{$\N(0,1)$} & \multicolumn{3}{c}{$\N(3,4)$} \\
\cmidrule(lr){2-4}\cmidrule(lr){5-7}
$n$ & $5\%$ & $1\%$ & $\Pr(\widehat c_n=c_0)$ & $5\%$ & $1\%$ & $\Pr(\widehat c_n=c_0)$ \\
\midrule
100 & 0.068 & 0.017 & 0.455 & 0.068 & 0.017 & 0.457 \\
400 & 0.060 & 0.014 & 0.477 & 0.061 & 0.015 & 0.477 \\
1600 & 0.057 & 0.013 & 0.492 & 0.055 & 0.012 & 0.492 \\
\bottomrule
\end{tabular}
\begin{minipage}{0.95\linewidth}
\footnotesize
Notes: The table compares the robust plug-in statistic under $\N(0,1)$ and $\N(3,4)$. The same tail-weighted criterion and robust projected-process scaling are used throughout. The boundary columns report $\Pr(\widehat c_n\le c_0+10^{-4})$ with $c_0=3/2$.
\end{minipage}
\end{table}

%% file: table_null_size_tail_exact_profile_appendix.tex
\begin{table}[!htbp]
\centering
\caption{Null size of the exact-profiled location-scale HGM statistic}
\label{tab:null_size_tail_exact_profile_appendix}
\begin{tabular}{rccccccc}
\toprule
$n$ & $10\%$ & $5\%$ & $1\%$ & $\Pr(\widehat c=c_0)$ & mean $\widehat c$ & mean $\widetilde T_n$ & $q_{.95}(\widetilde T_n)$ \\
\midrule
50 & 0.112 & 0.075 & 0.018 & 0.538 & 1.6381 & 0.581 & 3.557 \\
100 & 0.095 & 0.063 & 0.022 & 0.458 & 1.6098 & 0.588 & 3.485 \\
200 & 0.089 & 0.057 & 0.026 & 0.357 & 1.5796 & 0.652 & 3.109 \\
400 & 0.132 & 0.073 & 0.027 & 0.248 & 1.5598 & 0.904 & 3.451 \\
800 & 0.236 & 0.111 & 0.019 & 0.215 & 1.5316 & 1.092 & 3.776 \\
1600 & 0.309 & 0.193 & 0.050 & 0.180 & 1.5208 & 1.448 & 5.403 \\
\bottomrule
\end{tabular}
\begin{minipage}{0.95\linewidth}
\footnotesize
Notes: Empirical rejection frequencies of $\widetilde T_n^{EP}$ under $\N(0,1)$, $R=10^3$ replications. Rejection at $\widetilde T_n^{EP}>\chi^2_{1,1-2\alpha}$. Column $\Pr(\widehat c_n=c_0)$ reports $\Pr(\widehat c_n\leq c_0+10^{-4})$. Scaling constant $2H_{cc\cdot\lambda}/\Omega_{cc\cdot\lambda}\approx355$ (Table~\ref{tab:null_size_tail_scaling_robust}). Size distortion worsens with $n$ due to the near-degeneracy of $\Omega_{cc\cdot\lambda}$; see Remark~\ref{rem:plugin_vs_joint}.
\end{minipage}
\end{table}

%% file: table_appendix_mean_plugin_hgm_power.tex
\begin{table}[!p]
\centering
\caption{Mean/SD plug-in HGM power against HGM alternatives}
\label{tab:appendix_mean_plugin_hgm_power}
\begin{tabular}{llcccccc}
\toprule
Alternative & Test & $n=25$ & $n=50$ & $n=100$ & $n=200$ & $n=400$ & $n=800$ \\
\midrule
$c=1.55$ & $\widetilde{T}_n^P$ & 0.036 & 0.032 & 0.022 & 0.014 & 0.004 & 0.001 \\
\addlinespace
$c=1.60$ & $\widetilde{T}_n^P$ & 0.025 & 0.021 & 0.011 & 0.003 & 0.000 & 0.000 \\
\addlinespace
$c=1.75$ & $\widetilde{T}_n^P$ & 0.020 & 0.008 & 0.002 & 0.000 & 0.000 & 0.000 \\
\addlinespace
$c=2.00$ & $\widetilde{T}_n^P$ & 0.009 & 0.003 & 0.000 & 0.000 & 0.000 & 0.000 \\
\addlinespace
$c=3.00$ & $\widetilde{T}_n^P$ & 0.004 & 0.001 & 0.000 & 0.000 & 0.000 & 0.000 \\
\addlinespace
\bottomrule
\end{tabular}
\begin{minipage}{0.95\linewidth}
\footnotesize
Notes: The mean/sd plug-in statistic is reported as an appendix diagnostic.
\end{minipage}
\end{table}

%% file: table_appendix_mean_plugin_standard_power.tex
\begin{table}[!p]
\centering
\caption{Mean/SD plug-in HGM power against standard alternatives}
\label{tab:appendix_mean_plugin_standard_power}
\begin{tabular}{llcccccc}
\toprule
Alternative & Test & $n=25$ & $n=50$ & $n=100$ & $n=200$ & $n=400$ & $n=800$ \\
\midrule
$\chi^2_3$ & $\widetilde{T}_n^P$ & 0.034 & 0.015 & 0.004 & 0.001 & 0.000 & 0.000 \\
\addlinespace
$t_{10}$ & $\widetilde{T}_n^P$ & 0.020 & 0.017 & 0.006 & 0.003 & 0.001 & 0.000 \\
\addlinespace
$t_{3}$ & $\widetilde{T}_n^P$ & 0.004 & 0.001 & 0.000 & 0.000 & 0.000 & 0.000 \\
\addlinespace
$t_{5}$ & $\widetilde{T}_n^P$ & 0.011 & 0.005 & 0.001 & 0.000 & 0.000 & 0.000 \\
\addlinespace
Cont. normal & $\widetilde{T}_n^P$ & 0.026 & 0.018 & 0.008 & 0.002 & 0.000 & 0.000 \\
\addlinespace
Laplace & $\widetilde{T}_n^P$ & 0.002 & 0.000 & 0.000 & 0.000 & 0.000 & 0.000 \\
\addlinespace
\bottomrule
\end{tabular}
\begin{minipage}{0.95\linewidth}
\footnotesize
Notes: The mean/sd plug-in statistic is reported as an appendix diagnostic.
\end{minipage}
\end{table}

%% file: table_appendix_sizeadjusted_robust_hgm_power.tex
\begin{table}[!p]
\centering
\caption{Size-adjusted robust plug-in HGM power against HGM alternatives}
\label{tab:appendix_sizeadj_robust_hgm_power}
\begin{tabular}{llcccccc}
\toprule
Alternative & Test & $n=25$ & $n=50$ & $n=100$ & $n=200$ & $n=400$ & $n=800$ \\
\midrule
$c=1.55$ & $\widetilde{T}_n^R$ (size-adj.) & 0.053 & 0.054 & 0.062 & 0.064 & 0.074 & 0.087 \\
\addlinespace
$c=1.60$ & $\widetilde{T}_n^R$ (size-adj.) & 0.056 & 0.059 & 0.075 & 0.088 & 0.110 & 0.149 \\
\addlinespace
$c=1.75$ & $\widetilde{T}_n^R$ (size-adj.) & 0.066 & 0.083 & 0.108 & 0.149 & 0.232 & 0.361 \\
\addlinespace
$c=2.00$ & $\widetilde{T}_n^R$ (size-adj.) & 0.085 & 0.107 & 0.169 & 0.277 & 0.452 & 0.705 \\
\addlinespace
$c=3.00$ & $\widetilde{T}_n^R$ (size-adj.) & 0.112 & 0.195 & 0.327 & 0.557 & 0.838 & 0.981 \\
\addlinespace
\bottomrule
\end{tabular}
\begin{minipage}{0.95\linewidth}
\footnotesize
Notes: $\widetilde{T}_n^R$ uses finite-sample null 5\% critical values from the null-size simulations.
\end{minipage}
\end{table}

%% file: table_appendix_sizeadjusted_robust_standard_power.tex
\begin{table}[!p]
\centering
\caption{Size-adjusted robust plug-in HGM power against standard alternatives}
\label{tab:appendix_sizeadj_robust_standard_power}
\begin{tabular}{llcccccc}
\toprule
Alternative & Test & $n=25$ & $n=50$ & $n=100$ & $n=200$ & $n=400$ & $n=800$ \\
\midrule
$\chi^2_3$ & $\widetilde{T}_n^R$ (size-adj.) & 0.018 & 0.011 & 0.004 & 0.003 & 0.000 & 0.000 \\
\addlinespace
$t_{10}$ & $\widetilde{T}_n^R$ (size-adj.) & 0.062 & 0.074 & 0.090 & 0.131 & 0.181 & 0.276 \\
\addlinespace
$t_{3}$ & $\widetilde{T}_n^R$ (size-adj.) & 0.110 & 0.181 & 0.314 & 0.520 & 0.796 & 0.974 \\
\addlinespace
$t_{5}$ & $\widetilde{T}_n^R$ (size-adj.) & 0.076 & 0.109 & 0.161 & 0.262 & 0.433 & 0.686 \\
\addlinespace
Cont. normal & $\widetilde{T}_n^R$ (size-adj.) & 0.059 & 0.066 & 0.070 & 0.092 & 0.123 & 0.170 \\
\addlinespace
Laplace & $\widetilde{T}_n^R$ (size-adj.) & 0.166 & 0.291 & 0.498 & 0.771 & 0.960 & 0.999 \\
\addlinespace
\bottomrule
\end{tabular}
\begin{minipage}{0.95\linewidth}
\footnotesize
Notes: $\widetilde{T}_n^R$ uses finite-sample null 5\% critical values from the null-size simulations.
\end{minipage}
\end{table}

%% file: table_sp500_returns_desc.tex
\begin{table}[p]
\centering
\caption{Descriptive statistics: S\&P 500 daily log returns}
\label{tab:sp500_returns_desc}
\begin{tabular}{lc}
\toprule
Statistic & Value \\
\midrule
Observations & $1255.0000$ \\
Mean & $0.0472$ \\
Standard deviation & $1.0635$ \\
Median & $0.0824$ \\
IQR scale & $0.7999$ \\
Minimum & $-6.1609$ \\
Maximum & $9.0895$ \\
Skewness & $0.0215$ \\
Excess kurtosis & $6.5374$ \\
\bottomrule
\end{tabular}
\end{table}

%% file: table_sp500_returns_params.tex
\begin{table}[p]
\centering
\caption{Fitted distribution parameters: S\&P 500 daily log returns}
\label{tab:sp500_returns_params}
\begin{tabular}{lcccc}
\toprule
Model & Location & Scale & Shape \\
\midrule
Gaussian & $0.0472$ & $1.0630$ & — \\
Laplace & $0.0824$ & $0.7530$ & — \\
Student $t$ & $0.0789$ & $0.7529$ & $k$=3.8667 \\
HGM & $0.0824$ & $0.7999$ & $c$=1.6358 \\
\bottomrule
\end{tabular}
\end{table}

%% file: table_sp500_returns_tests.tex
\begin{table}[p]
\centering
\caption{Normality and HGM tests: S\&P 500 daily log returns}
\label{tab:sp500_returns_tests}
\begin{tabular}{lcccl}
\toprule
Test & Statistic & $p$-value \\
\midrule
Jarque--Bera & $2234.9388$ & $0.0000$ \\
Anderson--Darling & $10.9142$ & \textemdash \\
Shapiro--Wilk & $0.9450$ & $0.0000$ \\
Robust plug-in HGM & $26.7864$ & $0.0000$ \\
\bottomrule
\end{tabular}
\end{table}

%% file: table_sp500_returns_gof.tex
\begin{table}[p]
\centering
\caption{Goodness-of-fit measures: S\&P 500 daily log returns}
\label{tab:sp500_returns_gof}
\begin{tabular}{lccccl}
\toprule
Model & Mean log score & KS distance & CvM distance & Tail-weighted CvM \\
\midrule
Student $t$ & $-1.40734$ & $0.02050$ & $0.00007$ & $0.00006$ \\
Laplace & $-1.40947$ & $0.01968$ & $0.00008$ & $0.00007$ \\
HGM & $-1.42062$ & $0.02698$ & $0.00020$ & $0.00017$ \\
Gaussian & $-1.48007$ & $0.07051$ & $0.00153$ & $0.00120$ \\
\bottomrule
\end{tabular}
\end{table}

%% file: table_sp500_returns_var.tex
\begin{table}[p]
\centering
\caption{Lower-tail quantiles: S\&P 500 daily log returns}
\label{tab:sp500_returns_var}
\begin{tabular}{lcc}
\toprule
Model & Probability & Lower quantile \\
\midrule
Gaussian & $0.0100$ & $-2.4258$ \\
Gaussian & $0.0050$ & $-2.6910$ \\
Gaussian & $0.0010$ & $-3.2378$ \\
Laplace & $0.0100$ & $-2.8634$ \\
Laplace & $0.0050$ & $-3.3853$ \\
Laplace & $0.0010$ & $-4.5972$ \\
Student $t$ & $0.0100$ & $-2.7967$ \\
Student $t$ & $0.0050$ & $-3.4705$ \\
Student $t$ & $0.0010$ & $-5.5148$ \\
HGM & $0.0100$ & $-2.3371$ \\
HGM & $0.0050$ & $-3.0390$ \\
HGM & $0.0010$ & $-6.5740$ \\
Empirical & $0.0100$ & $-2.9948$ \\
Empirical & $0.0050$ & $-3.6009$ \\
Empirical & $0.0010$ & $-4.8233$ \\
\bottomrule
\end{tabular}
\end{table}

%% file: table_sp500_returns_tail.tex
\begin{table}[p]
\centering
\caption{Fitted tail probabilities: S\&P 500 daily log returns}
\label{tab:sp500_returns_tail}
\begin{tabular}{lcccc}
\toprule
Model & z & $\Pr(Y<-z)$ & $\Pr(Y>z)$ & $\Pr(|Y|>z)$ \\
\midrule
Gaussian & $2.00000$ & $0.07054$ & $0.06202$ & $0.13257$ \\
Gaussian & $3.00000$ & $0.01307$ & $0.01100$ & $0.02407$ \\
Gaussian & $4.00000$ & $0.00146$ & $0.00117$ & $0.00263$ \\
Laplace & $2.00000$ & $0.05975$ & $0.05975$ & $0.11950$ \\
Laplace & $3.00000$ & $0.02065$ & $0.02065$ & $0.04131$ \\
Laplace & $4.00000$ & $0.00714$ & $0.00714$ & $0.01428$ \\
Student $t$ & $2.00000$ & $0.05186$ & $0.05132$ & $0.10318$ \\
Student $t$ & $3.00000$ & $0.01754$ & $0.01739$ & $0.03493$ \\
Student $t$ & $4.00000$ & $0.00710$ & $0.00705$ & $0.01415$ \\
HGM & $2.00000$ & $0.04044$ & $0.04044$ & $0.08088$ \\
HGM & $3.00000$ & $0.01026$ & $0.01026$ & $0.02053$ \\
HGM & $4.00000$ & $0.00472$ & $0.00472$ & $0.00944$ \\
Empirical & $2.00000$ & $0.06693$ & $0.04701$ & $0.11394$ \\
Empirical & $3.00000$ & $0.02151$ & $0.01195$ & $0.03347$ \\
Empirical & $4.00000$ & $0.00876$ & $0.00159$ & $0.01036$ \\
\bottomrule
\end{tabular}
\end{table}

%% file: table_sp500_garch_residuals_desc.tex
\begin{table}[p]
\centering
\caption{Descriptive statistics: GARCH(1,1) standardised residuals}
\label{tab:sp500_garch_residuals_desc}
\begin{tabular}{lc}
\toprule
Statistic & Value \\
\midrule
Observations & $1255.0000$ \\
Mean & $-0.0403$ \\
Standard deviation & $0.9970$ \\
Median & $0.0050$ \\
IQR scale & $0.8768$ \\
Minimum & $-4.8092$ \\
Maximum & $3.7525$ \\
Skewness & $-0.4936$ \\
Excess kurtosis & $1.5648$ \\
\bottomrule
\end{tabular}
\end{table}

%% file: table_sp500_garch11_parameters.tex
\begin{table}[p]
\centering
\caption{GARCH(1,1) parameter estimates}
\label{tab:garch11_params}
\begin{tabular}{lccccccc}
\toprule
$\widehat\mu$ & $\widehat\omega$ & $\widehat{\alpha}$ & $\widehat{\beta}$ & $\widehat{\alpha}+\widehat{\beta}$ & Uncond.\ var. & Log like. \\
\midrule
$0.077947$ & $0.034779$ & $0.104845$ & $0.862822$ & $0.967667$ & $1.075629$ & $-1708.564423$ \\
\bottomrule
\end{tabular}
\end{table}

%% file: table_sp500_garch_residuals_params.tex
\begin{table}[p]
\centering
\caption{Fitted distribution parameters: GARCH(1,1) standardised residuals}
\label{tab:sp500_garch_residuals_params}
\begin{tabular}{lcccc}
\toprule
Model & Location & Scale & Shape \\
\midrule
Gaussian & $-0.0403$ & $0.9966$ & — \\
Laplace & $0.0050$ & $0.7531$ & — \\
Student $t$ & $-0.0046$ & $0.8432$ & $k$=6.9703 \\
HGM & $0.0050$ & $0.8768$ & $c$=1.5625 \\
\bottomrule
\end{tabular}
\end{table}

%% file: table_sp500_garch_residuals_tests.tex
\begin{table}[p]
\centering
\caption{Normality and HGM tests: GARCH(1,1) standardised residuals}
\label{tab:sp500_garch_residuals_tests}
\begin{tabular}{lcccl}
\toprule
Test & Statistic & $p$-value \\
\midrule
Jarque--Bera & $178.9911$ & $0.0000$ \\
Anderson--Darling & $4.3887$ & \textemdash \\
Shapiro--Wilk & $0.9809$ & $0.0000$ \\
Robust plug-in HGM & $6.3163$ & $0.0060$ \\
\bottomrule
\end{tabular}
\end{table}

%% file: table_sp500_garch_residuals_gof.tex
\begin{table}[p]
\centering
\caption{Goodness-of-fit measures: GARCH(1,1) standardised residuals}
\label{tab:sp500_garch_residuals_gof}
\begin{tabular}{lccccl}
\toprule
Model & Mean log score & KS distance & CvM distance & Tail-weighted CvM \\
\midrule
Student $t$ & $-1.39649$ & $0.02211$ & $0.00008$ & $0.00009$ \\
HGM & $-1.40096$ & $0.02969$ & $0.00013$ & $0.00013$ \\
Laplace & $-1.40955$ & $0.02872$ & $0.00029$ & $0.00026$ \\
Gaussian & $-1.41557$ & $0.04927$ & $0.00059$ & $0.00045$ \\
\bottomrule
\end{tabular}
\end{table}

%% file: table_sp500_garch_residuals_var.tex
\begin{table}[p]
\centering
\caption{Lower-tail quantiles: GARCH(1,1) standardised residuals}
\label{tab:sp500_garch_residuals_var}
\begin{tabular}{lcc}
\toprule
Model & Probability & Lower quantile \\
\midrule
Gaussian & $0.0100$ & $-2.3588$ \\
Gaussian & $0.0050$ & $-2.6074$ \\
Gaussian & $0.0010$ & $-3.1201$ \\
Laplace & $0.0100$ & $-2.9410$ \\
Laplace & $0.0050$ & $-3.4630$ \\
Laplace & $0.0010$ & $-4.6750$ \\
Student $t$ & $0.0100$ & $-2.5355$ \\
Student $t$ & $0.0050$ & $-2.9597$ \\
Student $t$ & $0.0010$ & $-4.0482$ \\
HGM & $0.0100$ & $-2.2896$ \\
HGM & $0.0050$ & $-2.6935$ \\
HGM & $0.0010$ & $-5.0123$ \\
Empirical & $0.0100$ & $-2.5777$ \\
Empirical & $0.0050$ & $-3.2448$ \\
Empirical & $0.0010$ & $-4.6984$ \\
\bottomrule
\end{tabular}
\end{table}

%% file: table_sp500_garch_residuals_tail.tex
\begin{table}[p]
\centering
\caption{Fitted tail probabilities: GARCH(1,1) standardised residuals}
\label{tab:sp500_garch_residuals_tail}
\begin{tabular}{lcccc}
\toprule
Model & z & $\Pr(Y<-z)$ & $\Pr(Y>z)$ & $\Pr(|Y|>z)$ \\
\midrule
Gaussian & $2.00000$ & $0.04325$ & $0.03553$ & $0.07878$ \\
Gaussian & $3.00000$ & $0.00474$ & $0.00363$ & $0.00837$ \\
Gaussian & $4.00000$ & $0.00026$ & $0.00018$ & $0.00044$ \\
Laplace & $2.00000$ & $0.04871$ & $0.04871$ & $0.09742$ \\
Laplace & $3.00000$ & $0.01520$ & $0.01520$ & $0.03041$ \\
Laplace & $4.00000$ & $0.00475$ & $0.00475$ & $0.00949$ \\
Student $t$ & $2.00000$ & $0.03879$ & $0.03750$ & $0.07629$ \\
Student $t$ & $3.00000$ & $0.00861$ & $0.00834$ & $0.01695$ \\
Student $t$ & $4.00000$ & $0.00217$ & $0.00211$ & $0.00428$ \\
HGM & $2.00000$ & $0.03122$ & $0.03122$ & $0.06243$ \\
HGM & $3.00000$ & $0.00554$ & $0.00554$ & $0.01107$ \\
HGM & $4.00000$ & $0.00221$ & $0.00221$ & $0.00442$ \\
Empirical & $2.00000$ & $0.05737$ & $0.02072$ & $0.07809$ \\
Empirical & $3.00000$ & $0.00956$ & $0.00239$ & $0.01195$ \\
Empirical & $4.00000$ & $0.00398$ & $0.00159$ & $0.00558$ \\
\bottomrule
\end{tabular}
\end{table}

%% file: table_correction_factor.tex
\begin{table}[!htbp]
\centering
\caption{Tail probability correction factor $\rho(z,c)=(1-G_c(z))/(1-\Phi(z))$}
\label{tab:correction}
\begin{tabular}{lcccc}
\toprule
& \multicolumn{4}{c}{$c - 3/2$} \\
\cmidrule(lr){2-5}
$z$ & $0.05$ & $0.10$ & $0.20$ & $0.30$ \\
\midrule
$2.0$ & 1.300 & 1.583 & 2.108 & 2.584 \\
$2.5$ & 1.785 & 2.538 & 3.958 & 5.277 \\
$3.0$ & 3.491 & 5.910 & 10.553 & 14.960 \\
$3.5$ & 11.181 & 21.167 & 40.591 & 59.337 \\
$4.0$ & 56.071 & 110.428 & 217.086 & 321.133 \\
\bottomrule
\end{tabular}
\begin{minipage}{\linewidth}
\smallskip\footnotesize
Notes: $\rho(z,c)=(1-G_c(z))/(1-\Phi(z))$ computed from the exact HGM survival function. $\rho>1$ for all $z>0$: the HGM distribution assigns strictly more probability to tail events than the Gaussian at every threshold. $\rho$ is strictly increasing in both $z$ and $c$ (Lemma~\ref{lem:R_monotone}).
\end{minipage}
\end{table}

%% file: table_var_ratio.tex
\begin{table}[!htbp]
\centering
\caption{Exact HGM versus Gaussian VaR for $c-3/2=0.30$}
\label{tab:var_ratio}
\begin{tabular}{lccc}
\toprule
$p$ & $q_p^{\Phi}$ & $q_p^{\mathrm{HGM}}$ & $q_p^{\mathrm{HGM}} / q_p^{\Phi}$ \\
\midrule
$0.990$ & 2.326 & 4.031 & 1.73 \\
$0.995$ & 2.576 & 5.579 & 2.17 \\
$0.999$ & 3.090 & 12.291 & 3.98 \\
\bottomrule
\end{tabular}
\begin{minipage}{\linewidth}
\smallskip\footnotesize
Notes: $q_p^{\mathrm{HGM}}$ is the exact HGM quantile, computed by numerical inversion of $1-G_c(q_p)=1-p$. $q_p^{\Phi}=\Phi^{-1}(p)$ is the Gaussian quantile.
\end{minipage}
\end{table}

%% file: ms.bib
@article{abadir93,
    author      = {Abadir, K. M.},
    title       = {The limiting distribution of the autocorrelation coefficient under a unit root},
    journal     = {Annals of Statistics},
    volume      = {21},
    number      = {2},
    pages       = {1058--1070},
    year        = {1993}
}

@article{abadir95,
    author      = {Abadir, K. M.},
    title       = {The limiting distribution of the $t$ ratio under a unit root},
    journal     = {Econometric Theory},
    volume      = {11},
    number      = {4},
    pages       = {775--793},
    year        = {1995}
}

@article{abadir99,
    author      = {Abadir, K. M.},
    title       = {An introduction to hypergeometric functions for economists},
    journal     = {Econometric Reviews},
    volume      = {18},
    number      = {3},
    pages       = {287--330},
    year        = {1999}
}

@article{abadir_rockinger03,
    author      = {Abadir, K. M. and Rockinger, M.},
    title       = {Density functionals, with an option-pricing application},
    journal     = {Econometric Theory},
    volume      = {19},
    number      = {5},
    pages       = {778--811},
    year        = {2003}
}

@book{abramowitz_stegun72,
    editor      = {Abramowitz, M. and Stegun, I. A.},
    title       = {Handbook of Mathematical Functions with Formulas, Graphs, and Mathematical Tables},
    publisher   = {Dover Publications},
    address     = {New York},
    year        = {1972}
}

@article{anderson_darling52,
    author      = {Anderson, T. W. and Darling, D. A.},
    title       = {Asymptotic theory of certain ``goodness of fit'' criteria based on stochastic processes},
    journal     = {Annals of Mathematical Statistics},
    volume      = {23},
    number      = {2},
    pages       = {193--212},
    year        = {1952}
}

@article{andrews01,
    author      = {Andrews, D. W. K.},
    title       = {Testing when a parameter is on the boundary of the maintained hypothesis},
    journal     = {Econometrica},
    volume      = {69},
    number      = {3},
    pages       = {683--734},
    year        = {2001}
}

@article{andrews_mallows74,
    author      = {Andrews, D. F. and Mallows, C. L.},
    title       = {Scale mixtures of normal distributions},
    journal     = {Journal of the Royal Statistical Society, Series B},
    volume      = {36},
    number      = {1},
    pages       = {99--102},
    year        = {1974}
}

@article{bahadur66,
        author  = {Bahadur, R. R.},
        title   = {A note on quantiles in large samples},
        journal = {Annals of Mathematical Statistics},
        volume  = {37},
        number  = {3},
        pages   = {577--580},
        year    = {1966}
}

@article{bai03,
    author      = {Bai, J.},
    title       = {Testing parametric conditional distributions of dynamic models},
    journal     = {Review of Economics and Statistics},
    volume      = {85},
    number      = {3},
    pages       = {531--549},
    year        = {2003}
}

@article{bai_ng05,
    author      = {Bai, J. and Ng, S.},
    title       = {Tests for skewness, kurtosis, and normality
               for time series data},
    journal     = {Journal of Business \& Economic Statistics},
    volume      = {23},
    number      = {1},
    pages       = {49--60},
    year        = {2005}
}

@article{barndorff-nielsen_etal82,
    author      = {Barndorff-Nielsen, O. and Kent, J. and S{\o}rensen, M.},
    title       = {Normal variance--mean mixtures and $z$ distributions},
    journal     = {International Statistical Review},
    volume      = {50},
    number      = {2},
    pages       = {145--159},
    year        = {1982}
}

@article{bollerslev86,
    author      = {Bollerslev, T.},
    title       = {Generalized autoregressive conditional heteroscedasticity},
    journal     = {Journal of Econometrics},
    volume      = {31},
    number      = {3},
    pages       = {307--327},
    year        = {1986}
}

@article{bollerslev87,
    author      = {Bollerslev, T.},
    title       = {A conditionally heteroskedastic time series model for speculative prices and rates of return},
    journal     = {Review of Economics and Statistics},
    volume      = {69},
    number      = {3},
    pages       = {542--547},
    year        = {1987}
}

@article{chernoff54,
    author      = {Chernoff, H.},
    title       = {On the distribution of the likelihood ratio},
    journal     = {Annals of Mathematical Statistics},
    volume      = {25},
    number      = {3},
    pages       = {573--578},
    year        = {1954}
}

@article{cont01,
    author      = {Cont, R.},
    title       = {Empirical properties of asset returns: stylized facts and statistical issues},
    journal     = {Quantitative Finance},
    volume      = {1},
    number      = {2},
    pages       = {223--236},
    year        = {2001}
}

@book{dasgupta08,
    author      = {{DasGupta}, A.},
    title       = {Asymptotic Theory of Statistics and Probability},
    publisher   = {Springer},
    address     = {New York},
    year        = {2008}
}

@article{donsker52,
    author   = {Donsker, M. D.},
    title    = {Justification and extension of {D}oob's heuristic approach to the {K}olmogorov--{S}mirnov theorems},
    journal  = {Annals of Mathematical Statistics},
    volume   = {23},
    number   = {2},
    pages    = {277--281},
    year     = {1952}
}

@article{doornik_hansen08,
    author      = {Doornik, J. A. and Hansen, H.},
    title       = {An omnibus test for univariate and multivariate normality},
    journal     = {Oxford Bulletin of Economics and Statistics},
    volume      = {70},
    number      = {S1},
    pages       = {927--939},
    year        = {2008}
}

@article{durbin73,
  author    = {Durbin, J.},
  title     = {Weak convergence of the sample distribution
             function when parameters are estimated},
  journal   = {Annals of Statistics},
  year      = {1973},
  volume    = {1},
  number    = {2},
  pages     = {279--290}
}

@article{engle82,
    author      = {Engle, R. F.},
    title       = {Autoregressive conditional heteroscedasticity with estimates of the variance of {United Kingdom} inflation},
    journal     = {Econometrica},
    volume      = {50},
    number      = {4},
    pages       = {987--1007},
    year     = {1982}
}

@article{fama65,
    author      = {Fama, E. F.},
    title       = {The behavior of stock-market prices},
    journal     = {Journal of Business},
    volume      = {38},
    number      = {1},
    pages       = {34--105},
    year        = {1965}
}

@article{forchini02,
    author      = {Forchini, G.},
    title       = {The exact cumulative distribution function of a ratio of quadratic forms in normal variables, with application to the {AR}(1) model},
    journal     = {Econometric Theory},
    volume      = {18},
    number      = {4},
    pages       = {823--852},
    year        = {2002}
}

@misc{fred_sp500,
    author      = {{Federal Reserve Bank of St.\ Louis}},
    title       = {{S\&P} 500},
    howpublished = {\url{https://fred.stlouisfed.org/series/SP500}},
    year      = {2026},
    note      = {Retrieved 2 June 2026}
}

@article{gabaix09,
    author      = {Gabaix, X.},
    title       = {Power laws in economics and finance},
    journal     = {Annual Review of Economics},
    volume      = {1},
    pages       = {255--293},
    year        = {2009}
}

@article{gallant_nychka87,
    author      = {Gallant, A. R. and Nychka, D. W.},
    title       = {Semi-nonparametric maximum likelihood estimation},
    journal     = {Econometrica},
    volume      = {55},
    number      = {2},
    pages       = {363--390},
    year        = {1987}
}

@article{gallant_tauchen89,
    author      = {Gallant, A. R. and Tauchen, G.},
    title       = {Seminonparametric estimation of conditionally constrained heterogeneous processes: asset pricing applications},
    journal     = {Econometrica},
    volume      = {57},
    number      = {5},
    pages       = {1091--1120},
    year        = {1989}
}

@techreport{gordy98,
    author      = {Gordy, M. B.},
    title       = {A generalization of generalized beta distributions},
    institution = {Federal Reserve Board},
    type        = {{Finance and Economics Discussion Series}},
    year        = {1998}
}

@article{hillier01,
    author      = {Hillier, G.},
    title       = {The density of a quadratic form in a vector uniformly distributed on the $n$-sphere},
    journal     = {Econometric Theory},
    volume      = {17},
    number      = {1},
    pages       = {1--28},
    year        = {2001}
}

@article{jarque_bera87,
    author      = {Jarque, C. M. and Bera, A. K.},
    title       = {A test for normality of observations
               and regression residuals},
    journal     = {International Statistical Review},
    volume      = {55},
    number      = {2},
    pages       = {163--172},
    year        = {1987}
}

@article{jondeau_rockinger01,
    author      = {Jondeau, E. and Rockinger, M.},
    title       = {Gram--{C}harlier densities},
    journal     = {Journal of Economic Dynamics and Control},
    volume      = {25},
    number      = {10},
    pages       = {1457--1483},
    year        = {2001}
}

@phdthesis{lawford01,
    author      = {Lawford, S.},
    title       = {Improved modelling in finite-sample and nonlinear frameworks},
    school      = {University of York},
    year        = {2001}
}

@article{mandelbrot63,
    author      = {Mandelbrot, B.},
    title       = {The variation of certain speculative prices},
    journal     = {Journal of Business},
    volume      = {36},
    number      = {4},
    pages       = {394--419},
    year        = {1963}
}

@book{mcneil_etal15,
    author      = {McNeil, A. J. and Frey, R. and Embrechts, P.},
    title       = {Quantitative Risk Management: Concepts, Techniques and Tools},
    publisher   = {Princeton University Press},
    address     = {Princeton NJ},
    year        = {2015}
}

@article{nelson91,
    author      = {Nelson, D. B.},
    title       = {Conditional heteroskedasticity in asset returns: a new approach},
    journal     = {Econometrica},
    volume      = {59},
    number      = {2},
    pages       = {347--370},
    year        = {1991}
}

@incollection{newey_mcfadden94,
    author      = {Newey, W. K. and McFadden, D. L.},
    title       = {Large sample estimation and hypothesis testing},
    booktitle   = {Handbook of Econometrics},
    editor      = {Engle, R. F. and McFadden, D. L.},
    volume      = {4},
    pages       = {2111--2245},
    publisher   = {Elsevier},
    address     = {Amsterdam},
    year        = {1994}
}

@book{olver_etal10,
    editor      = {Olver, F. W. J. and Daalhuis, A. B. O. and Lozier, D. W. and Schneider, B. I. and Boisvert, R. F. and Clark, C. W. and Miller, B. R. and Saunders, B. V. and Cohl, H. S. and {McClain}, M. A.},
    title       = {{NIST} Digital Library of Mathematical Functions},
    publisher   = {National Institute of Standards and Technology},
    year        = {2026},
    note        = {Release 1.2.6, \url{https://dlmf.nist.gov}}
}

@incollection{phillips83,
    author      = {Phillips, P. C. B.},
    title       = {Exact small sample theory in the simultaneous equations model},
    booktitle   = {Handbook of Econometrics},
    editor      = {Griliches, Z. and Intriligator, M. D.},
    volume      = {1},
    chapter     = {8},
    pages       = {449--516},
    publisher   = {North-Holland},
    address     = {Amsterdam},
    year        = {1983}
}

@article{royston92,
    author      = {Royston, P.},
    title       = {Approximating the {Shapiro--Wilk} {W-test} for non-normality},
    journal     = {Statistics and Computing},
    volume      = {2},
    pages       = {117-119},
    year        = {1992}
}

@article{self_liang87,
    author      = {Self, S. G. and Liang, K.-Y.},
    title       = {Asymptotic properties of maximum likelihood estimators and likelihood ratio tests under nonstandard conditions},
    journal     = {Journal of the American Statistical Association},
    volume      = {82},
    number      = {398},
    pages       = {605--610},
    year        = {1987}
}

@article{shapiro_wilk65,
    author      = {Shapiro, S. S. and Wilk, M. B.},
    title       = {An analysis of variance test for normality
               (complete samples)},
    journal     = {Biometrika},
    volume      = {52},
    number      = {3/4},
    pages       = {591--611},
    year        = {1965}
}

@book{shorack_wellner86,
    author      = {Shorack, G. R. and Wellner, J. A.},
    title       = {Empirical Processes with Applications to Statistics},
    publisher   = {Wiley},
    address     = {New York},
    year        = {1986}
}

@article{stephens74,
    author      = {Stephens, M. A.},
    title       = {{{EDF}} statistics for goodness of fit and some comparisons},
    journal     = {Journal of the American Statistical Association},
    volume      = {69},
    number      = {347},
    pages       = {730--737},
    year        = {1974}
}

@book{vandervaart98,
    author      = {van der Vaart, A. W.},
    title       = {Asymptotic Statistics},
    series      = {Cambridge Series in Statistical and                          Probabilistic Mathematics},
    publisher   = {Cambridge University Press},
    address     = {Cambridge},
    year        = {1998}
}

@book{vandervaart_wellner23,
    author      = {van der Vaart, A. W. and Wellner, J. A.},
    title       = {Weak Convergence and Empirical Processes},
    publisher   = {Springer},
    address     = {New York},
    year        = {2023}
}
